\documentclass[11pt]{article}

\usepackage[margin=1in]{geometry}
\usepackage[hyphens]{url}
\usepackage[colorlinks=true, allcolors=blue, breaklinks=true]{hyperref}
\usepackage{amsfonts,mathrsfs,mathpazo,eucal,xspace,graphicx,appendix,xcolor,array}
\usepackage{endnotes}
\usepackage{amssymb,latexsym}
\usepackage{enumitem}
\setlist{itemsep=0mm}
\usepackage{float}

\usepackage{amsmath, mathtools, physics}
\usepackage{complexity}

\newcommand\numberthis{\addtocounter{equation}{1}\tag{\theequation}}

\usepackage[linesnumbered,ruled,vlined]{algorithm2e}
\usepackage{verbatim}

\usepackage{amsthm}
\newtheorem{theorem}{Theorem}
\newtheorem*{theorem*}{Theorem}

\newtheorem{lemma}[theorem]{Lemma}
\newtheorem*{lemma*}{Lemma}

\newtheorem{corollary}[theorem]{Corollary}

\theoremstyle{definition}
\newtheorem{definition}[theorem]{Definition}

\newtheorem{assumption}[theorem]{Assumption}
\newtheorem{remark}[theorem]{Remark}
\newtheorem*{remark*}{Remark}

\newcommand{\projerror}{g(n)}
\newcommand{\depthname}{depth-$d$ }

\newcommand{\Proj}{P}

\newcommand{\all}{ALL}

\newcommand{\mnote}[1]{}
\newcommand{\nnote}[1]{}
\newcommand{\mnotetwo}[1]{}
\definecolor{darkbrown}{RGB}{101,67,33}

\title{\bfseries\Large
Quasi-polynomial time approximation of output probabilities of geometrically-local, shallow quantum circuits.
}
\author{%
	Nolan J. Coble\footnote{Authors are listed alphabetically.}
	\thanks{%
		\texttt{ncoble@terpmail.umd.edu}
	},
  Matthew Coudron%
  \thanks{%
     \texttt{mcoudron@umd.edu} - Corresponding Author
    }
}

\begin{document}
\sloppy

\maketitle
\begin{abstract}
   We present a classical algorithm that, for any 3D geometrically-local, polylogarithmic-depth quantum circuit $C$ acting on $n$ qubits, and any bit string $x \in \{0,1\}^n$, can compute the quantity $|\bra{x}C   \ket{0^{\otimes n}}|^2$ to within any inverse-polynomial additive error in quasi-polynomial time.  It is known that it is $\#P$-hard to compute this same quantity to within $2^{-n^2}$ additive error \cite{Mov19, KMM21}. The previous best known algorithm for this problem used $O(2^{n^{1/3}}\poly(1/\epsilon))$ time to compute probabilities to within additive error $\epsilon$ \cite{BGM19}. Notably, the \cite{BGM19} paper included an elegant polynomial time algorithm for this estimation task restricted to 2D circuits, which makes a novel use of 1D Matrix Product States (MPS) carefully tailored to the 2D geometry of the circuit in question.  Surprisingly, it is not clear that it is possible to extend this use of MPS to address the case of 3D circuits in polynomial time.  This raises a natural question as to whether the computational complexity of the 3D problem might be drastically higher than that of the 2D problem.  In this work we address this question by exhibiting a quasi-polynomial time algorithm for the 3D case.  In order to surpass the technical barriers encountered by previously known techniques we are forced to pursue a novel approach: instead of using MPS techniques, we construct a recursive sub-division of the given 3D circuit using carefully designed block-encodings.  To our knowledge this is the first use of the block-encoding technique in a purely classical algorithm.

   Our algorithm has a Divide-and-Conquer structure, demonstrating how to approximate the desired quantity via several instantiations of the same problem type, each involving 3D-local circuits on about half the number of qubits as the original.  This division step is then applied recursively, expressing the original quantity as a weighted combination of smaller and smaller 3D-local quantum circuits.  A central technical challenge is to control correlations arising from entanglement that may exist between the different circuit ``pieces" produced this way.  We believe that the division step, which makes use of block-encodings \cite{GSLW18,LC16,AG19}, together with an Inclusion-Exclusion argument to reduce error in each recursive approximation, may be of independent interest.
   
     \mnote{Put these sentences in the introduction?:  We believe that our algorithm extends naturally to any fixed dimension $D$ by induction on the dimension, but we focus on the 3D case, as the simplest unresolved case, for concreteness.  Furthermore, we show that, under a natural, polynomial-time-checkable condition on the circuit $C$, our algorithm runs in polynomial time. This highlights the possibility that the super-polynomial worst-case time bound on our algorithm might be due to limitations of our analysis.}
\end{abstract}

\pagebreak
\section{Introduction}

Many schemes for obtaining a quantum computational advantage with near-term quantum hardware are motivated by mathematical results proving the computational hardness of sampling from near-term quantum circuits.  In this work we consider quantum circuits which are geometrically local and have polylogarithmic circuit-depth.  It is known to be $\#P$-hard to compute output probabilities of $n$-qubit, geometrically-local, constant-depth quantum circuits to within $2^{-n^2}$ additive error \cite{Mov19}, a result which builds on an extensive line of research focusing on the hardness of sampling from quantum circuits \cite{AA11,BJS11, BMS16,NSC+17, BFNV19}.  It has even been shown, under several computational assumptions, that there is no classical polynomial time algorithm which, given a geometrically-local, constant-depth  quantum circuit, K, can produce samples whose distribution lies within a constant, in the $\ell_1$ distance, of the output distribution of K in the computational basis \cite{BVHSRE18}.

On the other hand, a series of works on the classical complexity of sampling from near-term quantum circuits, and related tasks, highlights the subtle nature of identifying an actual quantum advantage based on these tasks \cite{DHKLP18,HZN+20,NLPD+20}. 
These results frame the significance of the algorithm presented as Theorem 5 in \cite{BGM19}, which estimates output probabilites of 2D-local constant depth circuits to inverse polynomial additive error in polynomial time.  In fact, the original algorithm in \cite{BGM19}, actually estimates quantities of the form $\bra{0^{\otimes n}}C^{\dagger} \left (\otimes_{i=1}^n P_i
\right )  C\ket{0^{\otimes n}}$, where each $P_i \in \{X, Y, Z, I\}$ is a single-qubit Pauli observable operator.  However, it is straightforward to convert their algorithm to compute the quantity $\bra{0^{\otimes n}}C^{\dagger} \left (\otimes_{i=1}^n \ket{x_i}\bra{x_i} \right )  C\ket{0^{\otimes n}} =|\bra{x}C   \ket{0^{\otimes n}}|^2$, $x \in \{0,1\}^n$, instead. Theorem 5 of \cite{BGM19} constitutes a pertinent observation.  While it is hard to sample from constant-depth quantum circuits, it is still unresolved whether it is hard to estimate any property of such a circuit which could have been computed using a polynomial number of samples from the output of the quantum circuit itself.  In particular: A polynomial number of samples from a 2D-local, constant-depth quantum circuit only allows one to estimate output probabilites of that circuit to inverse polynomial additive error.  But, it is shown in Theorem 5 of \cite{BGM19} that this same task can be done in classical polynomial time!  One might ask: Is there a well-defined Decision problem which can be solved using only a polynomial number of samples from such a quantum circuit, together with classical post-processing, and yet cannot also be efficiently solved using classical computing alone?  This is unknown.

We note, at this point, some basic facts about the task of computing the quantity $|\bra{0^{\otimes n}}C   \ket{0^{\otimes n}}|^2$ which explain why we can focus on this task WLOG, and may motivate our interest in it:

\begin{itemize}
	
	\item If there is an algorithm to estimate the quantity $|\bra{0^{\otimes n}}C   \ket{0^{\otimes n}}|^2$, for any 3D-local \depthname quantum circuit $C$, then that algorithm can be used to estimate $|\bra{x}C   \ket{0^{\otimes n}}|^2$ for any $x \in \{0,1\}^n$.  The reason is that $|\bra{x}C   \ket{0^{\otimes n}}|^2 = |\bra{0^{\otimes n}}G  \ket{0^{\otimes n}}|^2$ where $G$ is taken to be the 3D-local circuit $G \equiv C \left ( \otimes_{i=1}^n X^{x_i}\right )$.  Here $X$ represents the single qubit Pauli operator $\sigma_X$. Note that $G$ is still a depth-$O(d)$ quantum circuit.

	\item  Any such algorithm can also estimate $|\bra{0^{\otimes n}}C Z^n C^{\dagger}  \ket{0^{\otimes n}}|^2$, which is the magnitude of the expected bias of the Parity of the output bits of $C$, when measured in the computational basis. This is true by virtue of the fact that $C Z^n C^{\dagger} $ is, itself, a 3D local, depth-$O(d)$ circuit.  So, this type of computational problem allows us to study the power of \depthname geometrically-local, quantum circuits combined with certain limited types of classical post-processing, like the Parity function. 
	
	\item The algorithm we present in this work can easily be modified to approximate marginal probabilities (e.g., the probability that $x_1 = 1$ for $x \in \{0,1\}^n$ sampled from the given circuit, etc).  Consequently, it is straightforward to use this algorithm to search for all  $x \in \{0,1\}^n$ which have probability at least $\delta$ in the output distribution of a given \depthname geometrically-local circuit $C$.  That is, searching for all of the ``$\delta$-heavy" strings of $C$.  When $\delta = 1/\poly(n)$ there can be at most $\poly(n)$ such strings and our algorithm can find them all in quasi-polynomial time.

\end{itemize}

The algorithm for 2D circuits presented in Theorem 5 of \cite{BGM19} makes a novel use of 1D Matrix Product States, carefully tailored to the 2D geometry of the circuit in question.  However, the authors of \cite{BGM19} point out that it is not clear that it is possible to extend this use of MPS to address the case of 3D circuits in polynomial time.  Instead they provide a sub-exponential time algorithm for the 3D case, which has time complexity $O(2^{n^{1/3}}\poly(1/\epsilon))$ for computing the desired quantity to within additive error $\epsilon$.  In this work we introduce a new set of techniques culminating in a divide-and-conquer algorithm which solves the 3D case in quasi-polynomial time.  

  Our algorithm has a divide-and-conquer structure with the goal being to divide the circuit $C$ into pieces, and reduce the original problem to a small number of new 3D-circuit problems involving circuits on only a fraction of the number of qubits as the original.  This division step requires the ability to construct Schmidt vectors of the state $C   \ket{0^{\otimes n}}$, across a given cut, via a \depthname geometrically-local quantum circuit, so that the new subproblems can be expressed as smaller instantiations of the original problem type.  We accomplish this through the use of block-encodings, a technique designed for quantum algorithms \cite{GSLW18,LC16,AG19, LMR14, KLL+17}, but used here as a subroutine of a classical simulation algorithm instead.  However, to date, we are only able to construct, as a block-encoding circuit, the \emph{leading} Schmidt vector across certain ``heavy" cuts. Due to this restriction we are forced to use a novel division step in our Divide-and-Conquer approach.  Instead of dividing about a single cut and constructing many of its Schmidt vectors as \depthname geometrically-local block-encodings, we must divide across many cuts and construct only their leading Schmidt vectors.  Interestingly, this process can still lead to low approximation error via an Inclusion-Exclusion style argument, as shown in Lemma \ref{clm:expansiontrick}.

These techniques culminate in a worst-case quasi-polynomial time algorithm for 3D circuits, which is our main result:

\begin{theorem}\label{lem:quasi-poly-run-time}
	There exists a classical algorithm which, for any 3D geometrically-local, depth-$d$  quantum circuit $C$ on $n$ qubits, can compute the scalar quantity  $|\bra{0^{\otimes n}}C   \ket{0^{\otimes n}}|^2$ to within $1/n^{\log(n)}$ additive error in time
	\begin{equation}\label{eq:quasi-poly-runtime-bound}
	T(n) =  2^{d^3 \polylog(n)}\mnote{ O(2^{\polylog(n)}n^{d^2})?}
	\end{equation}
	
	See Algorithm \ref{alg:quasi-poly-driver} in Section \ref{section:quasi-poly-time} for a precise definition of this classical algorithm.
\end{theorem}

\mnote{above is the new style of theorem statement.  Replace everywhere with this?  Just put a comment in one location with the full $\delta$ dependence from the original theorem statement.}

Note that our Theorem statement gives an inverse quasi-polynomial additive error approximation.  This is, therefore, asymptotically better than an inverse polynomial additive approximation for any polynomial.  There is a more explicit trade-off between runtime and approximation error given in Theorem \ref{thm:mainthmwdelta}, and in fact, Theorem \ref{lem:quasi-poly-run-time} follows from Theorem \ref{thm:mainthmwdelta} with $\delta = 1/n^{\log(n)}$,  but we use the above statement here for simplicity.  Note also, when the depth, $d$, is polylogarithmic our algorithm runs in quasi-polynomial time.

\section{Dividing the Cube:  Some Notation}

Given a 3D-local, \depthname circuit $C$, we wish to estimate the quantity $|\bra{0^{\otimes n}}C   \ket{0^{\otimes n}}|^2$.  To begin our divide-and-conquer approach we will divide the circuit $C$ in half via a cut through the center as shown in Figure \ref{fig:cutcube2}.  The width of the cut is dependent on $d$, and we will discuss how to select this width below. To begin with, we make the width large enough to have the non-empty sets $B$, $M$, and $F$ defined below.

\begin{figure}[H]
	\begin{center}
		\includegraphics{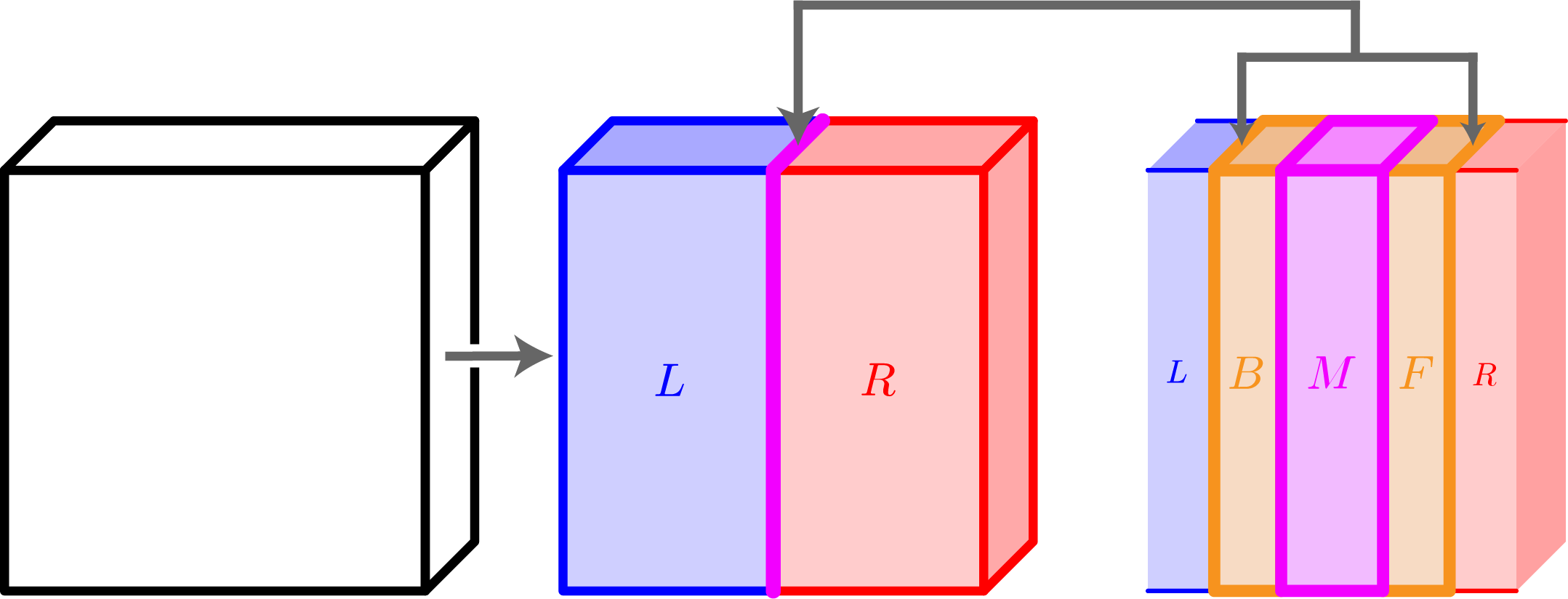}
		\caption{Cutting the Cube: (left) 3D cube of qubits. (center) Choose a location to cut the qubits. The qubits to the left and right of the cut are denoted by $L$ and $R$, respectively. (right) Within the cut there are three regions: a center region $M$ and regions to the left and right of $M$, denoted by $B$ and $F$, respectively.}
		\label{fig:cutcube2}
	\end{center}
\end{figure}

\begin{definition}[$M$, $B$, $F$, $R$, and $L$ (see Figure \ref{fig:cutcube2})] \label{def:bmf}
 Let $M$ be the set of all qubits in the ``Middle of the cut'' (the middle part of the cut which is not in the lightcone of qubits from outside the cut).  Let $B$ be the set of all qubits within the cut which are to the left of $M$. Let $F$ be the set of all qubits within the cut which are to the right of $M$.  We will choose the width of $M$ to be $O(d)$ such that the lightcones of $B$ and $F$ are disjoint.  We will choose the widths of $B$ and $F$ to be $O(d)$ such that the lightcone of $M$ is contained in $B \cup M \cup F$.  For concreteness we set the width of each of $B, M, F$ to be $10d$.  Since $C$ is geometrically-local, this is sufficiently large width to satisfy the above conditions on lightcones.

Let $L$ be all qubits outside the cut which are to the left of the cut (that is, to the left of $B$).  The set $L$ is colored blue.  Let $R$ be the set of all qubits outside to the right of the cut (that is, to the right of $F$).  The set $R$ is colored red.

\end{definition}

We will now define a 2D geometrically local, \depthname circuit $C_{B \cup M \cup F}$ which can be thought of as the sub-circuit of $C$ which lies within the light-cone of $M$.  Intuitively this circuit captures all of the local information that must be accounted for in the division step across this particular slice in our divide-and-conquer algorithm.  

\begin{definition}[$C_{B \cup M \cup F}$]\label{def:shortcircuit}
Now, let us begin with the all zeroes state on all the qubits $\ket{0}_{L \cup B \cup M \cup F \cup R} = \ket{0_{\all}}$, and apply the minimum number of gates from the circuit $C$ such that every gate on the qubits within $M$ has been applied.  We will call this unitary $C_{B \cup M \cup F}$.  Note that this unitary does not act on any qubits outside of $B \cup M \cup F$.  This is because the lightcone of $M$ is contained in $B \cup M \cup F$ by Definition \ref{def:bmf}.  Note that $C_{B \cup M \cup F}$ can be thought of as an approximately 2D (not 3D) geometrically-local, \depthname circuit, since the third dimension of the circuit is $O(d)$ which for $d=\polylog(n)$ grows asymptotically slower than $O(n^{1/3})$.

We define $C_{L\cup R}$ to be the unitary composed of the remainder of the gates of $C$ not yet applied in $C_{B \cup M \cup F}$, so that $C = C_{L\cup R} \circ C_{B \cup M \cup F}$.  We define $C_{L}$ (resp. $C_{R}$) to be the unitaries composed of the remainder of the gates of $C$ not yet applied in $C_{B \cup M \cup F}$ and which lie to the left (resp. right) of the $M$.  Note that $C_{L} \circ C_{R} = C_{L\cup R}$ since none of the circuits $C_{L}, C_{R}, C_{L\cup R}$ act non-trivially on $M$.  See Figure \ref{fig:cbmf} for an illustration of these unitaries with a 1D geometrically-local circuit, and Figure \ref{fig:3} for an illustration in a 2D circuit.
\end{definition}

\begin{figure}[H] 
	\begin{center}
		\includegraphics[scale=0.4]{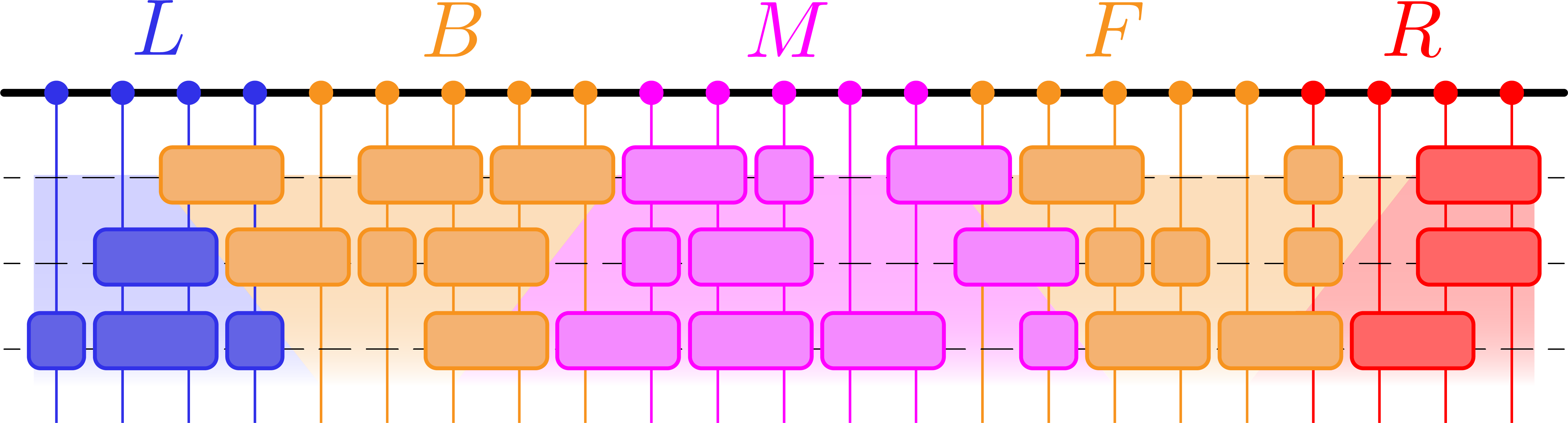}
		\caption{Block depiction of the unitaries defined in Definition \ref{def:shortcircuit} and Definition \ref{def:cutcircuit} for the case of a 1D geometrically-local constant-depth circuit $C$.   Here the vertical dimension represents the depth of the circuit, so the rectangle has the same dimensions as the circuit diagram would. $C_{B\cup M\cup F}$ is defined to be the unitary produced by the gates in the lightcone of $M$, which are colored {\color{magenta} magenta} in this diagram. The unitaries $C_{wrap}, C_{L-Wrap}, C_{R-Wrap}$, formally defined in Definition \ref{def:cutcircuit}, are also depicted here. $C_{wrap}$ is the unitary consisting of all the {\color{orange} orange} gates in the diagram, and $C_{L-Wrap}$  (resp.   $ C_{R-Wrap}$) is the unitary consisting of all the {\color{orange} orange} gates acting on the left (resp. right) of $M$.  Furthermore, $C_L'$ (resp. $C_R'$ ) denote the unitaries consisting of all the {\color{blue} blue} (resp. {\color{red} red}) gates to the left (resp. right) of $M$ in the diagram. We also illustrate the 2D case in Figure \ref{fig:3} below.  }  
		\label{fig:cbmf}
	\end{center}
\end{figure}

The sub-normalized quantum state produced by $C_{B \cup M \cup F}$, defined below, is the state whose Schmidt decomposition we consider in our division step.

\begin{definition}
	Let $\ket{\psi}_{B\cup F} \equiv \bra{0}_M C_{B \cup M \cup F}\ket{0}_{B \cup M \cup F}$. 
\end{definition}

Note that, $\bra{0}_{\text{ALL}} C_{L \cup R} \ket{0}_{L \cup R} \otimes \ket{\psi}_{B\cup F} =  \bra{0}_{\text{ALL}} C \ket{0}_{\text{ALL}}$.

(Throughout this document, the notation $\ket{0_{ALL}}$ will refer to the zero state on all unmeasured qubits for a given state.  It's meaning will be clear from context.)

\begin{definition}[$C_{Wrap}$ ] \label{def:cutcircuit}
 Define a new unitary $C_{Wrap}$ which consists of all the gates from $C$ which are in the reverse light-cone of $B \cup M \cup F$, but not in $C_{B \cup M \cup F}$ itself.  That is, let $C_{L-Wrap}$ (resp. $C_{R-Wrap}$) be the unitary consisting of all the of the gates in $C$ which are in the reverse light-cone of $B$ (resp. $F$), but not in $C_{B \cup M \cup F}$ itself, and let $C_{Wrap} \equiv C_{L-Wrap} \circ C_{R-Wrap}$.   Therefore, 
	\begin{align} \label{eq:cwrapexamp}
	C_{Wrap}^{\dagger} \circ C   = C'_{L} \circ  C_{B \cup M \cup F} \circ C'_{R}  
	\end{align}
	
	Where $C'_{L} \equiv C_{L-Wrap}^{\dagger}\circ C_{L} $ (see Definition  \ref{def:shortcircuit} for the definition of $C_{L}$) is a unitary acting only or $L$  (the remaining, untouched gates of $C$ within $L$), and $C'_{R} \equiv C_{R-Wrap}^{\dagger}\circ C_{R}$ (see Definition \ref{def:shortcircuit} for the definition of $C_{R}$) is a unitary acting only on $R$ (the remaining, untouched gates of $C$ within $R$).  Since $C$ is \depthname it is clear that every qubit in the non-trivial support of $C_{Wrap}$ lies within some $O(d)$ distance of $M$.  Let $R^{Wrap}$ (resp. $L^{Wrap}$) be the subset of qubits in $R$ (resp. $L$) that lie in the non-trivial support of $C_{Wrap}$.  In other words, $R^{Wrap}$ (resp. $L^{Wrap}$) is the non-trivial support of $C_{R-Wrap}$ (resp. $C_{L-Wrap}$). See Figure \ref{fig:cbmf} for an illustration of these unitaries with a 1D geometrically-local circuit, and Figure \ref{fig:3} for an illustration in a 2D circuit. 
\end{definition}

\begin{figure}[H] 
	\begin{center}
		\includegraphics[scale=0.23]{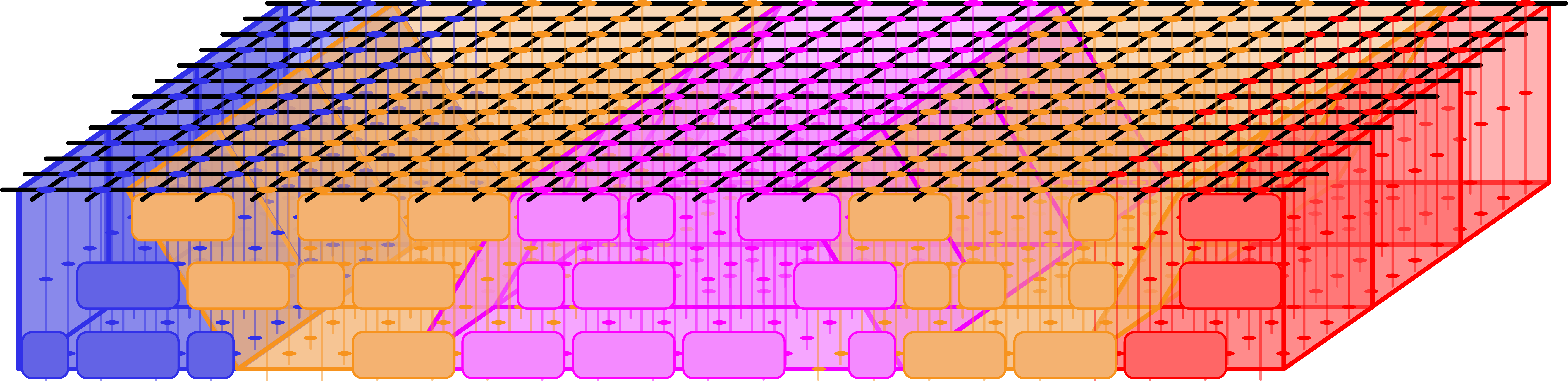}
		\caption{Geometric depiction of the unitaries defined in Definition \ref{def:shortcircuit} and Definition \ref{def:cutcircuit} for the case of a 2D grid of qubits.  Here the vertical dimension represents the depth of the circuit, so the rectangular prism has the same dimensions as the circuit diagram would.   We do not have an analogous figure for 3D circuits, which are the main focus of this work, because it would require 4 dimensions to illustrate. However, we believe the reader will gain sufficient intuition for the definitions from the 1D and 2D diagrams. }
		\label{fig:3}
	\end{center}
\end{figure}

\section{Divide and Conquer: Schmidt Vectors and Block Encodings} \label{sec:simpleexamp}

In this section we will show how to construct a geometrically-local, shallow quantum circuit for the largest Schmidt vector of the unnormalized state $\bra{0_M}C \ket{0_{\all}}$ across the cut $M$, in the case that the largest Schmidt coefficient is very large.  Let us begin, however, by outlining the intuition behind our divide-and-conquer approach, which explains why we are interested in approximating Schmidt vectors via shallow quantum circuits in the first place.  Consider expanding the quantity $\bra{0_{\all}} C \ket{0_{\all}}  = \bra{0}_{\text{ALL}} C_{L \cup R} \ket{0}_{L \cup R} \otimes \ket{\psi}_{B\cup F}$ as a sum over the Schmidt decomposition of $\ket{\psi}_{B\cup F}$ across the cut $M$.  Suppose, that $\ket{\psi}_{B\cup F}$ has almost all of its weight on the top polynomially many Schmidt vectors  (In Section \ref{sec:splitheavy} we will show that, in fact, we can restrict this part of the analysis WLOG to cases where $\ket{\psi}_{B\cup F}$ has a large fraction of its weight on $\lambda_1$).  Then $\ket{\psi}_{B\cup F} \approx \sum_{i=1}^{p(n)}\lambda_i \ket{v_i}_B\otimes \ket{w_i}_F$, and we have: 

\begin{align}
&\bra{0}_{\text{ALL}} C_{L \cup R} \ket{0}_{L \cup R} \otimes \ket{\psi}_{B\cup F} \approx \sum_{i=1}^{p(n)}\lambda_i  \bra{0}_{\text{ALL}} C_{L \cup R} \ket{0}_{L \cup R} \otimes \ket{v_i}_B\otimes \ket{w_i}_F  
\end{align}

\begin{align}
&=  \sum_{i=1}^{p(n)}\lambda_i  \bra{0}_{L \cup B} C_{L }  \ket{0}_{L } \otimes \ket{v_i}_B \cdot\bra{0}_{F \cup R} C_{ R} \ket{0}_{ R} \otimes \ket{w_i}_F, \label{eq:5}
\end{align}

where $\text{ALL} \equiv L \cup B \cup F \cup R$.

Suppose we could produce approximations for the Schmidt vectors $\ket{v_i}_B$ and $\ket{w_i}_F $ via 2D geometrically-local, shallow quantum circuits. Then, the quantity in Equation \ref{eq:5} would be a sum of polynomially many scalar quantities, each of which is the product of output probabilities of two new 3D geometrically-local circuit problems ($C_L$ and $C_R$).  Furthermore, these new 3D circuit problems involve about half the number of qubits as the original problem we were trying to solve.  This leads to a divide-and-conquer recursion which can yield a more efficient runtime for the original problem. The base case in this divide-and-conquer algorithm consists of estimating output probabilities of 3D-local, \depthname quantum circuits which have small width (width at most $w = \polylog(n)$) in one of their dimensions.  This base case can be solved efficiently using the algorithm from Theorem 5 of \cite{BGM19}, as discussed in Remark \ref{rm:3Dbasecase} below.

Note that this divide-and-conquer approach  only works if we can produce explicit approximations for the Schmidt vectors $\ket{v_i}_B$ and $\ket{w_i}_F $ via 2D geometrically-local, shallow quantum circuits. In the case when $\lambda_1$ is sufficiently large, it turns out that we can at least produce the top Schmidt vectors $\ket{v_1}_B$ and $\ket{w_1}_F$ in this way. (Note, we will also need to compute $\lambda_1$ efficiently, and this can also be done using Theorem 5 of \cite{BGM19}, as described in Definition \ref{def:kappa} and Remark \ref{rm:3Dbasecase}.)  However, we do not know how to construct 2D geometrically-local, shallow quantum circuits that approximate $\ket{v_i}_B$ and $\ket{w_i}_F$ for $i > 1$, and so we cannot pursue the divide-and-conquer approach described in Equation \ref{eq:5} verbatim.  Nonetheless, we will see in Section \ref{sec:splitheavy} that just approximating the top Schmidt vectors $\ket{v_1}_B$ and $\ket{w_1}_F$ is already sufficient to produce a (more involved) divide-and-conquer algorithm for the whole estimation problem.  The complete algorithm is explicitly written out in Section \ref{section:quasi-poly-time} (see Algorithms \ref{alg:quasi-poly-driver} and \ref{alg:quasi-poly-subroutine}). The key additional insight is to combine the intuition from Equation \ref{eq:5} above, with an additional expansion trick, expressed in Lemma \ref{clm:expansiontrick}.

\begin{remark}\label{rm:3Dbasecase}
	Theorem 5 of \cite{BGM19} shows that the output probabilities of 2D constant-depth circuits can be computed to inverse polynomial additive error in polynomial time.  Technically, this does not exactly cover the base case of our divide-and-conquer approach because our base case will consist of circuits which are 3D, but have a small width in the third dimension.  One might say that the base case circuits have a 2D structure with small ``thickness" in the third dimension. Fortunately, this extended case is also covered by additional analysis from the \cite{BGM19} paper, in which the authors show, on pages 25 and 26 (of the arXiv version), that a slightly modified version of their algorithm can, in fact, compute output probabilities of 3D-local, \depthname circuits to additive error $\epsilon$ in time $n\epsilon^{-2}2^{O(d^2\cdot w)}$, where $w$ is the width of the third dimension of the circuit. For convenience, throughout the remainder of this paper every reference to Theorem 5 of \cite{BGM19} will refer instead to this modified algorithm which can handle these ``small-width" 3D-local, \depthname circuits. Additionally, when we refer to 2D-local circuits we are including, within that definition, 3D-local circuits where the width in the third dimension is $w= \polylog(n)$. The reason that this is a reasonable use of terminology in the context of this paper is that Theorem 5 of \cite{BGM19}, and the subsequent discussion, can handle these small-width 3D-local circuits in time exponential in the size of $w$ (which, for $w=\polylog(n)$, is quasi-polynomial).\mnote{While Theorem 5 of \cite{BGM19} only explicitly shows that algorithm $\mathcal{B}$ can compute such quantities for circuits that are exactly 2D local (with no ``thickness" in the third dimension), it is straightforward to adapt their techniques to handle the cases where the circuit has constant thickness in the third dimension, while only increasing the runtime by a constant factor in the exponent.  This is done by simply increasing the bond dimension in the Matrix Product States used in Theorem 5 of \cite{BGM19} by a constant factor. }
\end{remark}

Our approach for explicitly constructing $\ket{w_1}_F$  is based on a tool called a ``block-encoding", which aims to generate a unitary whose top left corner contains the Hermitian matrix $\rho_F \equiv \tr_B(\ket{\psi}\bra{\psi}_{B\cup F})$, or the integer powers $\rho_F^K$ for $K = \polylog(n)$.   In fact, under an assumption that $\lambda_1$ is sufficiently large, $\frac{1}{\lambda_1^K}\rho_F^K$ is already very close to a projector onto $\ket{w_1}_F$ (see Lemma \ref{clm:schmidtproj} for the explicit scaling).

\begin{lemma}[Lemma 45 of \cite{GSLW18}]\label{prop:blockencoding}
	The following is a 2D-local (see Remark \ref{rm:3Dbasecase}), \depthname circuit which gives a block encoding for $\rho_F \equiv \tr_B(\ket{\psi}\bra{\psi}_{B\cup F})$:
	
	\[(C^{\dagger}_{B \cup M \cup F} \otimes I_{F'})(I_{B \cup M }\otimes \text{SWAP}_{F F'}) (C_{B \cup M \cup F} \otimes I_{F'})\]
\end{lemma}

\begin{proof}
	From Lemma 45 of \cite{GSLW18} it follows that the circuit $(C^{\dagger}_{B \cup M \cup F} \otimes I_{F'})(I_{B \cup M }\otimes \text{SWAP}_{F F'}) (C_{B \cup M \cup F} \otimes I_{F'})$ is a block-encoding of $\rho_F$.  Here $F'$ is a fresh register which is identical in size to $F$.  Note that $\text{SWAP}$ is not geometrically local a priori, but if we interleave the qubits of $F$ and $F'$ in the geometrically appropriate way, which we are free to do, then the $\text{SWAP}_{F F'}$ can be implemented in a geometrically local, depth-1 manner.  Thus the entire block-encoding is still given by a \depthname 2D-local circuit.  
	
	One additional subtlety:  We are neglecting to measure the $M$ register in the $\ket{0}$ basis here, but this is still a block-encoding for $\rho_F$ nonetheless.  The reason is that that measurement can be absorbed into the definition of block-encoding.  
\end{proof}

\smallskip

Following Lemma 53 of \cite{GSLW18}, we can now create a block encoding for the $K^{th}$ power of $\rho_F$ by creating $K$ distinct $F$ registers $F_1, ...., F_K$ (interwoven in the geometrically appropriate way just as in the proof of Lemma \ref{prop:blockencoding}), and multiplying $K$ different block encodings for $\rho_F$, each using a different one of the registers $F_i$, as so:

\begin{align} 
&\prod_{i=1}^K(C^{\dagger}_{B \cup M \cup F} \otimes I_{F_1,...,F_K})(I_{B \cup M }\otimes \text{SWAP}_{F F_i}) (C_{B \cup M \cup F} \otimes I_{F_1, ..., F_K}) 
\end{align}

We therefore have the following Lemma.

\begin{lemma}[Lemma 53 of \cite{GSLW18}]\label{lem:powerblockencoding}
	For any constant integer $K > 0$, the following is a 2D-local (see Remark \ref{rm:3Dbasecase}) quantum circuit which gives a block encoding for $\rho_F^K$, and has depth $O(dK^2)$:\mnote{changed lemma statement from O(K) to O($K^2$).  Should update in the proofs!}
	
\begin{align} \label{eq:powerblock}
&\prod_{i=1}^K(C^{\dagger}_{B \cup M \cup F} \otimes I_{F_1,...,F_K})(I_{B \cup M }\otimes \text{SWAP}_{F F_i}) (C_{B \cup M \cup F} \otimes I_{F_1, ..., F_K}) 
\end{align}
\end{lemma}

\begin{proof}
	The fact that Equation \ref{eq:powerblock} gives a block encoding of $\rho_F^K$ follows by repeated application of Lemma 53 of \cite{GSLW18}. The circuit in Equation \ref{eq:powerblock}, disregarding geometric locality, has depth $O(dK)$ because it is a composition of $3K$ circuits each having depth $O(d)$.  The circuit can be made 2D-local if we choose, WLOG, for the $F_j$ registers to be interleaved with the other qubits in a manner that matches the 2D geometry.  Since there are now $K$ different $F_j$ registers, this can increase the depth of our circuit by another factor of $K$ (adding nearest-neighbor SWAP gates to ensure that every gate is exactly 2D-local at each step).  So the depth of the geometrically-local version of the circuit is O($dK^2$).  
	
	For brevity we do not include a more explicit description of this process for interleaving registers in order to make the circuit geometrically local because we believe that the reader will understand this process from the above description.  We note that, in any case, the scaling of the depth of the circuit in Equation \ref{eq:powerblock} could be any polynomial in $K$ and our final result in Theorem \ref{lem:quasi-poly-run-time} would still hold.
\end{proof}

   Stated concretely, the fact that the circuit in Equation \eqref{eq:powerblock} is a block encoding for $\rho_F^K$ simply means that, if we define $\ket{0_{ancilla}} = \ket{0_{F_1, ...F_k, M, B}}$, then:

\begin{align}
\rho_F^K = \bra{0_{ancilla}} \prod_{i=1}^K(C^{\dagger}_{B \cup M \cup F} \otimes I_{F_1,...,F_K})(I_{B \cup M }\otimes \text{SWAP}_{F F_i}) (C_{B \cup M \cup F} \otimes I_{F_1, ..., F_K}) \ket{0_{ancilla}}. 
\end{align}

\mnote{  Reiterating the earlier discussion:  Under our assumption that $\lambda_1 \geq (1- O(\log(n)/n))$, $\rho_F^K$ is already very close to an (unnormalized) projector onto $\ket{w_1}_F$, even when $K = \log(n)$. If we wish to know the norm of $\rho_F^K$ so that we can normalize it, we can efficiently approximate the norm as $\lambda_1^K$.  Note that we can compute this norm to polynomial additive error using Theorem 5 of \cite{BGM19}.  This is because the circuit in Equation \eqref{eq:powerblock} is given by a 2D local (not 3D local) depth $O(K)$ quantum circuit  (When $K = \log(n)$ this might technically require quasi-polynomial time.  However, ultimately, we will be able to set $K=d$ where $d$ is such that we want a $\tilde{O}(1/n^d)$ additive approximation to $\bra{0_{\all}} C \ket{0_{\all}}$.  Therefore,  I will ignore this issue for the moment.  If necessary we will be willing to accept a quasi-polynomial time algorithm since that is still progress.  It also seems like it should be polynomial time though.) Since the pure state $\ket{w_1}$ is only well-defined up to a choice of phase, it makes sense that we are only able to produce the outer product $\ket{w_1}\bra{w_1}$ this way rather than the pure state $\ket{w_1}$ itself (the choice of phase on the pure state is arbitrary).In this note, the above argument will be used as the proof of Lemma \ref{clm:schmidtproj} in Section \ref{sec:splitheavy}.}

\section{Divide and Conquer: Splitting Over Heavy Slices} \label{sec:splitheavy}

In this section we will prove a set of results which will allow us to precisely define and analyze the division step in our divide-and-conquer algorithm.  The process begins by identifying slices of the \depthname circuit $C$ which are appropriate division points.  Those are the slices which have ``heavy weight" as defined below.

Consider a set of $O(d)$-width 2D slices $K = \{K_i\}$ of the qubits of $C$, where each slice $K_i$ is parallel to the cut $B\cup M \cup F$ shown in Figure \ref{fig:cutcube2}, and is made up of three analogous sections $B_i, M_i, F_i$ (see Figure \ref{fig:2D-slices}). Let the the slices in $K$ be evenly spaced at an $O(d)$ distance apart, where this value is chosen to be large enough that the light cones of $K_i$ and $K_j$ are disjoint when $i \neq j$.  For concreteness we will say that the distance between slices $K_i$ is equal to $10d$.  We will also set the width of each of the sections $B_i, M_i, F_i$ to be $10d$, just as discussed in Definition \ref{def:bmf}.  This ensures that the properties stipulated by Definition \ref{def:bmf} are satisfied by  $B_i, M_i, F_i$.  

\begin{figure}[H] \label{fig:2D-slices}
	\begin{center}
		\includegraphics[scale=0.75]{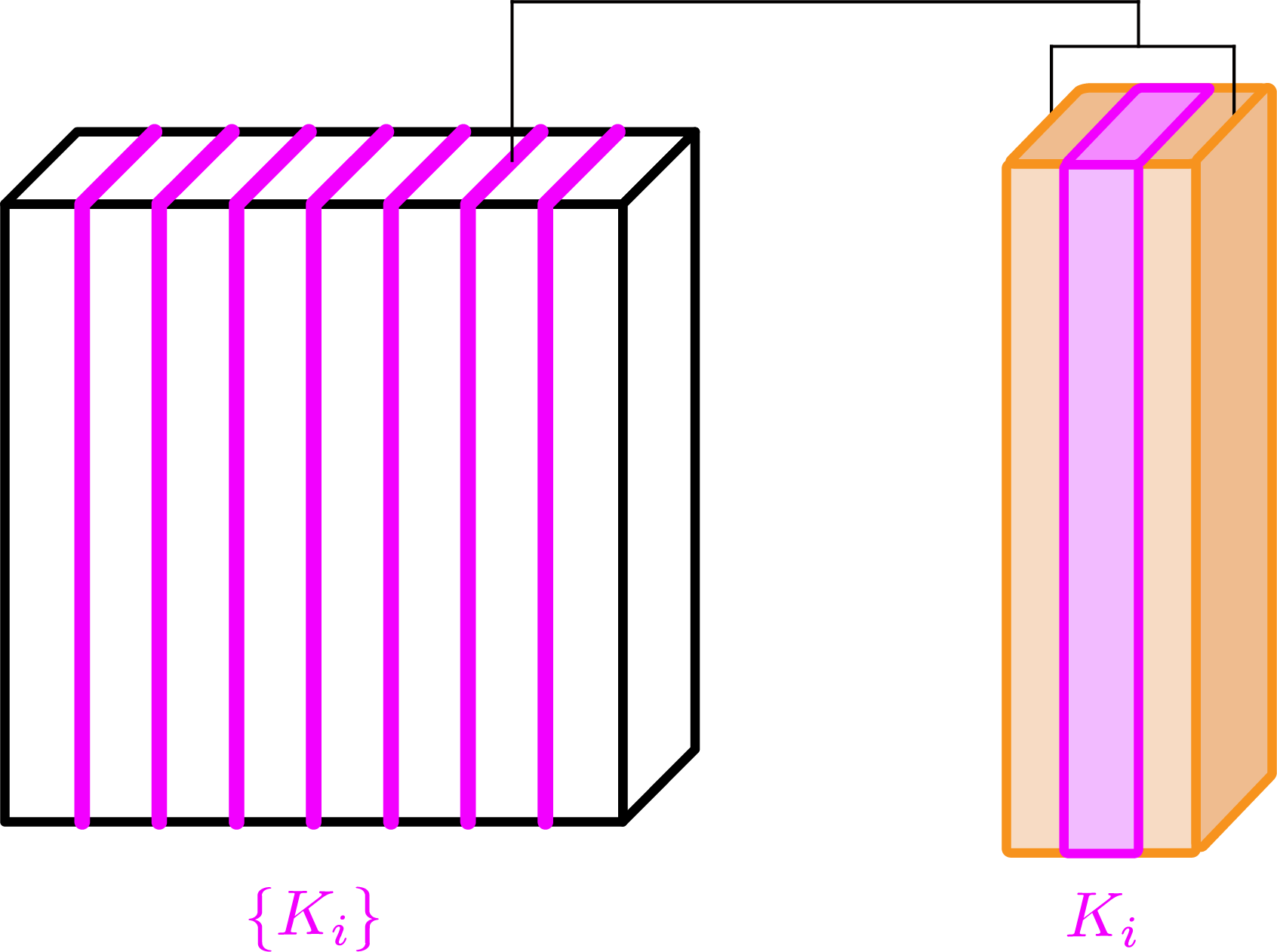}
		\caption{Set of slices $\{K_i\}$ \label{}}
	\end{center}
\end{figure}

\begin{definition}
	Let $I[M_i = 0]$ be the indicator random variable for the event that all of the qubits in $M_i$ collapse to $0$ when measured in the computational basis. Here joint probabilities are defined according the probability distribution $p_{total}$ produced by measuring $C  \ket{0^{\otimes n}}$ in the computational basis.  Let $p_{M_i = 0} := \mathbb{E}_{p_{total}}[I[M_i = 0]]$ be the probability that all of the bits in $K_i$ evaluate to $0$ according to the distribution $p_{total}$.
\end{definition}

\begin{lemma}\label{clm:independence}
	The $I[M_i = 0]$ are independent random variables.   Therefore, 
	\[p_{total}(M_i = 0 \text{ } \forall i) = \prod_i p_{total}(M_i = 0).  \]
\end{lemma}

\begin{proof}
The $I[M_i = 0]$ are independent random variables because the cuts $K_i$ are light-cone separated by definition.  The desired results follows. 
\end{proof}

Note that the variables $I[M_i = 0]$ may well be \emph{conditionally} dependent when conditioned on the outcomes of measuring the qubits in between the $K_i$ slices.  Indeed, that's what makes the global problem non-trivial in the first place.  But, when measuring the $K_i$ slices alone we see that the $I[M_i = 0]$ are independent as stated in Lemma \ref{clm:independence}.  

\
\begin{lemma}  \label{clm:findheavy}
If $|\bra{0^{\otimes n}}C   \ket{0^{\otimes n}}| > |1/q(n)|$, then, for any $0 \leq h \leq 1$, $h|K|$ of the slices $K_i$ in $K$ have the property that:

\begin{equation}
p_{total}(M_i = 0) \geq (|1/q(n)|)^{\frac{1}{(1-h)|K|}} \label{eq:heavyslice}
\end{equation}	

  We will let $K_{heavy}$ be the subset of $K$ consisting of those $K_i$ satisfying Equation \eqref{eq:heavyslice}.
\end{lemma}

\begin{proof}
	The proof of Lemma \ref{clm:findheavy} is given in Appendix \ref{appendix:lemma-proofs}.  
\end{proof}

%
\begin{definition}\label{def:slice}
	We define any particular slice $K_i$ as $K_i=B_i\cup M_i\cup F_i$., where $B_i, M_i, F_i$ are the analogous regions to $B,M,F$ (respectively) in Figure \ref{fig:cutcube2}. These slices are depicted in Figure \ref{fig:2D-slices}. Let $\ket{\psi}_{B_i \cup F_i}$   be analogous to $\ket{\psi}_{B\cup F} \equiv \bra{0}_M C_{B \cup M \cup F}\ket{0}_{B \cup M \cup F}$. Let $L_i$ and $R_i$ be sets of qubits analogous to the sets $L$ and $R$.  We define the unitaries $C_{B_i\cup M_i\cup F_i}$, $C_{wrap_i}$, $C'_{L_i} \equiv C_{{L-Wrap}_i}^{\dagger}\circ C_{L_i}$, and $C'_{R_i} \equiv C_{{R-Wrap}_i}^{\dagger}\circ C_{R_i}$ exactly as given in Definition \ref{def:cutcircuit} for the case of a single cut.
\end{definition}

	\begin{lemma} \label{lem:highschmidtnew}
		For any slice $K_i \in K_{heavy}$ satisfying:
		
				\begin{equation}\label{eq:exactweight}
				p_{total}(M_i = 0) \geq 1 - e(n),
				\end{equation}	
		 
		the top Schmidt coefficient of $\ket{\psi}_{B_i \cup F_i}$ satisfies $\lambda_1^i \geq 1 - O(e(n))$.  (Where the Schmidt decomposition is taken across the partition $B_i, F_i$.)  \mnote{Consider replacing e(n) with just some number x that doesn't need to scale with n.  Find explicit constant.  This might be less confusing.}
	\end{lemma}
	
	\begin{proof}
		The proof of Lemma \ref{lem:highschmidtnew} is given in Appendix \ref{appendix:lemma-proofs}.
	\end{proof}

The proof of Lemma \ref{lem:highschmidtnew} (See Appendix \ref{appendix:lemma-proofs}) suggests a way to perform a division step, dividing the original computational problem into the product of two new problems, but at the cost of an additive error that scales like $\Theta(e(n))$.  But, note that, if we want, say, $1/n^d$ additive error for $d \geq 2$, then this additive error term is way too large (in some cases $e(n)$ scales like $1/\log(n)$).  This means that, a priori, we cannot even afford to make use of Lemma \ref{lem:highschmidtnew} one single time!  However, Lemma \ref{clm:breaktoproduct} below shows how we can use this type of division step to divide the circuit at $\Delta$ different, light-cone separated cuts, $K_i$, simultaneously, and thereby achieve additive error that scales like $e(n)^{\Delta}$.

\begin{definition}\label{def:schmidtproj}
For any $K_i$ define the following two operators inspired by the block-encoding approach in Section \ref{sec:simpleexamp}:
	
	\begin{align} 
	&\Proj_{F_i}^K \equiv \nonumber  \frac{1}{\lambda_1^K}\bra{0^{B_i, M_i, F^1_i,...F^k_i}}\prod_{j=1}^K(C^{\dagger}_{B_i \cup M_i \cup F_i} \otimes I_{F^1_i,...,F^K_i})(I_{B_i \cup M_i }\otimes \text{SWAP}_{F_i F^j_i}) (C_{B_i \cup M_i \cup F_i} \otimes I_{F^1_i, ..., F^K_i})\ket{0^{B_i, M_i, F^1_i,...F^k_i}} \\
	&\text{ and } \\
	&	\Proj_{B_i}^K \equiv \nonumber  \frac{1}{\lambda_1^K}\bra{0^{F_i, M_i, B_i^1,...B_i^K}}\prod_{j=1}^K(C^{\dagger}_{B_i \cup M_i \cup F_i} \otimes I_{B_i^1,...,B_i^K})(I_{F_i \cup M_i }\otimes \text{SWAP}_{B_i B_i^j}) (C_{B_i \cup M_i \cup F_i} \otimes I_{B_i^1, ..., B_i^K})\ket{0^{F_i, M_i, B_i^1,...B_i^K}} 
	\end{align}  
	
	Here the first equation gives a linear operator on $F_i$, and the second equation gives a linear operator on $B_i$.  The registers $F_i^j$ (resp. $B_i^j$) are dummy registers that are used to create $K$ block encodings of the density matrix of the $F_i$ (resp. $B_i$) \mnote{is that the correct order, or flipped?} register of the state $C_{B_i \cup M_i \cup F_i}\ket{0^{F_i, M_i, B_i}}$.  These $K$ block encodings are then composed (multiplied) with each other in such a manner that they produce the block encoding of the $K^{th}$ power of the density matrix, as described in Section \ref{sec:simpleexamp}.
	
\end{definition}

\begin{lemma} \label{clm:schmidtproj}
	For any $K_i \in K_{heavy}$,
	
	\begin{align}
		&\|  \Proj_{F_i}^K -\ket{w_1}\bra{w_1}_{F_i}  \|_1 \leq \projerror \label{eq:Fprojclose1} \\
		& \text{and}  \nonumber\\
		&\|  \Proj_{B_i}^K -\ket{v_1}\bra{v_1}_{B_i}  \|_1 \leq \projerror \label{eq:Bprojclose1}
	\end{align}
	
	where $\projerror \equiv \left (\frac{1 - \lambda_1^i}{\lambda_1^i} \right )^K$, and $\ket{w_1}\bra{w_1}_{F_i}$, $\ket{v_1}\bra{v_1}_{B_i} $ are the projectors onto the top Schmidt vectors of $\ket{\psi}_{B_i \cup F_i}$ in $F_i$ and $B_i$ respectively.
\end{lemma}

\begin{proof}
	The proof of Lemma \ref{clm:schmidtproj} is given in Appendix \ref{appendix:lemma-proofs}.  
\end{proof}
\begin{definition} \label{def:wrappedpi}
	Define $\Pi^K_{F_i} \equiv C_{Wrap_i} \Proj^K_{F_i} C_{Wrap_i}^{\dagger}$.  
\end{definition}

Note that the operator $\Pi^K_{F_i}$ is in tensor product with $\ket{0_{M_i}}$ (it acts as the identity on the $M_i$ register since $C_{Wrap_i}$ and  $\Proj^K_{F_i}$ act trivially on that register).

\begin{definition}\label{def:psi-states}
	Let $\sigma\in\mathcal{P}[\Delta]\setminus\emptyset$ where $[\Delta]=\{1,\dots,\Delta\}$. Define the unormalized states $$\ket{\Psi_\sigma}=\otimes_{j \in \sigma }\Pi^K_{F_j} \otimes_{i\in [\Delta]} \bra{0_{M_i}}C\ket{0_{\all}}$$
And, $$\ket{\Psi_\emptyset}=\otimes_{i\in [\Delta]} \bra{0_{M_i}}C\ket{0_{\all}}$$
\end{definition}

\begin{lemma} \label{clm:expansiontrick}
	 Consider a set $K_{heavy}$ of slices such that, for every $K_i \in K_{heavy}$, $\ket{\psi}_{B_i \cup F_i}$ satisfies $\lambda_1^i \geq 1 - e(n)$, and such that for any $K_i, K_j \in K_{heavy}$, the operators $\Pi^K_{F_i}$ and $\Pi^K_{F_j}$ are light-cone separated whenever $i \neq j$.  Then, for any set of $\Delta$ slices, $\{K_i\}_{i \in [\Delta]} \subseteq K_{heavy}$, we have that:
	
	\begin{align}
	&\left \| \sum_{\sigma\in\mathcal{P}([\Delta])} (-1)^{\abs{\sigma}} \ket{\Psi_\sigma}\bra{\Psi_\sigma}\right \| = \left \|	\ket{\Psi_\emptyset}\bra{\Psi_\emptyset} -  \sum_{\sigma\in\mathcal{P}([\Delta])\setminus\emptyset} (-1)^{\abs{\sigma}+1} \ket{\Psi_\sigma}\bra{\Psi_\sigma}\right \| \leq (2e(n)+2\projerror)^{\Delta}, 
	\end{align}
	where $\projerror \equiv \left (\frac{1 - \lambda_1^i}{\lambda_1^i} \right )^K$.

\end{lemma}

\begin{proof}
	The proof of Lemma \ref{clm:expansiontrick} is given in Appendix \ref{appendix:lemma-proofs}.  
\end{proof}

\paragraph{Intuition for the statements of Lemmas \ref{clm:expansiontrick} and \ref{clm:breaktoproduct}:}  We will show, in Lemma \ref{clm:breaktoproduct} below, that each of the states $\ket{\Psi_\sigma}\bra{\Psi_\sigma}$, with $\sigma \neq \emptyset$, is very close to a product state about at least one of the $\Delta$ slices.  Thus, Lemma \ref{clm:expansiontrick} gives us a way to approximate $\ket{\Psi_\emptyset}\bra{\Psi_\emptyset}$ (which is the original state of interest) by a linear combination of product states $\ket{\Psi_\sigma}\bra{\Psi_\sigma}$.   Lemma \ref{clm:breaktoproduct} and Definition \ref{def:schmidtproj} then provide us with a way of constructing the corresponding product states using low-depth quantum circuits acting on approximately half as many qubits as the original circuit (this process is also further formalized in Definition \ref{def:subsynth}).  This combined use of Lemmas \ref{clm:expansiontrick} and \ref{clm:breaktoproduct} forms the backbone of our divide-and-conquer approach.

   In order to state Lemma \ref{clm:breaktoproduct} we now define three new states that are dependent on a particular choice of $K_i$.

\begin{definition} Given a shallow, 3D geometrically local quantum circuit $C$, and given a slice $K_i$ of $C$, define the states:
\begin{align}
  \ket{\Omega_i} &= \Pi^K_{F_i}\bra{0_{M_i}}C\ket{0_{\all}} \\
  \ket{\Xi_{L_i}}&= \Proj^K_{F_i} \bra{0_{M_i}}C_{L_i}C_{B_i \cup M_i \cup F_i}  \ket{0_{L_i\cup B_i \cup M_i \cup F_i}}\\
  \ket{\Xi_{R_i}}&= \Proj^K_{B_i} \bra{0_{M_i}}C_{R_i}C_{B_i \cup M_i \cup F_i}  \ket{0_{R_i\cup B_i \cup M_i \cup F_i}}
\end{align}

\end{definition}
  
   At this point it is pertinent to state Lemma \ref{clm:breaktoproduct}:

\begin{lemma}\label{clm:breaktoproduct}
	For any $K_i \in K_{heavy}$ (recall this means that $\ket{\psi}_{B_i \cup F_i}$ satisfies $\lambda_1^i \geq 1 - e(n)$), the state $\ket{\Omega_i}\bra{\Omega_i}$ is within $6\projerror$ of an unnormalized product state about $M_i$, described as follows:
	\begin{align}
	&\left \| \ket{\Omega_i}\bra{\Omega_i}  -  1/\lambda_1^i \tr_{F_i}\left (  \ket{\Xi_{L_i}}\bra{\Xi_{L_i}} \right )  \otimes  \tr_{B_i}\left ( \ket{\Xi_{R_i}}\bra{\Xi_{R_i}} \right )\right \|\leq 6\projerror
	\end{align}

	Here $\projerror \equiv \left (\frac{1 - \lambda_1^i}{\lambda_1^i} \right )^K \leq \left (\frac{e(n)}{1-e(n)} \right )^K$ just as in Lemma \ref{clm:schmidtproj}.  
	
\end{lemma}

 \begin{proof}
 	The proof of Lemma \ref{clm:breaktoproduct} is given in Appendix \ref{appendix:lemma-proofs}.  
 \end{proof}

 \begin{definition}[Synthesis] \label{def:synth}
 	We say that an unnormalized quantum state $\phi$ is \emph{synthesized} by a quantum circuit $\Gamma$, if $\Gamma$ has three registers of qubits $L, M, N$ such that:
 	
 	\begin{align}
 	\phi = \phi_{(\Gamma, L, M, N)}= \tr_{L\cup M}(  \bra{0_M}\Gamma\ket{0_{L \cup M \cup N}}\bra{0_{L \cup M \cup N}}\Gamma^{\dagger}\ket{0_M}).
 	\end{align}\mnote{still need to propogate this change throughout these definitions!!!}

 	In this case we say that the circuit $\Gamma$ together with a specification of the registers $L,M,N$ constitutes a \emph{synthesis} of $\phi$.  When $\phi$ is implicit we will call this collection $(\Gamma, L, M, N)$ a \emph{synthesis}.
 	
 	When $\Gamma$ is a 3D geometrically-local, \depthname circuit, and the register $N$ is one contiguous cubic subset of the qubits that $\Gamma$ acts on, with $L$, and $M$ only containing qubits on the ``edges", we call $(\Gamma, L, M, N)$ a 3D geometrically-local, \depthname \emph{synthesis}.
 \end{definition}

   \begin{definition}\label{def:middlegamma}[The Circuits $\Gamma_{i,j}$, $\Gamma_{L, i}$, $\Gamma_{R,j}$]
   	Recall, from Definition \ref{def:shortcircuit}, that $\Gamma_{B_k \cup M_k \cup F_k}$ ($k \in \{i,j\}$) is defined to be the circuit containing the minimal number of gates of $\Gamma$ such that every gate acting on $M_k$ is included.  Taking this definition for both $k = i$ and $k= j$, we now define $\Gamma_{i,j}$ to be a sub-circuit of $\Gamma$ consisting of the minimal number of gates of $\Gamma$ such that $\Gamma_{i,j} \circ \Gamma_{B_i \cup M_i \cup F_i} \circ \Gamma_{B_j \cup M_j \cup F_j}$ contains all of the gates of $\Gamma$ that lie between $M_i$ and $M_j$. Similarly define $\Gamma_{L_i}$ (resp. $\Gamma_{R_j}$) to be te sub-circuit of $\Gamma$ consisting of the minimal number of gates of $\Gamma$ such that $\Gamma_{L_i} \circ \Gamma_{B_i \cup M_i \cup F_i}$  (resp. $\Gamma_{R_j} \circ \Gamma_{B_j \cup M_j \cup F_j}$) contains all of the gates of $\Gamma$ that lie between $M_i$ (resp. $M_j$ )and the left-hand side (resp. right-hand side) of the Cube\mnote{of $\Gamma$}. 
   \end{definition}
 
 \mnote{It might be worthwhile to try to include a diagram depicting what $\Gamma_{i,j}$ looks like:  the circuit bridging between two different cuts}
 
 \begin{definition}  \label{def:subsynth}
 	Let $S=(\Gamma, G,H, N)$ be a 3D local, \depthname synthesis, and let $K_i$, $K_j$ ($i < j$) be two slices on the register $N$, as described in Definition \ref{def:slice}.  Let $M_i, F_i, B_i$, and $M_j, F_j, B_j$ be the subregisters of slices $K_i$ and $K_j$ respectively, as defined in Definition \ref{def:slice}.  Recall, from Definition \ref{def:synth}, that the state synthesized by $S$ is:
 	
\begin{align}
\phi_S = \tr_{G\cup H}(  \bra{0_H}\Gamma\ket{0_{G\cup H \cup N}}\bra{0_{G\cup H \cup N}}\Gamma^{\dagger}\ket{0_H}). \nonumber
\end{align}
 	
 	  We define three new pure states as follows:
 	  \begin{align}
 	  &     \ket{\varphi_{\L,i}}  = (\lambda_1^i)^K\Proj^K_{F_i} \bra{0_{M_i, H}}\Gamma_{L_i}  \Gamma_{B_i \cup M_i \cup F_i}  \ket{0_{L_i\cup B_i \cup M_i \cup F_i \cup G \cup H}} \nonumber  \\
 	  & \ket{\varphi_{j, \R}}  =   (\lambda_1^j)^K\Proj^K_{B_j} \bra{0_{M_j, H}}\Gamma_{R_j} \Gamma_{B_j \cup M_j \cup F_j}  \ket{0_{R_j\cup B_j \cup M_j \cup F_j \cup G \cup H}}\nonumber  \\
 	  &  \ket{\varphi_{i,j}}  =   (\lambda_1^i \lambda_1^j)^K \Proj^K_{B_i} \circ \Proj^K_{F_j}  \bra{0_{M_i, M_j, H}} \Gamma_{i,j} \circ \Gamma_{B_i \cup M_i \cup F_i} \circ \Gamma_{B_j \cup M_j \cup F_j} \ket{0_{ N_{i,j} \cup B_i \cup M_i \cup F_i \cup B_j \cup M_j \cup F_j \cup G \cup H}} \label{def:phiij}
 	  \end{align}

 	  Here $\Proj^K_{F_i}, \Proj^K_{B_j}$ are defined as in Definition \ref{def:schmidtproj}.  In the above the notation $N_{i,j}$ is defined to be the sub-register of $N$ containing all of the qubits between $F_i$ and $B_j$.

 	  From these, we define three new synthesized states (with corresponding syntheses) as follows:

 	  \begin{align}  \label{eq:def-phiij}
 	  &     \phi_{\L,i}  =  \tr_{F_i\cup M_i\cup G \cup H}\left ( \ket{\varphi_{\L,i}}\bra{\varphi_{\L,i}} \right ) \nonumber  \\
 	  & \phi_{j, \R}  =  \tr_{B_j\cup M_j \cup G \cup H}\left ( \ket{\varphi_{j, \R}} \bra{\varphi_{j, \R}}\right )\nonumber  \\
 	  &  \phi_{i,j}  =  \tr_{B_i\cup M_i \cup M_j \cup F_j \cup G \cup H}\left ( \ket{\varphi_{i,j}}\bra{\varphi_{i,j}} \right ) 
 	  \end{align}

 	  We can now write out the explicit synthesis for each of these synthesized states as follows:
 	  
 	  Recalling, from Definition \ref{def:schmidtproj} that,

 	  	\begin{align} \label{eq:repeatdefPF}
 	  	&\Proj_{F_i}^K \equiv \nonumber  \frac{1}{(\lambda_1^i)^K}\bra{0^{B_i, M_i, F^1_i,...F^K_i}}\prod_{j=1}^K(\Gamma^{\dagger}_{B_i \cup M_i \cup F_i} \otimes I_{F^1_i,...,F^K_i})(I_{B_i \cup M_i }\otimes \text{SWAP}_{F_i F^j_i}) (\Gamma_{B_i \cup M_i \cup F_i} \otimes I_{F^1_i, ..., F^K_i})\ket{0^{B_i, M_i, F^1_i,...F^K_i}} \\
 	  	\end{align} 
 	  
 	  We have that the explicit synthesis corresponding to $\phi_{\L,i} $ is:
 	  
 	  \begin{align}
 	  & S_{\L, i} \equiv & \left (\Gamma_{\Proj_{F_i}^K} \circ \Gamma_{L_i} \circ \Gamma_{B_i \cup M_i \cup F_i}, (F_i  \cup G), (M_i \cup M_i' \cup B_i' \cup F^1_i,... \cup F^K_i \cup H) \right . \nonumber\\
 	  & &\left . ,(L_i \cup B_i \cup M_i \cup F_i \cup M_i' \cup B_i' \cup G \cup H\cup F^1_i,... \cup F^K_i) \right),\nonumber
 	  \end{align}
 	  
 	  where $\Gamma_{\Proj_{F_i}^K}$ is defined as 
 	  
 	  \begin{align} \label{eq:gammaPF}
 	  \Gamma_{\Proj_{F_i}^K} \equiv \prod_{j=1}^K(\Gamma^{\dagger}_{B_i' \cup M_i' \cup F_i} \otimes I_{F^1_i,...,F^K_i})(I_{B_i' \cup M_i' }\otimes \text{SWAP}_{F_i F^j_i}) (\Gamma_{B_i' \cup M_i' \cup F_i} \otimes I_{F^1_i, ..., F^K_i}),
 	  \end{align}
 	  
 	  where $\Gamma_{B_i' \cup M_i' \cup F_i}$ is the same as $\Gamma_{B_i \cup M_i \cup F_i}$ except that it does not act on registers $B_i$ or $M_i$ at all, but instead, acts on dummy registers $B_i'$ and $M_i'$ in their place.

Symmetrically, for the explicit synthesis for $\ket{\phi_{j, \R} }$, recall that:

\begin{align} \label{eq:repeatdefPB}
\Proj_{B_i}^K \equiv   \frac{1}{(\lambda_1^i)^K}\bra{0^{F_i, M_i, B_i^1,...B_i^K}}\prod_{j=1}^K(C^{\dagger}_{B_i \cup M_i \cup F_i} \otimes I_{B_i^1,...,B_i^K})(I_{F_i \cup M_i }\otimes \text{SWAP}_{B_i B_i^j}) (C_{B_i \cup M_i \cup F_i} \otimes I_{B_i^1, ..., B_i^K})\ket{0^{F_i, M_i, B_i^1,...B_i^K}}, 
\end{align} 	 

and, therefore, we have that the explicit synthesis corresponding to $\ket{\phi_{j, \R} }$ is:

\begin{align}
& S_{j, \R} \equiv & \left (\Gamma_{\Proj_{B_j}^K} \circ \Gamma_{R_j} \circ \Gamma_{B_j \cup M_j \cup F_j}, (B_j  \cup G), (M_j \cup M_j' \cup F_j' \cup B^1_j,... \cup B^K_j \cup H) \right . \nonumber\\
& &\left . ,(R_j \cup B_j \cup M_j \cup F_j \cup M_j' \cup F_j' \cup G \cup H\cup B^1_j,... \cup B^K_j) \right), \nonumber
\end{align}

\mnote{does the last piece of the above synthesis, the one corresponding to $N$ in a generic synthesis, need to include all the registers total, or just the ones that have been mentioned in the $L$ and $M$ parts? Technically I think it might not matter, but would be clearer if you only include what is necessary. If this needs to be corrected here it also needs to be corrected in  the synthesis for  $\ket{\phi_{\L,i} }$ above  }

where $\Gamma_{\Proj_{B_i}^K}$ is defined as 

\begin{align} \label{eq:gammaPB}
\Gamma_{\Proj_{B_j}^K} \equiv \prod_{l=1}^K(C^{\dagger}_{B_j \cup M'_j \cup F'_j} \otimes I_{B_j^1,...,B_j^K})(I_{F'_j \cup M'_j }\otimes \text{SWAP}_{B_j B_j^l}) (C_{B_j \cup M'_j \cup F'_j} \otimes I_{B_j^1, ..., B_j^K}),
\end{align} 

where $\Gamma_{B_j \cup M_j' \cup F_j'}$ is the same as $\Gamma_{B_j \cup M_j \cup F_j}$ except that it acts on new dummy registers $F_j'$ and $M_j'$ instead of the original registers $F_j$ or $M_j$.

      Finally, reusing Equations  \ref{eq:repeatdefPF}, \ref{eq:gammaPF}, \ref{eq:repeatdefPB}, and  \ref{eq:gammaPB}, the explicit synthesis for $\ket{\phi_{i,j} }$ can be written as (See Equation \ref{def:phiij} for definition of $\ket{\phi_{i,j} }$):

      	  \begin{align}
      	   S_{i,j} \equiv & \left (\Gamma_{\Proj_{B_i}^K} \circ  \Gamma_{\Proj_{F_j}^K}  \circ \Gamma_{i,j} \circ \Gamma_{B_i \cup M_i \cup F_i} \circ \Gamma_{B_j \cup M_j \cup F_j}, (B_i \cup F_j  \cup G), \right . \nonumber\\
      	   &\left . (M_i \cup M_i' \cup F_i' \cup B^1_i,... \cup B^K_i  \cup M_j \cup M_j' \cup B_j' \cup F^1_j,... \cup F^K_j \cup H) \right . \nonumber\\
      	   &\left . ,( N_{i,j} \cup F_i \cup B_j \cup B_i \cup F_j  \cup G \cup M_i \cup M_i' \cup F_i' \cup B^1_i,... \cup B^K_i  \cup M_j \cup M_j' \cup B_j' \cup F^1_j,... \cup F^K_j \cup H ) \right),\nonumber
      	  \end{align}
      	  
      	  where $N_{i,j}$ is defined in the same manner as before:  the register containing all of the qubits between $K_i$ and $K_j$.  \mnote{should come back and address the issue about whether the final $N$ component of this synthesis should contain all previous registers, or just the new ones}

 \end{definition}

\begin{definition}
	Define syntheses

		\begin{align}
		 \Lambda_1^{j,T} \equiv & \left (\Gamma_{\Proj_{B_j}^T}  \circ \Gamma_{B_j \cup M_j \cup F_j}, (B_j  \cup F_j), (M_j \cup M_j' \cup F_j' \cup B^1_j,... \cup B^T_j ) \right . \nonumber\\
		 &\left . ,( B_j \cup M_j \cup F_j \cup M_j' \cup F_j' \cup B^1_j,... \cup B^T_j) \right), \nonumber
		\end{align}
		
			\begin{align}
			 Z_j^{T} \equiv & \left (\Gamma_{\Proj_{B_j}^T} , (B_j ), ( M_j' \cup F_j' \cup B^1_j,... \cup B^T_j ) \right . \nonumber\\
			 &\left . ,( B_j  \cup M_j' \cup F_j' \cup B^1_j,... \cup B^T_j) \right), \nonumber
			\end{align}

		Note that these two objects are, in this case, scalars ( see Definition \ref{def:synth} to understand why).  In fact,
		
			\begin{align} 
			 &Z_j^{T} = \tr(\rho_{B_j}^T ), \nonumber\\
			 & \text{ and }\\
			 &		 \Lambda_1^{j,T} = 	\tr(\rho_{B_j}^T\ket{\psi} \bra{\psi}_{B_j \cup M_j \cup F_j}\rho_{B_j}^T ) 
			\end{align}
		 
		 where
		
			\begin{align} 
			\rho_{B_i}^K \equiv \nonumber  \bra{0^{F_i, M_i, B_i^1,...B_i^K}}\prod_{j=1}^K(C^{\dagger}_{B_i \cup M_i \cup F_i} \otimes I_{B_i^1,...,B_i^K})(I_{F_i \cup M_i }\otimes \text{SWAP}_{B_i B_i^j}) (C_{B_i \cup M_i \cup F_i} \otimes I_{B_i^1, ..., B_i^K})\ket{0^{F_i, M_i, B_i^1,...B_i^K}} 
			\end{align}
			
			as in Lemma \ref{lem:powerblockencoding}.

		   We write the scalars $Z_j^{T}$, $\Lambda_1^{j,T}$ as 2D geometrically-local, shallow depth syntheses to emphasize that they can be computed by any algorithm which computes the probability of zero being output by a 2D geometrically-local, shallow depth synthesis.  In the following analysis we will use the 2D algorithm in Theorem 5 of \cite{BGM19} to compute these quantities to inverse polynomial additive error (see Remark \ref{rm:3Dbasecase}).

\end{definition}

Now we will give a definition of a scalar quantity $\kappa^{i}_{T, \epsilon_2}$ which is meant to be an approximation for the quantity $\lambda^i_1$ that we can compute using the ``base case" algorithm $\mathcal{B}$ described below.  We need this because we want to use $\lambda^i_1$ to normalize terms in Algorithm \ref{alg:quasi-poly-subroutine} below.  We will use its approximation, $\kappa^{i}_{T, \epsilon_2}$, as a substitute, since it is a quantity that we can compute in quasi-polynomial time (even when $T = \log^c(n)$ and $\epsilon_2 = O(1/n^{\log(n)})$, see Definition \ref{def:kappa}).  The quality of this approximation is the subject of Lemma \ref{lem:lambdaapprox}.

\begin{definition}\label{def:kappa}
\begin{align}
\kappa^{j}_{T, \epsilon_2} \equiv \frac{\mathcal{B}(\Lambda_1^{j,T},\epsilon_2)}{\mathcal{B}(Z_j^{2T},\epsilon_2)} = \frac{ \tr(\rho_{B_j}^T\ket{\psi} \bra{\psi}_{B_j \cup M_j \cup F_j}\rho_{B_j}^T ) \pm \epsilon_2}{\tr(\rho_{B_j}^{2T} )\pm \epsilon_2}\nonumber
\end{align}

Here the notation $\mathcal{B}(\Lambda_1^{j,T},\epsilon_2)$ (resp. $\mathcal{B}(Z_j^{2T},\epsilon_2)$ ) denotes a use of algorithm $\mathcal{B}$, which we define to be the algorithm from Theorem 5 of \cite{BGM19} (applied according the prescription in Remark \ref{rm:3Dbasecase}), to compute the scalar quantity $\Lambda_1^{j,T}$ (resp. $Z_j^{2T}$) to within additive error $\epsilon_2$.  We will elaborate further on this computational task (time complexity, etc) in the analysis of Algorithm \ref{alg:quasi-poly-subroutine}.
\end{definition}

\begin{lemma} \label{lem:lambdaapprox}
	If $\lambda_1^i \geq 1 - e(n)$ then $|\kappa^{i}_{T, \epsilon_2} - \lambda_1^i| \leq O \left (\frac{(e(n))^{2T} +\epsilon_2}{(\lambda_1^i)^{2T+1}} \right ) $. 
\end{lemma}

\begin{proof}
	Starting with the definition:\mnote{Should fill in more explanation for each of the steps below, and double check whether an adjustment is needed to account for the fact that the quantum state is not normalized.  At the moment I think no adjustment is needed}
	
	\begin{align}
	&\kappa^{i}_{T, \epsilon_2} \equiv \frac{\mathcal{B}(\Lambda_1^{i,T},\epsilon_2)}{\mathcal{B}(Z_i^{2T},\epsilon_2)} = \frac{ \tr(\rho_{B_i}^T\ket{\psi} \bra{\psi}_{B_i \cup M_i \cup F_i}\rho_{B_i}^T ) \pm \epsilon_2}{\tr(\rho_{B_i}^{2T} )\pm \epsilon_2}\nonumber\\
	& = \frac{(\lambda_1^i)^{2T+1} + O(e(n)^{2T})+O(\epsilon_2)}{(\lambda_1^i)^{2T} + O(e(n)^{2T})+O(\epsilon_2)} =  \frac{(\lambda_1^i)^{2T+1}(1 + O(\frac{e(n)^{2T}+\epsilon_2}{(\lambda_1^i)^{2T+1}}))}{(\lambda_1^i)^{2T}(1 + O(\frac{e(n)^{2T}+\epsilon_2}{(\lambda_1^i)^{2T}}))} \nonumber\\
	&= \lambda_1^i\frac{(1 + O(\frac{e(n)^{2T}+\epsilon_2}{(\lambda_1^i)^{2T+1}}))}{(1 + O(\frac{e(n)^{2T}+\epsilon_2}{(\lambda_1^i)^{2T}}))} 
	= \lambda_1^i \left (1 + O\left (\frac{e(n)^{2T}+\epsilon_2}{(\lambda_1^i)^{2T+1}} \right ) \right ) \nonumber \\
	& = \lambda_1^i + O \left (\frac{e(n)^{2T}+\epsilon_2}{(\lambda_1^i)^{2T+1}} \right )
	\end{align}
	
	The desired result follows.
\end{proof}

 \begin{definition} \label{def:powersetetc}
	For any natural number $\Delta$, we define $[\Delta] \equiv \{1, ... \Delta\}$.  We define $\mathcal{P}([\Delta])$ to be the set of all subsets of $[\Delta]$, that is, the power set of $[\Delta]$.  For any set $\sigma \in\mathcal{P}([\Delta])$, we let $\sigma_{max}$ denote the largest element of $\sigma$.  We let $|\sigma|$ denote the size of the set $\sigma$, and for any $0<i\leq |\sigma|$ we let $\sigma(i)$ denote the $i^{th}$ smallest element of $\sigma$.  
\end{definition}

\section{Estimating Amplitudes in Quasi-polynomial Time}\label{section:quasi-poly-time}

In this section we define and analyze our algorithm for computing $|\bra{0^{\otimes n}}C\ket{0^{\otimes n}}|^2$. \mnote{Write more sentences here to fill in space.  Mention that there is a driver algorithm and a main algorithm, and briefly explailn their roles.  State that we are defining them below.}

\begin{algorithm} \label{alg:constant-width-assumption-driver2}  \label{alg:quasi-poly-driver}
	\SetKwInOut{Input}{Input}\SetKwInOut{Output}{Output}
	\Input{3D Geometrically-Local, \depthname circuit $C$,  base-case algorithm $\mathcal{B}$, approximation error $\delta$}
	\Output{An approximation of  $|\bra{0_{ALL}}C\ket{0_{ALL}}|^2$ to within additive error $\delta$.} 
	
	\tcc{We begin by handling the case in which $\delta$ is so small that it trivializes our runtime, and the case in which $\delta$ is so large that it causes meaningless errors:}
	
	\If{$\delta \leq 1/n^{\log^2(n)}$ \label{ln:checkforsmalldelta}}{\Return The value $|\bra{0_{ALL}}C\ket{0_{ALL}}|^2$ computed with zero error by a ``brute force" $2^{O(n)}$-time algorithm. }
		
		\textbf{if}  $\delta \geq 1/2$ \label{ln:deltalarge} \textbf{then}  \textbf{return} $1/2$

	\tcc{Here begins the non-trivial part of the algorithm:}

	Let $N$ be the register containing all of the qubits on which $C$ acts.  Since these qubits are arranged in a cubic lattice, one of the sides of the cube $N$ must have length at most $n^{\frac{1}{3}}$.  We will call the length of this side the ``width" and will now describe how to ``cut" the cube $N$, and the circuit $C$, perpendicular to this particular side.    
	
	Select $\frac{1}{10d}n^{\frac{1}{3}}$  light-cone separated slices $K_i$ of $10d$ width in $N$, with at most 10$d$ distance between adjacent slices. \label{ln:startslicesearch}  Let $h(n) = \log^7(n)$.  Use the base case algorithm $\mathcal{B}$ to check if at least $\frac{1}{10d}n^{\frac{1}{3}} - h(n)$ of the slices obey:  \label{ln:endslicesearch} 
	
	\begin{equation} 
	\abs{\tr \left (\bra{0_{M_i}}C\ket{0_{ALL}}\bra{0_{ALL}}C^{\dagger}\ket{0_{M_i}}\right)}  \geq 2^{\frac{\log(\delta)}{h(n)}}. \nonumber 
	\end{equation}

	OR, there are fewer than $\frac{1}{10d}n^{\frac{1}{3}} - h(n)$ slices that obey: \label{ln:doubleendslicesearch}

\begin{equation} 
\abs{\tr \left (\bra{0_{M_i}}C\ket{0_{ALL}}\bra{0_{ALL}}C^{\dagger}\ket{0_{M_i}}\right)}  \geq  2^{\frac{\log(\delta)}{h(n)}}. \nonumber \label{eq:heavyweight8}
\end{equation}

\tcc{  See the runtime analysis in the proof of Theorem \ref{thm:mainthmwdelta} for a detailed explanation of how the base case algorithm $\mathcal{B}$ can efficiently distinguish between the above two cases (via Remark \ref{rm:3Dbasecase}).}
	
		\textbf{if}  Fewer than $\frac{1}{10d}n^{\frac{1}{3}} - h(n)$ of the slices obey Line \ref{eq:heavyweight8} \label{ln:easyifalg1} \textbf{then return} 0  
	
	\If{At least $\frac{1}{10d}n^{\frac{1}{3}} - h(n)$ of the slices obey Line \ref{eq:heavyweight8} \label{ln:hardifalg1}}{

		We will denote the set of these slices by $K_{heavy}$.  Note that the maximum amount of width between any two adjacent slices in $K_{heavy}$ is $10d \cdot h(n)$.  Furthermore, the maximum amount of width collectively between $\Delta$ slices in $K_{heavy}$ is $10d\Delta + 10d \cdot h(n)$.  Now that the set $K_{heavy}$ has been defined, we will use this fixed set in the recursive algorithm, Algorithm \ref{alg:constant-width-assumption2}.
		
		Define the geometrically-local, \depthname synthesis $S \equiv (C, L, M, N)$, where $L = M = \emptyset$, are empty registers, and $N$ is the entire input register for the circuit $C$.
		
		\Return $\mathcal{A}(S, \eta = \frac{\log(n)}{3\log(4/3)} , \Delta = \log(n), \epsilon = \delta2^{-10 \log(n) \log(\log(n)))} , h(n) = \log^7(n), K_{heavy} , \mathcal{B})$
	}

	\caption{$\mathcal{A}_{full}(C, \mathcal{B}, \delta)$:  Quasi-Polynomial Time Additive Error Approximation for $|\bra{0_{ALL}}C\ket{0_{ALL}}|^2$.  }
\end{algorithm}

\begin{algorithm}\label{alg:constant-width-assumption2}\label{alg:quasi-poly-subroutine}
	\SetKwInOut{Input}{Input}\SetKwInOut{Output}{Output}
	\Input{3D Geometrically-Local, \depthname synthesis $S$, number of iterations $\eta$, number of cuts $\Delta$, positive base-case error bound $\epsilon > 0$, base-case algorithm $\mathcal{B}$, a set of heavy slices $K_{heavy}$}
	\Output{An approximation of the quantity $\bra{0_{N}}\phi_S\ket{0_{N}}$ where $\phi_S$ is the un-normalized mixed state specified by the 3D geometrically-local, \depthname synthesis $S$, and $\ket{0_{N}}$ is the $0$ state on the entire $N$ register of that synthesis. The approximation error is bounded in the analysis below. }

	Given the geometrically-local, \depthname synthesis $S = (\Gamma, L, M, N)$, let us ignore the registers $L$ and $M$ as they have already been measured or traced-out. 
	
	Let $\ell$ be the width of the $N$ register of the synthesis $S$.  Define the stopping width $w_0 \equiv 20d(\Delta+h(n) + 2)$.

	\If{ $\ell < w_0 = 20d(\Delta+h(n)+2)$  OR $\eta <1$ \label{ln:alg2stoppingwidth} }{Use the base-case algorithm $\mathcal{B}$ to compute the quantity $\bra{0_{N}}\phi_S\ket{0_{N}}$ to within error $\epsilon$.
		
		\Return $\mathcal{B}(S , \epsilon)$ \label{ln:basecase}
		
	}
	
	\Else{ 	
		We will ``slice" the 3D geometrically-local, \depthname synthesis $S$ in $\Delta$ different locations, as follows:
		
		Since $N$ is 3D we define a region $Z \subset N$ to be the sub-cube of $N$ which has width $10d(\Delta+h(n)+2)$, and is centered at the halfway point of $N$ width-wise (about the point $\ell /2$ of the way across $N$). Since the maximum amount of width collectively between $\Delta$ slices in $K_{heavy}$ is $10d\Delta + 10d \cdot h(n)$ (see Algorithm \ref{alg:constant-width-assumption-driver2}), we are guaranteed that the region $Z$ will contain at least $\Delta$ slices, $K_1,K_2,\dots,K_\Delta$, from $K_{heavy}$.  For any two slices $K_i, K_j\in K_{heavy}$, let the un-normalized states $\ket{\varphi_{\L, i}}, \ket{\varphi_{i, j}}, \ket{\varphi_{j, \R}}$, and corresponding sub-syntheses $S_{\L, i}, S_{i,j}, S_{j, \R}$ be as defined in Definition \ref{def:subsynth}, with $K = \log^3(n)$.   We will use these to describe the result of our division step below.\label{ln:alg2defz}
		
		For each $K_i \in K_{heavy}$ pre-compute the quantity $\kappa^{i}_{T, \epsilon_2}$, with $T = \log^3(n)$, and  $\epsilon_2 = \delta 2^{-10 \log(n) \log(\log(n)))}$. \label{ln:compnorm}  \label{ln:pre-compute}
		
		\Return \label{ln:returnofA}\begin{align}
		&\sum_{i=1}^{\Delta}\frac{1}{(\kappa^{i}_{T, \epsilon_2})^{4K+1}}\mathcal{A}(S_{L,i},\eta-1) \cdot \mathcal{A}(S_{i,R},\eta-1) \label{eq:onecutterms}\\
		&-\sum_{i=1}^{\Delta}\sum_{j=i+1}^{\Delta}\frac{1}{(\kappa^{i}_{T, \epsilon_2}\kappa^{j}_{T, \epsilon_2})^{4K+1}}\mathcal{A}(S_{L,i},\eta-1) \cdot \mathcal{B}(S_{i,j},\epsilon)\cdot \mathcal{A}(S_{j,R},\eta-1) \label{eq:twocutterms} \\
		&+\sum_{i=1}^{\Delta}\sum_{j=i+2}^{\Delta}\frac{1}{(\kappa^{i}_{T, \epsilon_2}\kappa^{j}_{T, \epsilon_2})^{4K+1}}\mathcal{A}(S_{L,i},\eta-1) \cdot \mathcal{A}(S_{j,R},\eta-1) \nonumber \\
		&\cdot\Bigg[\sum_{\sigma\in\mathcal{P}(\{i+1, \cdots, j-1\})\setminus\emptyset}(-1)^{\abs{\sigma}+1} \mathcal{B}\left ( \left (\otimes_{k \in \sigma} \Pi^K_{F_k} \bra{0_{M_k}} \right )\phi_{i,j}\left (\otimes_{k \in \sigma} \ket{0_{M_k}}\Pi^K_{F_k}\right ),\frac{\epsilon}{2^\Delta} \right) \Bigg]  \label{eq:basecaseinsum}
		\end{align}

		\tcc{  In the above $\mathcal{B}\left ( \left (\otimes_{k \in \sigma} \Pi^K_{F_k} \bra{0_{M_k}} \right )\phi_{i,j}\left (\otimes_{k \in \sigma} \ket{0_{M_k}}\Pi^K_{F_k}\right ),\frac{\epsilon}{2^\Delta} \right ) $ denotes an $\frac{\epsilon}{2^\Delta}$ approximation of the quantity $\left (\bra{0_{ALL}}\otimes_{k \in \sigma} \Pi^K_{F_k}  \right )\phi_{i,j}\left (\otimes_{k \in \sigma} \Pi^K_{F_k}\ket{0_{ALL}} \right )$ obtained via base case Algorithm $\mathcal{B}$.   Note that for brevity it is implied that $\mathcal{A}(S, \eta)=\mathcal{A}(S, \eta, \Delta, \epsilon, h(n),  K_{heavy},  \mathcal{B})$. }

	}
	\caption{$\mathcal{A}(S, \eta, \Delta, \epsilon, h(n),  K_{heavy},  \mathcal{B})$: Recursive Divide-and-Conquer Subroutine for Algorithm \ref{alg:constant-width-assumption-driver2}.}

\end{algorithm}

\begin{table}\label{table:parameter-list}
	\centering
	\begin{tabular}{| c | c | p{0.6\linewidth} |} 
		\hline
 		Parameter & Value  & Description \\ [1ex] 
 		\hline\hline
 		$C$ & —  & 3D geometrically-local quantum circuit on $n$ qubits. Recall that approximating $|\bra{0^{\otimes n}}C\ket{0^{\otimes n}}|^2$ is the goal of Algorithm \ref{alg:constant-width-assumption-driver2}. \\
 		\hline
 		$\mathcal{B}$ & —  & Algorithm for 2D geometrically-local circuits to be used in the base case of Algorithm \ref{alg:constant-width-assumption2}. See Remark \ref{rm:3Dbasecase} and Theorem 5 of \cite{BGM19} for the base case algorithm used in our analysis. \\ 
 		\hline
 		$\delta$ & $1/n^{-\log(n)}$ & Desired additive error for the approximation output by Algorithm \ref{alg:constant-width-assumption-driver2}.  Note that this is better than inverse polynomial error, for any polynomial. \\
 		\hline
 		$h(n)$ & $\log^7(n)$ & Helps control the overall width of the central region $Z$ (see Description of $Z$ below). \nnote{Should we include the lemma that uses this width?} \\
 		\hline
 		$K_{heavy}$ & — & Set of slices $\{K_i\}$ satisfying Line \ref{eq:heavyweight8}  of Algorithm \ref{alg:constant-width-assumption-driver2}. Existence of these slices follows from an application of Lemma \ref{clm:findheavy}.\\
 		\hline
 		$S$ & —  & Synthesis for a circuit $\Gamma$ as described in Definition \ref{def:subsynth}. During the first run of Algorithm \ref{alg:constant-width-assumption2} this will correspond to a synthesis for the circuit $C$. \\
 		\hline
 		$\eta$ & $\frac{\log(n)}{3\log(4/3)}$  & Maximum depth for the recursive calls to Algorithm \ref{alg:constant-width-assumption2}. $\eta=0$ is one stopping condition for using the base case algorithm in Line \ref{ln:basecase} of Algorithm \ref{alg:constant-width-assumption2}.\\ 
 		\hline
 		$\Delta$ & $\log(n)$  & Number of slices from $K_{heavy}$ that will be used in the division step for Algorithm \ref{alg:constant-width-assumption2}. \\
 		\hline
 		$\epsilon$ & $\delta2^{-10 \log(n) \log(\log(n)))}$  & Desired error for applications of the base-case algorithm in the return statement of Algorithm \ref{alg:constant-width-assumption2}.\\
 		\hline
 		$d$ & —  & Depth of the circuit $C$ \nnote{or is this $\Gamma$?} \\
 		\hline
 		$\ell$ & — & Width of the $N$ register of the synthesis $S$. \\
 		\hline
 		$w_0$ & $20d(\Delta+h(n)+2)$ & Stopping width for Algorithm \ref{alg:constant-width-assumption2}. $\ell<w_0$ is one stopping condition for using the base case algorithm in Line \ref{ln:basecase} of Algorithm \ref{alg:constant-width-assumption2}.\\
 		\hline
 		$Z$ & — & Subset of width $10d(\Delta+h(n)+2)$ in the center of the $N$ register, specified in Algorithm \ref{alg:quasi-poly-subroutine},  from which $\Delta$ slices in $K_{heavy}$ will be chosen. Note that any subproblems contained within $Z$ will, by definition, satisfy the stopping condition $\ell<w_0$, and will consequently be handled by the base-case algorithm $\mathcal{B}$ (see also Remark \ref{rm:3Dbasecase}).\\
 		\hline
 		$K$ & $\log^3(n)$ & The number of repeated compositions of the block encoding for $\rho_{F_i}$ used to produce the approximation $\rho^K_{F_i}$ for the top Schmidt vector $\ket{w_1}\bra{w_1}_{F_i}$.  \\
 		\hline
 		$T$ & $\log^3(n)$ & The number of repeated compositions of the block encoding for $\rho_{B_i}$ used to produce the approximation $\kappa^{i}_{T, \epsilon_2}$ for the top Schmidt coefficient $\lambda_1^i$, as prescribed in Definition \ref{def:kappa}.  \\
 		\hline
 		$\epsilon_2$ & $\delta2^{-10 \log(n) \log(\log(n)))}$ & Desired error for applications of the base-case algorithm when computing the quantities $\kappa_{T,\epsilon_2}^i$ in Line \ref{ln:pre-compute} of Algorithm \ref{alg:constant-width-assumption2}. \\
 		\hline
	\end{tabular}
	\caption{Parameters used within Algorithms \ref{alg:constant-width-assumption-driver2} and \ref{alg:constant-width-assumption2}}
	
\end{table}

\subsection{Run-Time and Error Analysis for Algorithm \ref{alg:quasi-poly-driver}}

\begin{theorem}\label{thm:mainthmwdelta}
Let $C$ be any depth-$d$, 3D geometrically local quantum circuit on $n$ qubits. Algorithm \ref{alg:quasi-poly-driver}, $\mathcal{A}_{full}(C, \mathcal{B}, \delta)$, where $\mathcal{B}$ is the base case algorithm specified in Theorem 5 of \cite{BGM19}, will produce the scalar quantity  $|\bra{0^{\otimes n}}C   \ket{0^{\otimes n}}|^2$ to within $\delta$ error in time
\begin{equation}\label{eq:quasi-poly-runtime-bound}
T(n) = \delta^{-2}2^{d^3\polylog(n)(1/\delta)^{1/\log^2(n)}}
\end{equation}
\end{theorem}

\begin{proof}[Proof of Theorem \ref{thm:mainthmwdelta}:]
 	  The proof proceeds in two parts, the first bounding the approximation error obtained by the algorithm, and the second bounding the runtime. See Table \ref{table:parameter-list} for a brief summary of the parameters used throughout Algorithms \ref{alg:constant-width-assumption-driver2} and \ref{alg:constant-width-assumption2}.
	
	\paragraph{Approximation Error:}
	
	 The analysis of the approximation error obtained by $\mathcal{A}_{full}(C, \mathcal{B}, \delta)$ can be broken into four cases according to the IF statements on Lines \ref{ln:checkforsmalldelta}, \ref{ln:deltalarge}, \ref{ln:easyifalg1}, and \ref{ln:hardifalg1} of Algorithm \ref{alg:quasi-poly-driver}.  The first three cases are easy.  If the condition in Line \ref{ln:checkforsmalldelta} is satisfied, then the specified additive error $\delta$ is so small that we can compute the desired quantity, $|\bra{0_{ALL}}C\ket{0_{ALL}}|^2$, exactly, by brute force, in $2^{O(n)}$ time, and this will still take less time than the guaranteed runtime:  
	
	\[T(n) = \delta^{-2}2^{d^3\polylog(n)(1/\delta)^{1/\log^2(n)}}.\]  
	
	So, if the condition  in Line \ref{ln:checkforsmalldelta} is satisfied, then we are done.  If not, we proceed.
	
	 Next, if the condition in Line \ref{ln:deltalarge} is satisfied, then $\delta \geq 1/2$, in which case, outputting $0$ is clearly a $\delta$ additive approximation of $|\bra{0_{ALL}}C\ket{0_{ALL}}|^2$, since $0 \leq |\bra{0_{ALL}}C\ket{0_{ALL}}|^2 \leq 1$.  So, if this is the case, we are done, otherwise we proceed.
	
	 Next, if the condition in Line \ref{ln:easyifalg1} is satisfied, then, either $\delta \geq 1/2$  
	 (in which case, outputting $1/2$ is clearly a $\delta$ additive approximation of $|\bra{0_{ALL}}C\ket{0_{ALL}}|^2$, since $0 \leq |\bra{0_{ALL}}C\ket{0_{ALL}}|^2 \leq 1$), OR there is a set of slices $K_{lightweight}$ of size at least $h(n)$, such that every slice $K_i \in K_{lightweight}$ satisfies:

	\begin{equation} 
	\| \bra{0_{M_i}}C\ket{0_{ALL}} \|^2 = \tr \left (\bra{0_{M_i}}C\ket{0_{ALL}}\bra{0_{ALL}}C^{\dagger}\ket{0_{M_i}}\right)  < 2^{\frac{\log(\delta)}{h(n)}}. \nonumber \label{eq:lightweight9}  
	\end{equation}
	\mnote{need to correct this for the fact that the base case algorithm can only estimate this quantity approximately}
	
	In this case, since, for all $K_i, K_j \in K_{lightweight}$ with $K_i \neq K_j$ we know that $K_i$ is lightcone separated from $K_j$.  It follows that:
	
	\begin{align} 
	|\bra{0_{ALL}}C\ket{0_{ALL}}| \leq \Pi_{K_i \in K} \| \bra{0_{M_i}}C\ket{0_{ALL}} \|\leq \left (2^{\frac{\log(\delta)}{2h(n)}}\right)^{h(n)}  = 2^{\log(\delta)/2} = \sqrt{\delta}. \label{eq:easycaseerror}
	\end{align}
	
	So,
	\begin{align} 
	|\bra{0_{ALL}}C\ket{0_{ALL}}|^2 \leq  \delta. \label{eq:easycaseerror2}
	\end{align}

	\mnote{maybe explain the partial trace in the above equation more?}
	
	Therefore, in this case, Algorithm \ref{alg:quasi-poly-driver} returns the quantity $0$ as an answer \mnote{cite the line number for this?}, which is trivially a $\delta$-additive error approximation of $|\bra{0_{ALL}}C\ket{0_{ALL}}|^2$ by Equation \ref{eq:easycaseerror}.

	On the other hand, if the IF statement on Line \ref{ln:easyifalg1} of Algorithm \ref{alg:quasi-poly-driver} is not satisfied, then that means that the IF statement of Line \ref{ln:hardifalg1} must be satisfied, by definition.  In that case Algorithm \ref{alg:quasi-poly-driver} returns the quantity:
	
	\[\mathcal{A}(S, \eta = \frac{\log(n)}{3\log(4/3)}, \Delta = \log(n), \epsilon = \delta 2^{-10 \log(n) \log(\log(n))} , h(n) = log^7(n), K_{heavy} , \mathcal{B})\]
	
	which we know is an $f(S,\eta,\Delta,\epsilon)$-additive error approximation of $|\bra{0_{ALL}}C\ket{0_{ALL}}|^2$. With $\eta =\frac{\log(n)}{3\log(4/3)}$, by Lemma \ref{lem:error-bound-even-spaced} we know that:
	
	\mnote{the following equation needs to be updated to account for $E_3$}
	
	\mnote{Add comment discussing the tradeoff with $\Delta$, and the need for $K > \Delta$}
	
	\begin{align*}
	f(S,\eta,\Delta,\epsilon) &\leq \eta 20^{\eta}\Delta^{2\eta}\left(E_3(n, K, T, \epsilon_2, \epsilon,\Delta) + (2e(n)+2\projerror)^{\Delta} \right) \\
	&= \eta 20^{\eta}\Delta^{3\eta}	O \left (2^{\Delta} (2e(n))^K + 2^\Delta K \left (e(n)^{2T}+\epsilon_2 \right ) + \epsilon \right)  \\
		&= 1/ 3\log(4/3) \cdot  \log(n) 20^{\frac{\log(n)}{3\log(4/3)}}(\log(n))^{3\frac{\log(n)}{3\log(4/3)}} O	\left (2^{\log(n)} (2(1 - 2^{\frac{\log(\delta)}{\log^7(n)}}))^{\log^3(n)} \right . \\
		& \left .+ 2^{\log(n)} \log^3(n)\left ((1 - 2^{\frac{\log(\delta)}{\log^7(n)}})^{2\log^3(n)}+\epsilon_2 \right ) +  \delta 2^{-10 \log(n) \log(\log(n))} \right )\\
		& \leq  (\log(n))^{2\log(n)} \cdot \poly(n) \cdot 	\left ( (2(1 - 2^{\frac{\log(\delta)}{\log^7(n)}}))^{\log^3(n)} +\epsilon_2  + \delta 2^{-10 \log(n) \log(\log(n))} \right )\\
		&  \leq  (\log(n))^{2\log(n)} \cdot \poly(n) \cdot \left (	\left (  O\left (\frac{1}{\log^4(n)}\right )\right )^{\log^3(n)} +2 \cdot \delta 2^{-10 \log(n) \log(\log(n))} \right )\\
		&\leq  2^{2\log(n)\log(\log(n))} \cdot \poly(n) \cdot 	\left (  O\left (\frac{1}{\log^4(n)}\right )\right )^{\log^3(n)} + \delta 2^{-8 \log(n) \log(\log(n))} \\
	&  \leq  o(1) \cdot \delta + o(1)  \cdot \delta =  o(1)  \cdot \delta \numberthis \label{eq:upper-bound-even-spaced-larger2}
	\end{align*}

	\mnote{need to update the above equation to exactly match the quantities in the new error analysis.}
	where the first inequality follows from Lemma \ref{lem:error-bound-even-spaced} and the rest follows by calculation, noting that $E_3(n, K, T, \epsilon_2, \epsilon,\Delta) \geq  (2e(n)+2\projerror)^{\Delta}$ for our specific choice of parameters (in particular $\Delta = \log(n)$), recalling that $e(n) \leq (1 - 2^{\frac{\log(\delta)}{\log^7(n)}}) = O(1/\log^4(n))$   (since $\delta \geq n^{-\log^2(n)} = 2^{-\log(n)^3}$ as verified in Algorithm \ref{alg:quasi-poly-driver}), $K = \log^3(n)$, $T = \log^3(n)$, and $\epsilon_2 = \delta 2^{-10 \log(n) \log(\log(n))}$.  The final inequality, which claims $ 2^{2\log(n)\log(\log(n))} \cdot \poly(n) \cdot 	\left (  O\left (\frac{1}{\log^2(n)}\right )\right )^{\log^2(n)} = o(1) \cdot \delta$, again follows because $\delta \geq n^{-\log^2(n)}$ as verified in the driver algorithm, Algorithm \ref{alg:quasi-poly-driver}.

	\paragraph{Runtime:}
	
	The runtime analysis of Algorithm \ref{alg:quasi-poly-driver}, $\mathcal{A}_{full}(C, \mathcal{B}, \delta)$, proceeds by considering the same four cases in the IF statements on Lines \ref{ln:checkforsmalldelta}, \ref{ln:deltalarge}, \ref{ln:easyifalg1}, and \ref{ln:hardifalg1}, just as in the error analysis above.  Just as before, the first three cases are easy.  If the IF statement Line \ref{ln:checkforsmalldelta} is satisfied, then the specified additive error $\delta$ is so small that we can compute the desired quantity, $|\bra{0_{ALL}}C\ket{0_{ALL}}|^2$, exactly, by brute force, in $2^{O(n)}$ time, and this will still take less time than the guaranteed runtime:  
	
	\[T(n) = \delta^{-2}2^{d^3\polylog(n)(1/\delta)^{1/\log^2(n)}}.\]  
	
	So, if the IF statement on Line \ref{ln:checkforsmalldelta} is satisfied, then we are done.
	
	If the the IF statement in Line \ref{ln:deltalarge} is satisfied then the algorithm outputs $1/2$, which is a constant time operation, and we are done.

	On the other hand, in the case that these first two IF statements are not satisfied, we must bound the running time of Line \ref{ln:startslicesearch}.  Line \ref{ln:startslicesearch} can clearly be done in polynomial time, which is an additive cost that is significantly less than our ultimate quasi-polynomial running time upper bound, so we can absorb it into the $O(\cdot)$ notation, and continue without explicitly tracking it.  Line \ref{ln:endslicesearch} calls for the use of the base case algorithm $\mathcal{B}$ (which we have specified to be the algorithm from Theorem 5 of \cite{BGM19}) to estimate, for every slice $K_i$, the quantity: 
		
		\begin{equation} 
		\bra{0_{ALL}}C^{\dagger}\ket{0_{M_i}}\bra{0_{M_i}}C\ket{0_{ALL}} = \tr \left (\bra{0_{M_i}}C\ket{0_{ALL}}\bra{0_{ALL}}C^{\dagger}\ket{0_{M_i}}\right),   \nonumber  
		\end{equation}
		
		In particular we use $\mathcal{B}$ to estimate this quantity to within additive error $\tilde{\epsilon} \equiv 2^{\frac{\log(\delta)}{2h(n)} - 1} - 2^{\frac{\log(\delta)}{h(n)}-1}$, and we count only those slices for which the approximation output by $B$ is at least $2^{\frac{\log(\delta)}{h(n)} } + \tilde{\epsilon}$.  So, all the slices accepted by this count will necessarily have weight at least $2^{\frac{\log(\delta)}{h(n)}}+ \tilde{\epsilon} - \tilde{\epsilon} = 2^{\frac{\log(\delta)}{h(n)} }$.  Furthermore, any slice with weight at least $2^{\frac{\log(\delta)}{h(n)}} + 2 \tilde{\epsilon} = 2^{\frac{\log(\delta)}{2h(n)} }$ will certainly be counted by this process.  Therefore, this procedure is able to determine which of Line \ref{ln:endslicesearch} or Line \ref{ln:doubleendslicesearch} is true. (As noted in a comment in the Algorithm, one of these two must be the case.)  It remains to bound the running time cost of these uses of algorithm $\mathcal{B}$.  The key observation here is that the quantity:
		
			\begin{equation} 
			 \tr \left (\bra{0_{M_i}}C\ket{0_{ALL}}\bra{0_{ALL}}C^{\dagger}\ket{0_{M_i}}\right),   \nonumber  
			\end{equation}
			
		only depends on the part of the circuit $C$ that lies in the lightcone of slice $M_i$.  By definition the lightcone of $M_i$ is contained in $K_i$, and $K_i$ is a 2D slice with thickness $30d = O(d)$ in third dimension.  By the discussion in Remark \ref{rm:3Dbasecase} it follows that the base case algorithm   $\mathcal{B}$ can compute the quantity: 
		\begin{equation} 
		\tr \left (\bra{0_{M_i}}C\ket{0_{ALL}}\bra{0_{ALL}}C^{\dagger}\ket{0_{M_i}}\right),   \nonumber  
		\end{equation}
		
		to within additive error $\tilde{\epsilon} \equiv 2^{\frac{\log(\delta)}{2h(n)} - 1} - 2^{\frac{\log(\delta)}{h(n)}-1}$ in time $O(\poly(n)/\tilde{\epsilon}^2)
 = O(\poly(n)/(1 - \delta^{1/2h(n)})^2) = O(2^{\polylog(n)})$, where the final equality follows by straightforward calculation whenever $\delta \leq 1/2$ (recall that we previously established that $\delta \leq 1/2$  by checking that the IF statement in Line \ref{ln:deltalarge} was not satisfied).

	Next, if the IF statement in Line \ref{ln:easyifalg1} is satisfied, then our Algorithm \ref{alg:quasi-poly-driver} returns the quantity $0$ as an answer, which is a constant time operation, and we are done.

	On the other hand, if the IF statement on Line \ref{ln:easyifalg1} of Algorithm \ref{alg:quasi-poly-driver} is not satisfied, then that means that the IF statement of Line \ref{ln:hardifalg1} must be satisfied, by definition.  In that case Algorithm \ref{alg:quasi-poly-driver} returns the quantity:
	
	\[\mathcal{A}(S, \eta = 1/3 \log(n), \Delta = \log(n), \epsilon = \delta 2^{-10 \log(n) \log(\log(n)))} , h(n) = log^3(n), K_{heavy} , \mathcal{B})\]
	
	which we know is an $f(S,\eta,\Delta,\epsilon)$-additive error approximation of $|\bra{0_{ALL}}C\ket{0_{ALL}}|^2$, and takes $
	T(n)<\delta^{-2} 2^{d^3 \polylog(n)}$  time to compute.  This time bound follows directly from the runtime bound on Algorithm \ref{alg:quasi-poly-subroutine}, which is given in Theorem \ref{thm:quasi-poly-runtime-bound} of Subsection \ref{subsection:quasi-run-time}. All together, regardless of which IF statements are true, no step of Algorithm \ref{alg:quasi-poly-driver} exceeds a running time of $\delta^{-2}2^{d^3\polylog(n)(1/\delta)^{1/\log^2(n)}}$, and so we are done.
	
	\mnote{could point out here that, for the non-trivial part of the algorithm the $\delta$ dependence appears substantially better than in the edge cases.  The edge cases may be regarded as a technicality.}
	

\end{proof}

\subsection{Error Analysis for Algorithm \ref{alg:constant-width-assumption2}}\label{subsection:error-constant-width2}

In this subsection we will derive, by induction, an error bound on the estimate produced by Algorithm \ref{alg:constant-width-assumption2}, $\mathcal{A}(S, \eta, \Delta, \epsilon, \mathcal{B}, K_{heavy}, r)$.  We will only pursue an error analysis of  $\mathcal{A}$ under the assumption that the driver algorithm, Algorithm \ref{alg:quasi-poly-driver}, has actually called Algorithm \ref{alg:quasi-poly-subroutine},  and has thus constructed the set $K_{heavy}$ according to specification.  This is because, if Algorithm \ref{alg:quasi-poly-driver} does not call Algorithm \ref{alg:quasi-poly-subroutine}, that means that it has already found an easier approximation to the answer, and the output of $\mathcal{A}$ (Algorithm \ref{alg:quasi-poly-subroutine}) is not relevant.   Recall that, given a synthesis $S$, the goal of Algorithm \ref{alg:constant-width-assumption2} is to compute the quantity $|\bra{0_N}\ket{\phi}\bra{\phi}\ket{0_N}|^2$ where $\ket{\phi}$ is the state synthesized by synthesis $S$, and $N$ is the active register, as defined in Definition \ref{def:synth}.  The Algorithm $\mathcal{A}(S, \eta, \Delta, \epsilon, \mathcal{B},  K_{heavy}, r)$ is a recursive algorithm and, since the variables $\Delta, \epsilon, \mathcal{B}, K_{heavy},$ and $r$ remain unchanged throughout, we will use a simplification $\mathcal{A}(S, \eta, \Delta, \epsilon, \mathcal{B}, K_{heavy}, r) = \mathcal{A}(S, \eta)$ throughout this analysis. The output of Algorithm \ref{alg:constant-width-assumption2} is the scalar quantity:

\begin{equation}
\begin{split}
{}&\mathcal{A}(S, \eta) \equiv \sum_{i=1}^{\Delta}\frac{1}{(\kappa^{i}_{T, \epsilon_2})^{4K+1}}\mathcal{A}(S_{L,i},\eta-1) \cdot \mathcal{A}(S_{i,R},\eta-1) \\
&-\sum_{i=1}^{\Delta}\sum_{j=i+1}^{\Delta}\frac{1}{(\kappa^{i}_{T, \epsilon_2}\kappa^{j}_{T, \epsilon_2})^{4K+1}}\mathcal{A}(S_{L,i},\eta-1) \cdot \mathcal{B}(S_{i,j},\epsilon)\cdot \mathcal{A}(S_{j,R},\eta-1) \\
&+\sum_{i=1}^{\Delta}\sum_{j=i+2}^{\Delta}\frac{1}{(\kappa^{i}_{T, \epsilon_2}\kappa^{j}_{T, \epsilon_2})^{4K+1}}\mathcal{A}(S_{L,i},\eta-1) \cdot \mathcal{A}(S_{j,R},\eta-1) \\
&\cdot\Bigg[\sum_{\sigma\in\mathcal{P}(\{i+1, \dots, j-1\})\setminus\emptyset}(-1)^{\abs{\sigma}+1} \mathcal{B}\left ( \left (\otimes_{k \in \sigma} \Pi^K_{F_k} \bra{0_{M_k}} \right )\phi_{i,j}\left (\otimes_{k \in \sigma} \ket{0_{M_k}}\Pi^K_{F_k}\right ),\frac{\epsilon}{2^\Delta} \right) \Bigg]
\end{split}
\end{equation}
Here the notation $\mathcal{B}(T,\epsilon)$ denotes the use of the ``base case" algorithm to estimate the quantity $|\bra{0_{ALL}}\ket{\phi_T}\bra{\phi_T}\ket{0_{ALL}}|^2$ to within a desired additive error $\epsilon$, for the specified synthesis $T$. For us the ``base case" algorithm will be defined to be the algorithm from Theorem 5 of \cite{BGM19}.  This algorithm can be used for middle sections of the circuit since these sections have a 2D geometry with a ``thickness" of at most a polylogarithmic number of qubits in the third dimension.  The algorithm from Theorem 5 of \cite{BGM19} can compute an $\epsilon$-additive-error approximation of output probabilities of such syntheses in $O(2^{\polylog(n)}\poly(\frac{1}{\epsilon}))$ time.  Note this is not stated explicitly in \cite{BGM19}, which technically only handles true 2D circuits (in other words, circuits with ``thickness" exactly 1 in the third dimension), but their techniques can be extended to the case of polylogarithmic thickness in a straightforward manner (to do so, increase the bond dimension of their Matrix Product States to polylogarithmic size account for the added ``thickness" of qubits). 

Since we have assumed that Algorithm \ref{alg:quasi-poly-driver} has called Algorithm \ref{alg:quasi-poly-subroutine}, we know that every slice $K_i$ in the input set $K_{heavy}$ to $\mathcal{A}(S, \eta, \Delta, w_0, \epsilon, \mathcal{B}, K_{heavy}, r) = \mathcal{A}(S, \eta)$ satisfies:

\begin{equation} 
\abs{\tr \left (\bra{0_{M_i}}C\ket{0_{ALL}}\bra{0_{ALL}}C^{\dagger}\ket{0_{M_i}}\right)}  \geq 2^{\frac{\log(\delta)}{h(n)}} = 2^{\frac{-\log(1/\delta)}{h(n)}}. \nonumber 
\end{equation}

\mnotetwo{{\color{gray} See Algorithm \ref{alg:multilevel} for the definitions of the summands in this sum.  This output is meant to approximate the quantity:

		\begin{align}
		&\sum_{\sigma \in \mathcal{P}([\Delta]): \sigma \neq \emptyset} (-1)^{|\sigma|+1} \bra{\phi_{\L, \sigma(1)}}\ket{0_{N_{\L, \sigma(1)}}} \cdot \prod_{i=1}^{ |\sigma|-1}  \bra{\phi_{\sigma(i), \sigma(i+1)}} \ket{0_{N_{\sigma(i), \sigma(i+1)}}}\cdot  \bra{\phi_{\sigma(|\sigma|), \R}}\ket{0_{N_{\sigma(|\sigma|), \R}}},
		\end{align}
		
		where $\bra{\phi_{\L, \sigma(1)}}, \bra{\phi_{\sigma(i), \sigma(i+1)}},   \bra{\phi_{\sigma(|\sigma|), \R}}$ are the states synthesized by syntheses $S_{\L, \sigma(1)}, S_{\sigma(i), \sigma(i+1)}, S_{\sigma(|\sigma|), \R}$ respectively, as in Definition \ref{def:synth}.  The syntheses $S_{\L, \sigma(1)}, S_{\sigma(i), \sigma(i+1)}, S_{\sigma(|\sigma|), \R}$ are defined in Algorithm \ref{alg:multilevel}. }\nnote{Fix the above quantity.}}

Equivalently, $p_{total}(M_i = 0) \geq 1 - e(n)$, where we define: 

\begin{equation}
e(n) \equiv (1 - 2^{\frac{-\log(1/\delta)}{h(n)}} \label{def:e(n)}).
\end{equation}  

It follows from Lemma \ref{lem:highschmidtnew} that,  $\forall K_i \in K_{heavy}$, $\lambda_1^i  \geq 1 - O(e(n))$.

We know, by Lemma \ref{clm:expansiontrick} that,

\begin{equation}
\left \| \sum_{\sigma\in\mathcal{P}([\Delta])} (-1)^{\abs{\sigma}} \ket{\Psi_\sigma}\bra{\Psi_\sigma}\right \| = \left \|	\ket{\Psi_\emptyset}\bra{\Psi_\emptyset} -  \sum_{\sigma\in\mathcal{P}([\Delta])\setminus\emptyset} (-1)^{\abs{\sigma}+1} \ket{\Psi_\sigma}\bra{\Psi_\sigma}\right \| \leq (2e(n)+2\projerror)^{\Delta} \label{eq:expbounderroan}
\end{equation}
where $\projerror \equiv \left (\frac{1 - \lambda_1^i}{\lambda_1^i} \right )^K \leq \left (\frac{e(n)}{1-O(e(n))} \right )^K $, and the states $\ket{\Psi_\sigma}$ are defined as:

$$\ket{\Psi_\sigma}=\otimes_{j \in \sigma }\Pi^K_{F_j} \otimes_{i\in [\Delta]} \bra{0_{M_i}}C\ket{0_{\all}},$$
$$\ket{\Psi_\emptyset}=\otimes_{i\in [\Delta]} \bra{0_{M_i}}C\ket{0_{\all}}.$$

Note that, $\bra{0_{ALL}}\ket{\Psi_\emptyset}\bra{\Psi_\emptyset}\ket{0_{ALL}}$  is exactly the quantity that we wish for Algorithm \ref{alg:quasi-poly-subroutine} to output!  So, the error between the returned output of Algorithm \ref{alg:quasi-poly-subroutine}, (defined on Line \ref{ln:returnofA} of that algorithm), which we will denote by $\mathcal{A}$ for short, and the desired output quantity $\bra{0_{ALL}}\ket{\Psi_\emptyset}\bra{\Psi_\emptyset}\ket{0_{ALL}}$ is:

\mnote{Uncomment the above figure for arxiv version!!!}

\begin{align}
{}&f(S,\eta,\Delta,\epsilon) \leq \Big\| \bra{0_{ALL}}\ket{\Psi_\emptyset}\bra{\Psi_\emptyset}\ket{0_{ALL}} - \mathcal{A} \Big\| \\
{}&\leq \Big\| \bra{0_{ALL}}\ket{\Psi_\emptyset}\bra{\Psi_\emptyset}\ket{0_{ALL}} - \sum_{\sigma\in\mathcal{P}([\Delta])\setminus\emptyset} (-1)^{\abs{\sigma}+1} \bra{0_{ALL}}\ket{\Psi_\sigma}\bra{\Psi_\sigma}\ket{0_{ALL}} \Big\| \\
&+ \Big\| \sum_{\sigma\in\mathcal{P}([\Delta])\setminus\emptyset} (-1)^{\abs{\sigma}+1} \bra{0_{ALL}}\ket{\Psi_\sigma}\bra{\Psi_\sigma}\ket{0_{ALL}} - \mathcal{A} \Big\| \\
\begin{split}
{}&\leq \left (2e(n)+2\projerror \right)^{\Delta} + \Big\| \sum_{\sigma\in\mathcal{P}([\Delta])\setminus\emptyset} (-1)^{\abs{\sigma}+1} \bra{0_{ALL}}\ket{\Psi_\sigma}\bra{\Psi_\sigma}\ket{0_{ALL}} - \mathcal{A} \Big\| \\
& = \left (2e(n)+2\projerror \right)^{\Delta} + \Bigg\|\sum_{\sigma\in\mathcal{P}([\Delta])\setminus\emptyset} (-1)^{\abs{\sigma}+1} \bra{0_{ALL}}\ket{\Psi_\sigma}\bra{\Psi_\sigma}\ket{0_{ALL}} \\
&-\Bigg( \sum_{i=1}^{\Delta}\frac{1}{(\kappa^{i}_{T, \epsilon_2})^{4K+1}}\mathcal{A}(S_{L,i},\eta-1) \cdot \mathcal{A}(S_{i,R},\eta-1) \\
&-\sum_{i=1}^{\Delta}\sum_{j=i+1}^{\Delta}\frac{1}{(\kappa^{i}_{T, \epsilon_2}\kappa^{j}_{T, \epsilon_2})^{4K+1}}\mathcal{A}(S_{L,i},\eta-1) \cdot \mathcal{B}(S_{i,j},\epsilon)\cdot \mathcal{A}(S_{j,R},\eta-1) \\
&+\sum_{i=1}^{\Delta}\sum_{j=i+2}^{\Delta}\frac{1}{(\kappa^{i}_{T, \epsilon_2}\kappa^{j}_{T, \epsilon_2})^{4K+1}}\mathcal{A}(S_{L,i},\eta-1) \cdot \mathcal{A}(S_{j,R},\eta-1) \\
&\cdot\Bigg[\sum_{\sigma\in\mathcal{P}(\{i+1,\dots,j-1\})\setminus\emptyset}(-1)^{\abs{\sigma}+1} \mathcal{B}\left ( \left (\otimes_{k \in \sigma} \Pi^K_{F_k} \bra{0_{M_k}} \right )\phi_{i,j}\left (\otimes_{k \in \sigma} \ket{0_{M_k}}\Pi^K_{F_k}\right ),\frac{\epsilon}{2^\Delta} \right) \Bigg]
\Bigg)\Bigg\|
\end{split}
\end{align}

Grouping analogous terms and using triangle inequality gives:

\begin{align*}
&f(S,\eta,\Delta,\epsilon) \leq \left (2e(n)+2\projerror \right)^{\Delta} \\
& +  \Bigg \|  \sum_{i=1}^{\Delta}\left ( \frac{1}{(\kappa^{i}_{T, \epsilon_2})^{4K+1}}\mathcal{A}(S_{L,i},\eta-1) \cdot \mathcal{A}(S_{i,R},\eta-1)  - \bra{0_{ALL}}\ket{\Psi_{\{i\}}}\bra{\Psi_{\{i\}}}\ket{0_{ALL}} \right )\\
&-\sum_{i=1}^{\Delta}\sum_{j=i+1}^{\Delta} \left (\frac{1}{(\kappa^{i}_{T, \epsilon_2}\kappa^{j}_{T, \epsilon_2})^{4K+1}}\mathcal{A}(S_{L,i},\eta-1) \cdot \mathcal{B}(S_{i,j},\epsilon)\cdot \mathcal{A}(S_{j,R},\eta-1) - \bra{0_{ALL}}\ket{\Psi_{\{i,j\}}}\bra{\Psi_{\{i,j\}}}\ket{0_{ALL}} \right ) \\
&+\sum_{i=1}^{\Delta}\sum_{j=i+2}^{\Delta}\sum_{\sigma\in\mathcal{P}(\{i+1,\dots,j-1\})\setminus\emptyset}\left ( \frac{1}{(\kappa^{i}_{T, \epsilon_2}\kappa^{j}_{T, \epsilon_2})^{4K+1}}\mathcal{A}(S_{L,i},\eta-1) \cdot \mathcal{A}(S_{j,R},\eta-1) \right .\\
& \cdot(-1)^{\abs{\sigma}+1} \mathcal{B}\left ( \left (\otimes_{k \in \sigma} \Pi^K_{F_k} \bra{0_{M_k}} \right )\phi_{i,j}\left (\otimes_{k \in \sigma} \ket{0_{M_k}}\Pi^K_{F_k}\right ),\frac{\epsilon}{2^\Delta} \right)  \\
&\left .- \bra{0_{ALL}}\ket{\Psi_{\{i,j\}\cup \sigma}}\bra{\Psi_{\{i,j\}\cup \sigma}}\ket{0_{ALL}} \right )
\Bigg\|\\
& \leq \left (2e(n)+2\projerror \right)^{\Delta} \\
& +  \sum_{i=1}^{\Delta}\Bigg \|  \left ( \frac{1}{(\kappa^{i}_{T, \epsilon_2})^{4K+1}}\mathcal{A}(S_{L,i},\eta-1) \cdot \mathcal{A}(S_{i,R},\eta-1)  - \bra{0_{ALL}}\ket{\Psi_{\{i\}}}\bra{\Psi_{\{i\}}}\ket{0_{ALL}} \right )  \Bigg \|\\
&+  \sum_{i=1}^{\Delta}\sum_{j=i+1}^{\Delta}\Bigg \| \left (\frac{1}{(\kappa^{i}_{T, \epsilon_2}\kappa^{j}_{T, \epsilon_2})^{4K+1}}\mathcal{A}(S_{L,i},\eta-1) \cdot \mathcal{B}(S_{i,j},\epsilon)\cdot \mathcal{A}(S_{j,R},\eta-1) - \bra{0_{ALL}}\ket{\Psi_{\{i,j\}}}\bra{\Psi_{\{i,j\}}}\ket{0_{ALL}} \right )  \Bigg \|\\
&+\sum_{i=1}^{\Delta}\sum_{j=i+2}^{\Delta}\Bigg \| \sum_{\sigma\in\mathcal{P}(\{i+1,\dots,j-1\})\setminus\emptyset}(-1)^{|\sigma|+1} \left (  \frac{1}{(\kappa^{i}_{T, \epsilon_2}\kappa^{j}_{T, \epsilon_2})^{4K+1}}\mathcal{A}(S_{L,i},\eta-1) \cdot \mathcal{A}(S_{j,R},\eta-1) \right .\\
& \cdot \mathcal{B}\left ( \left (\otimes_{k \in \sigma} \Pi^K_{F_k} \bra{0_{M_k}} \right )\phi_{i,j}\left (\otimes_{k \in \sigma} \ket{0_{M_k}}\Pi^K_{F_k}\right ),\frac{\epsilon}{2^\Delta} \right)  \\
&\left .- \bra{0_{ALL}}\ket{\Psi_{\{i,j\}\cup \sigma}}\bra{\Psi_{\{i,j\}\cup \sigma}}\ket{0_{ALL}} \right )
\Bigg\| \numberthis \label{eq:threesplits}
\end{align*}

The three summations in Equation \ref{eq:threesplits} are written separately for notational convenience.  The first summation includes all the terms corresponding to the case when the input synthesis is ``sliced" at exactly one cut, and these terms are bounded in Lemma \ref{lem:singlecuttrick}.  The second summation includes all the terms involving ``slices" at exactly 2 cuts, and these terms are bounded in Lemma \ref{lem:doublecuttrick}.  The third summation includes all the remaining terms, which cover all the cases involving ``slices" at three or more of the $\Delta$ cuts.  Each of these terms is bounded in Lemma \ref{lem:multicuttrick}.  For a depiction of these three cases see Figure \ref{fig:error_terms}.

\begin{lemma} \label{lem:singlecuttrick}
	
	\begin{align*}
&\Bigg \|  \left ( \frac{1}{(\kappa^{i}_{T, \epsilon_2})^{4K+1}}\mathcal{A}(S_{L,i},\eta-1) \cdot \mathcal{A}(S_{i,R},\eta-1)  - \bra{0_{ALL}}\ket{\Psi_{\{i\}}}\bra{\Psi_{\{i\}}}\ket{0_{ALL}} \right )  \Bigg \| \\
&\leq E_1(n, K, T, \epsilon_2) +2f(S,\eta-1,\Delta,\epsilon),
	\end{align*}
	
	where $E_1(n, K, T, \epsilon_2) \equiv 10K(e(n)^{2T} + 6\projerror +  \epsilon_2)$.
\end{lemma}

\begin{proof}
	The proof of this Lemma is a simpler special case of the proof of Lemma \ref{lem:multicuttrick} below.  It is simpler in that it follows by using Lemma \ref{lem:lambdaapprox}, and Lemma \ref{clm:breaktoproduct}, and does not require the use of Lemma \ref{clm:expansiontrick} as the proof of Lemma \ref{lem:multicuttrick} does.  For succinctness, instead of writing out this entire proof, we refer the reader to the proof of Lemma \ref{lem:multicuttrick} in the Appendix, of which the proof of this Lemma is a special case. \mnote{Maybe don't need Lemma \ref{clm:schmidproj} because it is already used in the proof of Lemma \ref{clm:breaktoproduct}??}
\end{proof}

\begin{lemma}\label{lem:doublecuttrick}
		\begin{equation}
		\begin{split}
		\Bigg \| \left (\frac{1}{(\kappa^{i}_{T, \epsilon_2}\kappa^{j}_{T, \epsilon_2})^{4K+1}}\mathcal{A}(S_{L,i},\eta-1) \cdot \mathcal{B}(S_{i,j},\epsilon)\cdot \mathcal{A}(S_{j,R},\eta-1) - \bra{0_{ALL}}\ket{\Psi_{\{i,j\}}}\bra{\Psi_{\{i,j\}}}\ket{0_{ALL}} \right )  \Bigg \| \\ 
		\leq E_2(n, K, T, \epsilon_2, \epsilon) +2f(S,\eta-1,\Delta,\epsilon),
		\end{split}
		\end{equation}
		
		where $E_2(n, K, T, \epsilon_2, \epsilon) \equiv 10K(e(n)^{2T} + 6\projerror +  \epsilon_2) + \epsilon $
\end{lemma}

\begin{proof}
	
	The proof of this Lemma is a simpler special case of the proof of Lemma \ref{lem:multicuttrick} below.  It is simpler in that it follows by using Lemma \ref{lem:lambdaapprox}, and Lemma \ref{clm:breaktoproduct}, and does not require the use of Lemma \ref{clm:expansiontrick} as the proof of Lemma \ref{lem:multicuttrick} does.  For succinctness, instead of writing out this entire proof, we refer the reader to the proof of Lemma \ref{lem:multicuttrick} in the Appendix, of which the proof of this Lemma is a special case.
\end{proof}

\begin{lemma} \label{lem:multicuttrick}
	\begin{align*}
	&\Bigg \|\sum_{\sigma\in\mathcal{P}(\{i+1, \dots, j-1\})\setminus\emptyset} (-1)^{|\sigma|+1} \Bigg (  \frac{1}{(\kappa^{i}_{T, \epsilon_2}\kappa^{j}_{T, \epsilon_2})^{4K+1}}\mathcal{A}(S_{L,i},\eta-1) \cdot \mathcal{A}(S_{j,R},\eta-1) \\
	& \cdot  \mathcal{B}\left ( \Big (\otimes_{k \in \sigma} \Pi^K_{F_k} \bra{0_{M_k}} \Big )\phi_{i,j}\Big (\otimes_{k \in \sigma} \ket{0_{M_k}}\Pi^K_{F_k}\Big ),\frac{\epsilon}{2^\Delta} \right) - \bra{0_{ALL}}\ket{\Psi_{\{i,j\}\cup \sigma}}\bra{\Psi_{\{i,j\}\cup \sigma}}\ket{0_{ALL}} \Bigg ) \Bigg\| \numberthis\\
	&\leq E_3(n, K, T, \epsilon_2, \epsilon,\Delta) + 16f(S,\eta-1,\Delta,\epsilon),
	\end{align*}
	
	where 
	
	\begin{align*}
	&E_3(n, K, T, \epsilon_2, \epsilon,\Delta) \equiv   O \left (2^{\Delta} (6\projerror) + 2^\Delta K \left (e(n)^{2T}+\epsilon_2 \right ) + \epsilon \right) 
	\end{align*}
\end{lemma}

\begin{proof}
The proof follows by two uses Lemma \ref{clm:breaktoproduct}, Lemma \ref{lem:lambdaapprox},  AND (unlike the previous two Lemmas) Lemma \ref{clm:expansiontrick}.  See Appendix \ref{appendix:lemma-proofs} for a full proof.
\end{proof}

Returning to where we left off in Equation \ref{eq:threesplits}, using all three of the above Lemmas, we have

\begin{equation}\label{eq:error-recursive-bound}
\begin{split}
f(S,\eta,\Delta,\epsilon) &\leq \left (2e(n)+2\projerror \right)^{\Delta}+ \Delta\left(E_1(n, K, T, \epsilon_2, \epsilon) +2f(S,\eta-1,\Delta,\epsilon)\right) \\
&+ \Delta^2\left(E_2(n, K, T, \epsilon_2, \epsilon) +2f(S,\eta-1,\Delta,\epsilon)\right)+\Delta^2\left(E_3(n, K, T, \epsilon_2, \epsilon,\Delta) +16 f(S,\eta-1,\Delta,\epsilon)\right) \\
&\leq \left (2e(n)+2\projerror \right)^{\Delta}+ 3 \Delta^2 E_3(n, K, T, \epsilon_2, \epsilon,\Delta) + 20\Delta^2f(S,\eta-1,\Delta,\epsilon) ,
\end{split}
\end{equation}

where the final inequality follows because $E_3(n, K, T, \epsilon_2, \epsilon,\Delta) \geq E_2(n, K, T, \epsilon_2, \epsilon) \geq E_1(n, K, T, \epsilon_2, \epsilon)$.

\mnote{We note that, during the course of Algorithm \ref{alg:constant-width-assumption2}, any time that we may wish to apply this recursion, we always have $\ell \geq w_0$.  This is because, in the definition of Algorithm \ref{alg:constant-width-assumption2}, if $\ell < w_0 = 10d(\Delta + h(n))$ the algorithm the calls subroutine $\mathcal{B}(S, \epsilon)$ and computes the final desired quantity with error $\epsilon$.  Thus, whenever $\ell < w_0$ for a given synthesis $S$, we have }

\begin{figure}[h!]
	\centering
	\includegraphics[scale=0.9]{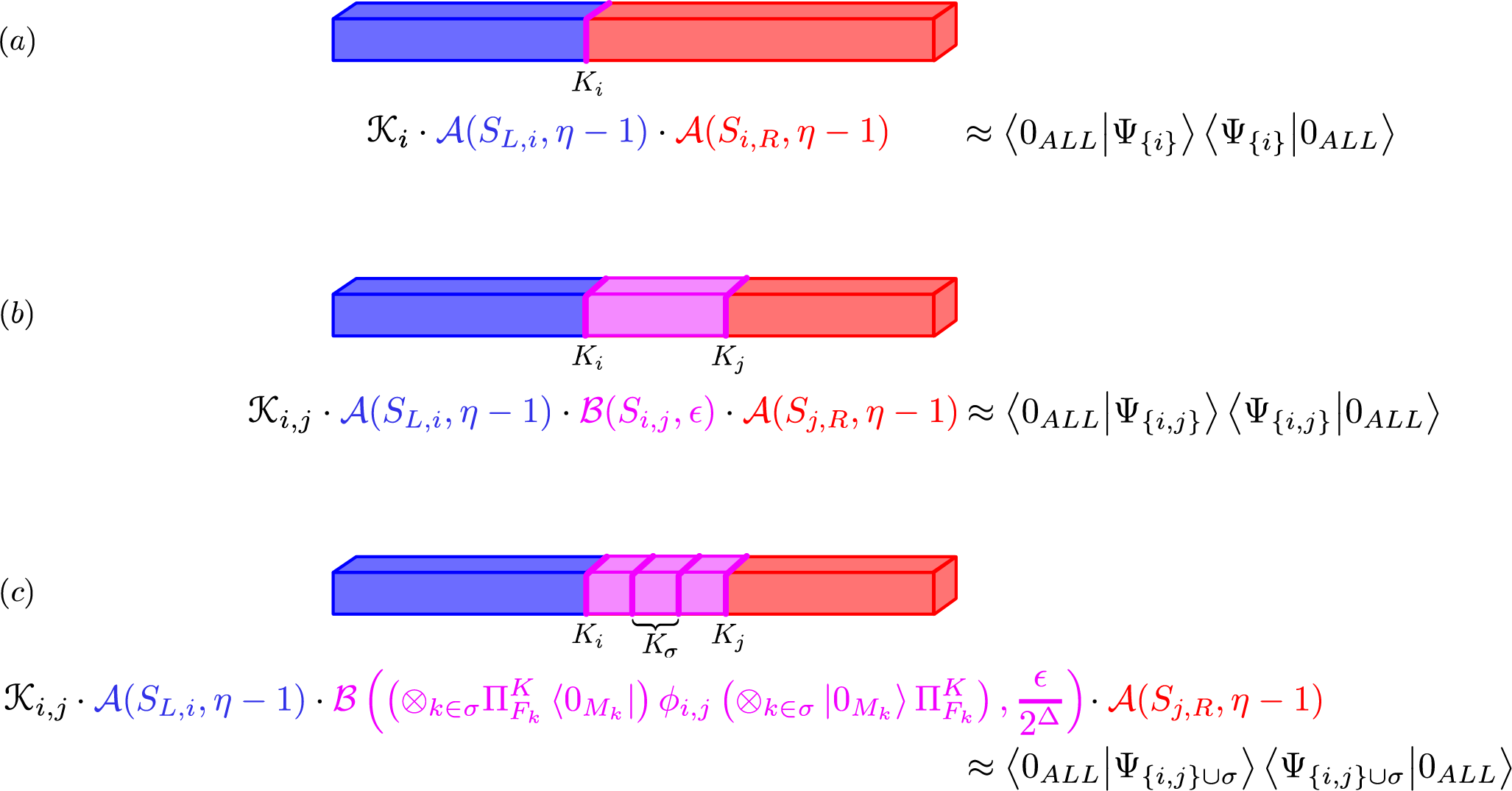}
	\caption{Depiction of the types of terms that appear in Equation \ref{eq:threesplits}: (a) those involving exactly one cut, $\{i\}$, (b) those involving exactly two cuts, $\{i,j\}$, and (c) those involved three or more of the $\Delta$ cuts, $\{i,j\}\cup\sigma$, for $\sigma\in\mathcal{P}(\{ i+1,\dots,j-1 \})\setminus\emptyset$. For brevity we have denoted $\mathcal{K}_i\equiv \frac{1}{(\kappa^{i}_{T, \epsilon_2})^{4K+1}}$ and $\mathcal{K}_{i,j}\equiv \frac{1}{(\kappa^{i}_{T, \epsilon_2}\kappa^{j}_{T, \epsilon_2})^{4K+1}}$.} \label{fig:error_terms}
\end{figure}

We note that, when $\eta = 0$ we have $f(S, \eta, \Delta, \epsilon) = f(S, 0, \Delta, \epsilon) \leq \epsilon$.  This is because, in the definition of Algorithm \ref{alg:constant-width-assumption2}, if $\eta < 1$ the algorithm the calls subroutine $\mathcal{B}(S, \epsilon)$ and computes the final desired quantity with error $\epsilon$.  This gives us the base case that we need to bound $f(S, \eta, \Delta, \epsilon)$ via standard recursive analysis:

\begin{lemma}\label{lem:error-bound-even-spaced}
	The error function $f(S, \eta, \Delta, \epsilon)$ obeys the following bound:

	\begin{equation}\label{eq:upper-bound-even-spaced}
	\begin{split}
	f(S, \eta, \Delta, \epsilon) &\leq 	\eta 20^{\eta}\Delta^{2\eta}\left(E_3(n, K, T, \epsilon_2, \epsilon\Delta) + (2e(n)+2\projerror)^{\Delta} \right)
	\end{split}
	\end{equation}

\end{lemma}

\begin{proof}
	The Lemma follows by using standard analysis of the recursion in Equation \ref{eq:error-recursive-bound}, and with the base case $f(S, 0, \Delta, \epsilon) \leq \epsilon \leq E_3(n, K, T, \epsilon_2, \epsilon)$.
\end{proof}

\subsection{ Run-Time Analysis for Algorithm  \ref{alg:constant-width-assumption2}}\label{subsection:quasi-run-time}

In this section we will derive a bound on the run-time for Algorithm \ref{alg:quasi-poly-subroutine}.  Recall that, given a synthesis $S$, the goal of Algorithm \ref{alg:quasi-poly-subroutine} is to compute the quantity $|\bra{0_N}\phi_S\ket{0_N}|^2$ where $\phi_S$ is the state synthesized by synthesis $S$, and $N$ is the active register, as defined in Definition \ref{def:synth}.  Recall that $\ell$ is defined, in Algorithm \ref{alg:quasi-poly-subroutine}, to be the width of the $N$ register for our input synthesis $S$.  We use $T(\ell)$ to denote the run-time bound for our algorithm on a synthesis with an $N$ register of width $\ell$.

 The main time cost for Algorithm \ref{alg:quasi-poly-subroutine} is accrued by the \textbf{return} line of the algorithm, which makes recursive calls to Algorithm \ref{alg:quasi-poly-subroutine} on a variety of smaller subproblems, as well as calls to the base case algorithm $\mathcal{B}$, and computing the $\kappa^i_{T, \epsilon_2}$ quantities.  All of the steps performed in Algorithm \ref{alg:quasi-poly-subroutine} before the \textbf{return} line (selecting the region $Z$, etc), can easily be done in (lesser) polynomial time, and their total cost will be absorbed into the term $\mu(\cdot)$ in our calculations below. 

Note that the \textbf{return} line of Algorithm \ref{alg:quasi-poly-subroutine} makes $2 \Delta$ distinct recursive calls to itself, which each only need to be computed once, even though they are reused multiple times in Equations \ref{eq:onecutterms}, \ref{eq:twocutterms}, and \ref{eq:basecaseinsum}.  Furthermore, the width of the active register of the synthesis input to each of the recursive calls to Algorithm \ref{alg:quasi-poly-subroutine} is at most $\frac{\ell+|Z|}{2} \leq \frac{3}{4} \ell$ (where the inequality follows because, in the context of Algorithm 2, the relationship $|Z| \leq \frac{\ell}{2}$ is enforced by lines \ref{ln:alg2stoppingwidth} and \ref{ln:alg2defz}).     Therefore, we know that our recursive time analysis will have the form:

\[T(\ell) < 2 \Delta T \left(\frac{3}{4} \ell \right ) + \zeta_1(n) \]

Where $\zeta_1(n)$ absorbs the time cost of all steps in Algorithm \ref{alg:quasi-poly-subroutine} other than the $2 \Delta$ recursive calls.  We will now explicitly bound the term $\zeta_1(n)$ in order to obtain a complete runtime bound.  The main time cost in $\zeta_1(n)$ comes from executing the calls to the base case algorithm $\mathcal{B}$ in lines  \ref{eq:twocutterms}, and \ref{eq:basecaseinsum} of Algorithm \ref{alg:quasi-poly-subroutine}, as well as computing the values  $\kappa^i_{T, \epsilon_2}$ that appear throughout the \textbf{return} line of  Algorithm \ref{alg:quasi-poly-subroutine}.  We will analyze the time cost of these two tasks separately in the two corresponding paragraphs below.  The only remaining time cost then comes from the steps performed in Algorithm \ref{alg:quasi-poly-subroutine} before the \textbf{return} line (selecting the region $Z$, etc),  and can easily be done with (less significant) polynomial time cost, which we will denote by $\mu(n)$.

\paragraph{Uses of $\mathcal{B}$ in the \textbf{return} line of Algorithm \ref{alg:quasi-poly-subroutine}: }There are at most $2\Delta^2$ calls to the base case algorithm $\mathcal{B}$ in the \textbf{return} line of Algorithm \ref{alg:quasi-poly-subroutine}, which all occur in Equations \ref{eq:twocutterms} and \ref{eq:basecaseinsum} of the \textbf{return} line.  The terms in Equation \ref{eq:twocutterms} have the form $\mathcal{B}(S_{i,j},\epsilon)$, and the terms in Equation \ref{eq:basecaseinsum} have the form $\mathcal{B}\left ( \left (\otimes_{k \in \sigma} \Pi^K_{F_k} \bra{0_{M_k}} \right )\phi_{i,j}\left (\otimes_{k \in \sigma} \ket{0_{M_k}}\Pi^K_{F_k}\right ),\frac{\epsilon}{2^\Delta} \right)$.  By Remark \ref{rm:3Dbasecase} we know that, for any 3D geometrically-local, depth-$D$ synthesis $S$ (on n qubits, total) we can use the algorithm from Theorem 5 of \cite{BGM19} to compute the quantity $\mathcal{B}(S,\epsilon) = |\bra{0_{ALL}} \phi_S \ket{0_{ALL}}|^2 \pm \epsilon $ in time $n\epsilon^{-2}2^{O(D^2\cdot w)}$, where $w$ is the width of the ``active register", $N$, of the synthesis $S$ in the third dimension (see Definition \ref{def:synth} for the definitions of syntheses $S$, and the register $N$).  Using  Remark \ref{rm:3Dbasecase} in this way, we see that the quantities $\mathcal{B}(S_{i,j},\epsilon) = |\bra{0_{ALL}} \phi_{i,j} \ket{0_{ALL}}|^2 \pm \epsilon $ from Equation \ref{eq:twocutterms} can each be computed in time $n\epsilon^{-2}2^{O((dK^2)^2\cdot 10d(\Delta+h(n)+2) )}$, because the synthesis $S_{i,j}$ has depth at most $O(dK^2)$ \footnote{The reason that the depth of the synthesis $S_{i,j}$ is bounded by $O(dK^2)$ is that $S_{i,j}$ is constructed using the original depth-$d$ circuit $C$, acted on by a number of operators $\Pi^K_{F_k}$.  However, it is important to note that each $\Pi^K_{F_k}$ is always acts in tensor product with $\Pi^K_{F_j}$ on any other cut $K_j$ for $j \neq k$.  Also, each cut is only acted on at most once this way.  It follows that these additive depths do not pile up during the course of Algorithm \ref{alg:quasi-poly-subroutine}, and thus the total depth never exceeds $O(dK^2)$. }, and has an active region of width at most $10d(\Delta+h(n)+2)$ in the third dimension (this width is enforced by line \ref{ln:alg2defz} of Algorithm \ref{alg:quasi-poly-subroutine}, where $S_{i,j}$ is explicitly specified in conjunction with Definition \ref{def:subsynth}). 

Similarly, using  Remark \ref{rm:3Dbasecase} in the same way, we see that the quantities:

\begin{align}
\mathcal{B}\left ( \left (\otimes_{k \in \sigma} \Pi^K_{F_k} \bra{0_{M_k}} \right )\phi_{i,j}\left (\otimes_{k \in \sigma} \ket{0_{M_k}}\Pi^K_{F_k}\right ),\frac{\epsilon}{2^\Delta} \right) = | \bra{0_{ALL}}\left (\otimes_{k \in \sigma} \Pi^K_{F_k} \bra{0_{M_k}} \right )\phi_{i,j}\left (\otimes_{k \in \sigma} \ket{0_{M_k}}\Pi^K_{F_k}\right )\ket{0_{ALL}} |^2 \pm \frac{\epsilon}{2^\Delta}, \nonumber
\end{align}

can be computed in time $n2^{2\Delta}\epsilon^{-2}2^{O((dK^2+2dK^2)^2\cdot 10d(\Delta+h(n)+2) )}$, because the synthesis $\left (\otimes_{k \in \sigma} \Pi^K_{F_k} \bra{0_{M_k}} \right )\phi_{i,j}\left (\otimes_{k \in \sigma} \ket{0_{M_k}}\Pi^K_{F_k}\right )$ has an active region that is a subset of the active region of the synthesis $S_{i,j}$ for $\phi_{i,j}$, and therefore still has width as most $10d(\Delta+h(n)+2) )$ in the third dimension.  Additionally, the depth of the synthesis $\left (\otimes_{k \in \sigma} \Pi^K_{F_k} \bra{0_{M_k}} \right )\phi_{i,j}\left (\otimes_{k \in \sigma} \ket{0_{M_k}}\Pi^K_{F_k}\right )$ is at most $O(2dK^2)$ higher than the depth of the synthesis $S_{i,j}$ for $\phi_{i,j}$, because the depth of the operator $\otimes_{k \in \sigma} \Pi^K_{F_k}$ is at most the maximum depth of the operator $\Pi^K_{F_k}$ for any $k$, which is at most $O(dK^2)$, by Lemma \ref{lem:powerblockencoding} and Definitions \ref{def:schmidtproj} and \ref{def:wrappedpi}.  Note that we are making a slight abuse of notation in that we have always referred to the state $\left (\otimes_{k \in \sigma} \Pi^K_{F_k} \bra{0_{M_k}} \right )\phi_{i,j}\left (\otimes_{k \in \sigma} \ket{0_{M_k}}\Pi^K_{F_k}\right )$ as a synthesis, since it's first appearance in the \textbf{return} line of Algorithm \ref{alg:quasi-poly-subroutine}, even though we never wrote out the explicit synthesis structure for this state according to Definition \ref{def:synth}.  This was done in order to keep the statement of Algorithm \ref{alg:quasi-poly-subroutine} and the subsequent discussion succinct.  We leave to the reader the exercise of expressing the state $\left (\otimes_{k \in \sigma} \Pi^K_{F_k} \bra{0_{M_k}} \right )\phi_{i,j}\left (\otimes_{k \in \sigma} \ket{0_{M_k}}\Pi^K_{F_k}\right )$ as a synthesis according to Definition \ref{def:synth}, which is accomplished by a straightforward modification of the (already defined) synthesis  $S_{i,j}$ for $\phi_{i,j}$, in order to include the operators and projections $\left (\otimes_{k \in \sigma} \Pi^K_{F_k} \bra{0_{M_k}} \right )$ into the synthesis (this is something that the definition of a synthesis was designed to do in the first place).  

So, we are now ready to explicitly include, in our recursive time analysis, the time cost for the use of the base case algorithm $\mathcal{B}$ in the \textbf{return} line of Algorithm \ref{alg:quasi-poly-subroutine}.  We now have:

\[T(\ell) < 2 \Delta T \left(\frac{3}{4} \ell \right ) + \Delta^2 n\epsilon^{-2}2^{O((dK^2)^2\cdot 10d(\Delta+h(n)+2) )} + \Delta^2 n2^{2\Delta}\epsilon^{-2}2^{O((dK^2+2dK^2)^2\cdot 10d(\Delta+h(n)+2) )} +  \zeta_2(n) + \mu(n) \]

Where $\zeta_2(n)$ denotes the time cost for computing the quantities $\kappa^i_{T, \epsilon_2}$ in the \textbf{return} line of Algorithm \ref{alg:quasi-poly-subroutine} (discussed below), and $\mu(n)$ denotes the (less significant) polynomial time cost that comes from the steps performed in Algorithm \ref{alg:quasi-poly-subroutine} before the \textbf{return} line (selecting the region $Z$, etc), as previously discussed. 

Noting that, in Algorithm \ref{alg:quasi-poly-subroutine}, the parameters are set as $\Delta = \log(n)$, $K = \log^3(n)$, $h(n) = \log^7(n)$, $\epsilon = \delta 2^{-10\log(n)\log(\log(n))}$ we have:

\begin{align}
T(\ell) < 2 \Delta T \left(\frac{3}{4} \ell \right ) + \delta^{-2} 2^{d^3 \polylog(n)}+  \zeta_2(n) + \mu(n) \nonumber
\end{align}

We now bound the remaining time cost, $\zeta_2(n)$, in the following paragraph.

\paragraph{Computation of quantities $\kappa^i_{T, \epsilon_2}$ in the \textbf{return} line of Algorithm \ref{alg:quasi-poly-subroutine}: }

The task of computing the $\kappa^i_{T, \epsilon_2}$ values in the \textbf{return} line of Algorithm \ref{alg:quasi-poly-subroutine} requires using the base case algorithm $\mathcal{B}$ as specified and discussed in Definition \ref{def:kappa} (via the usual prescription in Remark \ref{rm:3Dbasecase}).  In particular, according to Definition \ref{def:kappa}, we have that:

	\begin{align}
	\kappa^{i}_{T, \epsilon_2} \equiv \frac{\mathcal{B}(\Lambda_1^{j,T},\epsilon_2)}{\mathcal{B}(Z_j^{2T},\epsilon_2)} = \frac{ \tr(\rho_{B_j}^T\ket{\psi} \bra{\psi}_{B_j \cup M_j \cup F_j}\rho_{B_j}^T ) \pm \epsilon_2}{\tr(\rho_{B_j}^{2T} )\pm \epsilon_2}\nonumber
	\end{align}
	
	where the syntheses $\Lambda_1^{j,T}$, and $Z_j^{2T}$ are explicitly specified as 
	
	\begin{align}
	\Lambda_1^{j,T} \equiv & \left (\Gamma_{\Proj_{B_j}^T}  \circ \Gamma_{B_j \cup M_j \cup F_j}, (B_j  \cup F_j), (M_j \cup M_j' \cup F_j' \cup B^1_j,... \cup B^T_j ) \right . \nonumber\\
	&\left . ,( B_j \cup M_j \cup F_j \cup M_j' \cup F_j' \cup B^1_j,... \cup B^T_j) \right), \nonumber
	\end{align}
	
	\begin{align}
	Z_j^{T} \equiv & \left (\Gamma_{\Proj_{B_j}^T} , (B_j ), ( M_j' \cup F_j' \cup B^1_j,... \cup B^T_j ) \right . \nonumber\\
	&\left . ,( B_j  \cup M_j' \cup F_j' \cup B^1_j,... \cup B^T_j) \right), \nonumber
	\end{align}

	Note that this is a different way of using syntheses that in other parts of Algorithm \ref{alg:quasi-poly-subroutine} because $\Lambda_1^{j,T}$, and $Z_j^{2T}$ are scalars rather density matrices ( see Definition \ref{def:synth} to understand why).  These scalars can, nonetheless, be described, as above, by 3D geometrically-local, shallow depth syntheses as described above. In fact,
	
	\begin{align} 
	&Z_j^{T} = \tr(\rho_{B_j}^T ), \nonumber\\
	& \text{ and }\\
	&		 \Lambda_1^{j,T} = 	\tr(\rho_{B_j}^T\ket{\psi} \bra{\psi}_{B_j \cup M_j \cup F_j}\rho_{B_j}^T ) 
	\end{align}
	
	where
	
	\begin{align} 
	\rho_{B_i}^K \equiv \nonumber  \bra{0^{F_i, M_i, B_i^1,...B_i^K}}\prod_{j=1}^K(C^{\dagger}_{B_i \cup M_i \cup F_i} \otimes I_{B_i^1,...,B_i^K})(I_{F_i \cup M_i }\otimes \text{SWAP}_{B_i B_i^j}) (C_{B_i \cup M_i \cup F_i} \otimes I_{B_i^1, ..., B_i^K})\ket{0^{F_i, M_i, B_i^1,...B_i^K}} 
	\end{align}
	
	as in Lemma \ref{lem:powerblockencoding}.  
	
	So, applying Theorem 5 of \cite{BGM19} according to Remark \ref{rm:3Dbasecase}, as usual, we see that the time cost of computing $\mathcal{B}(\Lambda_1^{j,T},\epsilon_2) = \tr(\rho_{B_j}^T\ket{\psi} \bra{\psi}_{B_j \cup M_j \cup F_j}\rho_{B_j}^T ) \pm \epsilon_2$ is at most $n \epsilon_2^{-2}2^{O(D^2)w} = n \epsilon_2^{-2}2^{O(d^2T^2K^4)30d}$ because the synthesis $\Lambda_1^{j,T}$ has depth $O(dTK^2)$, and width at most $30d$ (the width of the slice $K_j$) in the third dimension.  The same depth and width bounds apply to the synthesis $Z_j^{2T}$, and so, the time cost of computing $\mathcal{B}(Z_j^{2T},\epsilon_2) = \tr(\rho_{B_j}^{2T} )\pm \epsilon_2$ is also bounded by $n \epsilon_2^{-2}2^{O(d^2T^2K^4)30d}$.  So, the total time cost of computing $\kappa^i_{T, \epsilon_2}$ according to Definition \ref{def:kappa} is at most $2n \epsilon_2^{-2}2^{O(d^2T^2K^4)30d} + \poly(n)$.  Since there are $\Delta$ distinct $\kappa^i_{T, \epsilon_2}$ values appearing in the \textbf{return} line of Algorithm \ref{alg:quasi-poly-subroutine}, the total time $\zeta_2(n)$, for computing all of the  $\kappa^i_{T, \epsilon_2}$ values, is $2\Delta n \epsilon_2^{-2}2^{O(d^2T^2K^4)30d} + \Delta \poly(n)$.  We an now update our recursive time analysis for Algorithm \ref{alg:quasi-poly-subroutine} as follows:

\begin{align}
&T(\ell) < 2 \Delta T \left(\frac{3}{4} \ell \right ) + \delta^{-2} 2^{d^3 \polylog(n)}+  \zeta_2(n) + \mu(n) \\
& = 2 \Delta T \left(\frac{3}{4} \ell \right ) + \delta^{-2} 2^{d^3 \polylog(n)}+  2\Delta n \epsilon_2^{-2}2^{O(d^2T^2K^4)30d} + \Delta \poly(n) + \mu(n)
\end{align} 

Recalling that, in Algorithm \ref{alg:quasi-poly-subroutine}, we specify parameter scalings $\Delta = \log(n)$, $K = T = \log^3(n)$, $\epsilon_2 = \delta 2^{-10\log(n)\log(\log(n))}$, this gives 

\begin{align}
&T(\ell) <  2 \Delta T \left(\frac{3}{4} \ell \right ) + \delta^{-2} 2^{d^3 \polylog(n)}+ \delta^{-2} 2^{d^3 \polylog(n)} + \mu(n)\\
& = 2 \Delta T \left(\frac{3}{4} \ell \right ) + \delta^{-2} 2^{d^3 \polylog(n)} + \poly(n)\\
&=2 \Delta T \left(\frac{3}{4} \ell \right ) + \delta^{-2} 2^{d^3 \polylog(n)} \label{eq:recursiveruntimebound}
\end{align}

Note that Equation \ref{eq:recursiveruntimebound} is a recursive run-time whereby, at each level, we have at most $2\Delta$ subproblems, each with size at most $\frac{3}{4}$ of the original problem. This is a common formula, and we can use the Master Theorem for divide-and-conquer algorithms to determine an upper bound for our run-time as
\begin{align}
T(n^{1/3}) &< (2\Delta)^\eta\cdot T\left( \left( \frac{3}{4} \right)^{\eta} n^{\frac{1}{3}} \right) +\sum_{i=0}^\eta (2\Delta)^i\delta^{-2} 2^{d^3 \polylog(n)}\\
&< (2\Delta)^\eta\cdot T\left( \left( \frac{3}{4} \right)^{\eta} n^{\frac{1}{3}} \right) + \frac{(2\Delta)^{\eta+1}-1}{(2\Delta)-1}\cdot\delta^{-2} 2^{d^3 \polylog(n)}  \\
&< (2\Delta)^\eta\cdot T\left( \left( \frac{3}{4} \right)^{\eta} n^{\frac{1}{3}} \right) + (2\Delta)^{\eta+1}\cdot\delta^{-2} 2^{d^3 \polylog(n)} \\
T(n^{1/3}) &< (2\Delta)^\eta \left[ T\left( \left( \frac{3}{4} \right)^{\eta} n^{\frac{1}{3}} \right) + 2 \Delta \delta^{-2} 2^{d^3 \polylog(n)} \right] \label{eq:quasi-run-time-master-bound}
\end{align}
where $\eta$ is the depth of our recursive calls. Note that the reason we start our recursion at $n^{1/3}$ instead of at $n$ is because of the technical definition of $T(\ell)$.  Recall that, we use $T(\ell)$ to denote the run-time bound for our algorithm on a synthesis with an $N$ register of width $\ell$.  The starting point of our recursion is a cube of $n$ qubits, which has side length $N^{1/3}$ in each dimension, and this is the reason that the total runtime for the original problem is bounded by $T(n^{1/3})$.  

\begin{theorem}\label{thm:quasi-poly-runtime-bound}
	Suppose $\eta =\frac{\log(n)}{3\log(4/3)}$ and $\Delta=\log{n}$. Given these values, the run-time for Algorithm \ref{alg:quasi-poly-subroutine} will be bounded by
	\begin{equation}
		T(n)<\delta^{-2} 2^{d^3 \polylog(n)}
	\end{equation} 
\end{theorem}

\begin{proof}
	Theorem \ref{thm:quasi-poly-runtime-bound} follows directly from the above calculations.
	
\end{proof}

\vspace{1em}
\textbf{\Large Acknowledgments}

MC thanks David Gosset and Sergey Bravyi for helpful discussions.  Part of this work was completed while MC was attending the Simons Institute Quantum Wave in Computing workshop.

\newpage

\appendix
\appendixpage

\section{Proofs of Lemma Statements}\label{appendix:lemma-proofs}

\subsection{Statements from Section \ref{sec:splitheavy}}

\begin{lemma*}[Restatement of Lemma \ref{clm:findheavy}]
If $|\bra{0^{\otimes n}}C   \ket{0^{\otimes n}}| > |1/q(n)|$, then, for any $0 \leq h \leq 1$, $h|K|$ of the slices $K_i$ in $K$ have the property that:

\begin{equation}
p_{total}(M_i = 0) \geq (|1/q(n)|)^{\frac{1}{(1-h)|K|}} \label{eq:heavyslice2}
\end{equation}	

  We will let $K_{heavy}$ be the subset of $K$ consisting of those $K_i$ satisfying Equation \eqref{eq:heavyslice2}.
\end{lemma*}

\begin{proof}
	Using \ref{clm:independence} we have that:
	
	\begin{align}
	p_{total}(M_i = 0 \text{ } \forall i) = \prod_i p_{total}(M_i = 0) \geq |\bra{0^{\otimes n}}C   \ket{0^{\otimes n}}| > |1/q(n)|
	\end{align}
	
	So,

	\begin{align}
  \log(\prod_i p_{total}(M_i = 0))  = \sum_i \log(p_{total}(M_i = 0)) \geq  \log( |1/q(n)|)
	\end{align}
	
	Since every term on both sides of the equation is negative, it follows that at least $h |K|$ of the slices $K_i$ in $K$ must satisfy $\log(p_{total}(M_i = 0)) \geq \frac{1}{(1-h)|K|} \log( |1/q(n)|)$.
	
	So, at least $h |K|$ of the slices $K_i$ in $K$ must satisfy 
	
	\[p_{total}(M_i = 0) \geq \exp{ \frac{1}{(1-h)|K|}\log( |1/q(n)|)} = (|1/q(n)|)^{\frac{1}{(1-h)|K|}}   \]

\end{proof}

\begin{lemma*}[Restatement of Lemma \ref{lem:highschmidtnew}]
		For any slice $K_i \in K_{heavy}$ satisfying:
		
				\begin{equation}\label{eq:exactweight}
				p_{total}(M_i = 0) \geq 1 - e(n),
				\end{equation}	
		 
		the top Schmidt coefficient of $\ket{\psi}_{B_i \cup F_i}$ satisfies $\lambda_1^i \geq 1 - O(e(n))$.  (Where the Schmidt decomposition is taken across the partition $B_i, F_i$.)
	\end{lemma*}
	
	\begin{proof}

		For any $K_i \in K_{heavy}$, recall that, by definition, the width of $M_i$ is chosen large enough that $B_i$ and $F_i$ do not have any intersecting light cones (so the two halves of the circuit are lightcone separated).  It follows that,
		
		\begin{align}
		&\tr_{M_i}(C_{B_i\cup M_i \cup F_i }\ket{0_{B_i\cup M_i \cup F_i }}\bra{0_{B_i\cup M_i \cup F_i }}C^{\dagger}_{B_i\cup M_i \cup F_i })=\nonumber\\
		&  \tr_{M_i \cup F_i }(C_{B_i\cup M_i \cup F_i }\ket{0_{B_i\cup M_i \cup F_i }}\bra{0_{B_i\cup M_i \cup F_i }}C^{\dagger}_{B_i\cup M_i \cup F_i }) \otimes \tr_{M_i \cup B_i }(C_{B_i\cup M_i \cup F_i }\ket{0_{B_i\cup M_i \cup F_i }}\bra{0_{B_i\cup M_i \cup F_i }}C^{\dagger}_{B_i\cup M_i \cup F_i }) \nonumber \label{eq:lightconeproduct}
		\end{align}
		
		But, from Equation \ref{eq:exactweight}, which is an assumption of the Lemma, we see that,
		
			\begin{align}
			&\tr \left (  \ket{0_{M_i}}\bra{0_{M_i}}C_{B_i\cup M_i \cup F_i }\ket{0_{B_i\cup M_i \cup F_i }}\bra{0_{B_i\cup M_i \cup F_i }}C^{\dagger}_{B_i\cup M_i \cup F_i }\ket{0_{M_i}}\bra{0_{M_i}} \right )  \\ & =\tr ( \ket{\psi}\bra{\psi}_{B_i \cup F_i} ) =  p_{total}(M_i = 0)  \geq 1 - e(n), \label{eq:highweighteq} 
			\end{align}
		
		and so,
		\begin{align}
		&\left \| C_{B_i\cup M_i \cup F_i }\ket{0_{B_i\cup M_i \cup F_i }}\bra{0_{B_i\cup M_i \cup F_i }}C^{\dagger}_{B_i\cup M_i \cup F_i } \right. \nonumber\\
		& \left .- \ket{0_{M_i}}\bra{0_{M_i}}C_{B_i\cup M_i \cup F_i }\ket{0_{B_i\cup M_i \cup F_i }}\bra{0_{B_i\cup M_i \cup F_i }}C^{\dagger}_{B_i\cup M_i \cup F_i }\ket{0_{M_i}}\bra{0_{M_i}} \right \| \nonumber \\ 
		&\leq 1- p_{total}(M_i = 0) \nonumber \leq e(n) 
		\end{align}

		
		So, 
		
		\begin{align}
		\ket{\psi}\bra{\psi}_{B_i \cup F_i} & = \tr_{M_i}(\ket{0_{M_i}}\bra{0_{M_i}}C_{B_i\cup M_i \cup F_i }\ket{0_{B_i\cup M_i \cup F_i }}\bra{0_{B_i\cup M_i \cup F_i }}C^{\dagger}_{B_i\cup M_i \cup F_i }\ket{0_{M_i}}\bra{0_{M_i}})\nonumber \\
		&= \tr_{M_i}(C_{B_i\cup M_i \cup F_i }\ket{0_{B_i\cup M_i \cup F_i }}\bra{0_{B_i\cup M_i \cup F_i }}C^{\dagger}_{B_i\cup M_i \cup F_i }) +O(e(n)) \nonumber\\
		&= \tr_{M_i \cup F_i }(C_{B_i\cup M_i \cup F_i }\ket{0_{B_i\cup M_i \cup F_i }}\bra{0_{B_i\cup M_i \cup F_i }}C^{\dagger}_{B_i\cup M_i \cup F_i }) \nonumber\\
		& \otimes \tr_{M_i \cup B_i }(C_{B_i\cup M_i \cup F_i }\ket{0_{B_i\cup M_i \cup F_i }}\bra{0_{B_i\cup M_i \cup F_i }}C^{\dagger}_{B_i\cup M_i \cup F_i })+O(e(n)) \nonumber\\
		& = \tr_{M_i \cup F_i }(\ket{0_{M_i}}\bra{0_{M_i}}C_{B_i\cup M_i \cup F_i }\ket{0_{B_i\cup M_i \cup F_i }}\bra{0_{B_i\cup M_i \cup F_i }}C^{\dagger}_{B_i\cup M_i \cup F_i }\ket{0_{M_i}}\bra{0_{M_i}}) \nonumber\\
		& \otimes \tr_{M_i \cup B_i }(\ket{0_{M_i}}\bra{0_{M_i}}C_{B_i\cup M_i \cup F_i }\ket{0_{B_i\cup M_i \cup F_i }}\bra{0_{B_i\cup M_i \cup F_i }}C^{\dagger}_{B_i\cup M_i \cup F_i }\ket{0_{M_i}}\bra{0_{M_i}})+O(e(n)) \nonumber\\
		&=  \tr_{F_i}(\ket{\psi}\bra{\psi}_{B_i \cup F_i}) \otimes \tr_{B_i}(\ket{\psi}\bra{\psi}_{B_i \cup F_i})+O(e(n))  \label{eq:highweightprod}
		\end{align}

		Now, by definition, the largest Schmidt coefficient $\lambda^i_1$ of $\ket{\psi}_{B_i \cup F_i}$ is equal to the largest eigenvalue of the mixed state $\tr_{F_i}(\ket{\psi}\bra{\psi}_{B_i \cup F_i})$, which is equivalent to the largest eigenvalue of the mixed state $\tr_{B_i}(\ket{\psi}\bra{\psi}_{B_i \cup F_i})$  (since $\ket{\psi}\bra{\psi}_{B_i \cup F_i}$ is an unormalized pure state).

		 For notational brevity we define $\rho_{B_i} \equiv \tr_{F_i}(\ket{\psi}\bra{\psi}_{B_i \cup F_i})$ and $\rho_{F_i} \equiv \tr_{B_i}(\ket{\psi}\bra{\psi}_{B_i \cup F_i})$.  By Holder's Inequality (with Holder parameters set top $p=1$ and $q = \infty$) we have that:
		
		\begin{align}
		&\| \rho_{B_i}^2 \|_1 \leq \| \rho_{B_i} \|_1 \| \rho_{B_i} \|_{\infty} = \| \rho_{B_i} \|_{\infty} = \lambda^i_1, \nonumber\\
		& \text{and,} \nonumber\\
		&\| \rho_{F_i}^2 \|_1 \leq \| \rho_{F_i} \|_1 \| \rho_{F_i} \|_{\infty} = \| \rho_{F_i} \|_{\infty} = \lambda^i_1, \nonumber
		\end{align}
		
		where the second to last inequality follows because $\| \rho_{B_i} \|_1 = \tr( \rho_{B_i} ) \leq 1$ (resp. $\| \rho_{F_i} \|_1 = \tr( \rho_{F_i} ) \leq 1$), and the last equality follows by the definition of $\lambda^i_1$.\mnote{Try to explain the use of norms in Holder's Inequality better}	\mnote{perhaps could include a second proof of the above equation that explicitly writes out the eigenvalues rather than just citing Holder's inequality?}  So, we have:
		
		\begin{align}
		&(\lambda^i_1)^2 \geq \| \rho_{B_i}^2 \|_1\| \rho_{F_i}^2 \|_1 = \tr(\rho_{B_i}^2) \tr(\rho_{F_i}^2) = \tr(\rho_{B_i}^2 \otimes \rho_{F_i}^2 ) = \tr(\left ( \rho_{B_i} \otimes \rho_{F_i}  \right )^2) \nonumber \\
		& = \tr( \left (\ket{\psi}\bra{\psi}_{B_i \cup F_i} \right )^2 ) + O(e(n))\geq 1 - O(e(n)),
		\end{align}
		
		Where the first three equalities follow by definition, the fourth equality follows by two uses of Equation \ref{eq:highweightprod}  (and the fact that $\| \ket{\psi}\bra{\psi}_{B_i \cup F_i}  \|_1 \leq 1$), and the final inequality follows by Equation \ref{eq:highweighteq}. It follows that:
		
		\begin{align}
		&\lambda^i_1 \geq  1 - O(e(n))
		\end{align}

	\end{proof}

	\begin{lemma*}[Restatement of Lemma \ref{clm:schmidtproj}]
	For any $K_i \in K_{heavy}$,
	
	\begin{align}
		&\|  \Proj_{F_i}^K -\ket{w_1}\bra{w_1}_{F_i}  \|_1 \leq \projerror \label{eq:Fprojclose1} \\
		& \text{and}  \nonumber\\
		&\|  \Proj_{B_i}^K -\ket{v_1}\bra{v_1}_{B_i}  \|_1 \leq \projerror \label{eq:Bprojclose1}
	\end{align}
	
	where $\projerror \equiv \left (\frac{1 - \lambda_1^i}{\lambda_1^i} \right )^K$, and $\ket{w_1}\bra{w_1}_{F_i}$, $\ket{v_1}\bra{v_1}_{B_i} $ are the projectors onto the top Schmidt vectors of $\ket{\psi}_{B_i \cup F_i}$ in $F_i$ and $B_i$ respectively.
\end{lemma*}

\begin{proof}
 Here we will write the proof for Equation \ref{eq:Fprojclose1}, but the proof for Equation \ref{eq:Bprojclose1} is exactly analogous.  In particular, by Lemma \ref{lem:powerblockencoding} we have that:
	
		For any constant integer $K > 0$, the following is a 2D-local circuit which gives a block encoding for $\rho_{F_i}^K$:
		
		\begin{align} 
		&\prod_{j=1}^K(C^{\dagger}_{B_i \cup M_i \cup F_i} \otimes I_{F^1_i,...,F^K_i})(I_{B_i \cup M_i }\otimes \text{SWAP}_{F_i F^j_i}) (C_{B_i \cup M_i \cup F_i} \otimes I_{F^1_i, ..., F^K_i}). \nonumber
		\end{align}
		
		It follows, by the definition of a block encoding, that,

		\begin{align} 
		&\rho_{F_i}^K = \bra{0^{B_i, M_i, F^1_i,...F^k_i}}\prod_{j=1}^K(C^{\dagger}_{B_i \cup M_i \cup F_i} \otimes I_{F^1_i,...,F^K_i})(I_{B_i \cup M_i }\otimes \text{SWAP}_{F_i F^j_i}) (C_{B_i \cup M_i \cup F_i} \otimes I_{F^1_i, ..., F^K_i})\ket{0^{B_i, M_i, F^1_i,...F^k_i}}  \nonumber
		\end{align} 
		
		Recall the definition of $\Proj_{F_i}^K$:
		\begin{align} 
		&\Proj_{F_i}^K \equiv \nonumber  \frac{1}{(\lambda_1^i)^K}\bra{0^{B_i, M_i, F^1_i,...F^k_i}}\prod_{j=1}^K(C^{\dagger}_{B_i \cup M_i \cup F_i} \otimes I_{F^1_i,...,F^K_i})(I_{B_i \cup M_i }\otimes \text{SWAP}_{F_i F^j_i}) (C_{B_i \cup M_i \cup F_i} \otimes I_{F^1_i, ..., F^K_i})\ket{0^{B_i, M_i, F^1_i,...F^k_i}}. \nonumber 
		\end{align} 
		
		And thus,
		
		\begin{align} 
		&\Proj_{F_i}^K = \nonumber  \frac{1}{(\lambda_1^i)^K}\rho_{F_i}^K = \left (\frac{\rho_{F_i}}{\lambda_1^i} \right )^K. 
		\end{align} 
		\mnote{maybe should include the above as part of the definition of $\Proj_{F_i}^K$??}
		
		By the definition of $\lambda_1^i$ and the leading Schmidt coefficient we have that: 
		
		\begin{align} 
		&\frac{\rho_{F_i}}{\lambda_1^i}  = \ket{w_1}\bra{w_1}_{F_i}  + E , \nonumber 
		\end{align} 
		
		where $E \equiv (\frac{\rho_{F_i}}{\lambda_1^i}  - \ket{w_1}\bra{w_1}_{F_i})$ is a PSD operator with trace norm $\|E \| \leq \frac{1 - \lambda_1^i}{\lambda_1^i}$, and which is orthogonal to $\bra{w_1}_{F_i}$  (i.e. $\ket{w_1}\bra{w_1}_{F_i} \cdot E = 0$). \mnote{is orthogonal the right word?  Clean up explanation}  It follows that:
		
		\begin{align} 
		& \left (\frac{\rho_{F_i}}{\lambda_1^i} \right )^K = \left ( \ket{w_1}\bra{w_1}_{F_i}  + E \right )^K = \left ( \ket{w_1}\bra{w_1}_{F_i}  \right )^K  + E^K. 
		\end{align} 
		
		So,

		\begin{align} 
		&\left \|  \Proj_{F_i}^K - \ket{w_1}\bra{w_1}_{F_i} \right \| = \left \| \left (\frac{\rho_{F_i}}{\lambda_1^i} \right )^K - \ket{w_1}\bra{w_1}_{F_i} \right \| = \|E^K\| = \|E\|^K \leq \left (\frac{1 - \lambda_1^i}{\lambda_1^i} \right )^K 
		\end{align}

\end{proof}

\begin{lemma*} [Restatement of Lemma \ref{clm:expansiontrick}]
	 Consider a set $K_{heavy}$ of slices such that, for every $K_i \in K_{heavy}$, $\ket{\psi}_{B_i \cup F_i}$ satisfies $\lambda_1^i \geq 1 - e(n)$, and such that for any $K_i, K_j \in K_{heavy}$, the operators $\Pi^K_{F_i}$ and $\Pi^K_{F_j}$ are light-cone separated whenever $i \neq j$.  Then, for any set of $\Delta$ slices, $\{K_i\}_{i \in [\Delta]} \subseteq K_{heavy}$, we have that:
	
	\begin{align}
	&\left \| \sum_{\sigma\in\mathcal{P}([\Delta])} (-1)^{\abs{\sigma}} \ket{\Psi_\sigma}\bra{\Psi_\sigma}\right \| = \left \|	\ket{\Psi_\emptyset}\bra{\Psi_\emptyset} -  \sum_{\sigma\in\mathcal{P}([\Delta])\setminus\emptyset} (-1)^{\abs{\sigma}+1} \ket{\Psi_\sigma}\bra{\Psi_\sigma}\right \| \leq (2e(n)+2\projerror)^{\Delta}, 
	\end{align}
	where $\projerror \equiv \left (\frac{1 - \lambda_1^i}{\lambda_1^i} \right )^K$.

\end{lemma*}

\begin{proof}
	 For the following, we use the shorthand $\rho_M(U)$ for the density matrix of the state prepared by linear operator $U$ acting on the all $0$ state of the $M$ register. For instance, $\rho_{\all}(U)=U\ket{0_{ALL}}\bra{0_{ALL}}U^\dagger$. Note that $\rho_M(VU)=V\rho_M(U)V^\dagger$.
	
	Following the definition of $\ket{\Psi_\sigma}$ in Definition \ref{def:psi-states}, we are trying to upper bound the quantity
	\begin{align}
	& \left \| \sum_{\sigma\in\mathcal{P}([\Delta])} (-1)^{\abs{\sigma}} \ket{\Psi_\sigma}\bra{\Psi_\sigma}\right \| = \left \| \sum_{\sigma\in\mathcal{P}([\Delta])} (-1)^{\abs{\sigma}} \rho_{\all}\left((\otimes_{j \in \sigma }\Pi^K_{F_j})(\otimes_{i\in [\Delta]} \bra{0_{M_i}})C\right) \right \|.\label{eq:sum-psi-states}
	\end{align}
	
	The proof proceeds in three parts. First, we consider Equation \ref{eq:sum-psi-states} without the post-selection on $\otimes_{i\in [\Delta]} \bra{0_{M_i}}$, and give an equivalent formulation in terms of a product of similar quantities. Second, we show this formulation holds under the post-selection. And lastly, we bound each term in this product formulation.
	
	From the definition of $C$ we have
	\begin{align}
	&\left \| \sum_{\sigma\in\mathcal{P}([\Delta])} (-1)^{\abs{\sigma}} \rho_{\all}\left((\otimes_{j \in \sigma }\Pi^K_{F_j})C\right) \right \|= \\
	&\left \| \sum_{\sigma\in\mathcal{P}([\Delta])} (-1)^{\abs{\sigma}} \rho_{\all}\left((\otimes_{j \in \sigma }\Pi^K_{F_j})(C_{L, \sigma_1} \circ \otimes_{j \in [|\sigma|-1]}  C_{\sigma_j, \sigma_{j+1}} \circ C_{\sigma_{|\sigma|}, R} \circ \otimes_{j \in [|\sigma|]} C_{B_{\sigma_j} \cup M_{\sigma_j}  \cup F_{\sigma_j} })\right) \right \|,
	\end{align}
	and by the definitions of $\Pi_{F_j}^K$ and $C_{wrap_i}$ we have
	\begin{equation}
		=\left \| \sum_{\sigma\in\mathcal{P}([\Delta])} (-1)^{\abs{\sigma}} \rho_{\all}\left(  C_{L, \sigma_1} \circ \otimes_{j \in [|\sigma|-1]} C_{\sigma_j, \sigma_{j+1}} \circ C_{\sigma_{|\sigma|}, R} \circ \otimes_{j \in [|\sigma|]} P^K_{F_j}C_{B_{\sigma_j} \cup M_{\sigma_j}  \cup F_{\sigma_j}}  \right) \right \|.
	\end{equation}
	We can rewrite this by expanding the summation and regrouping terms in tensor product \nnote{maybe?}
	\begin{align}
		=&\left\|  (C_{L, 1} \circ \otimes_{j \in [\Delta-1]} C_{\sigma_j, \sigma_{j+1}} \circ C_{\sigma_{\Delta}, R})  \circ  \otimes_{j \in [\Delta]} \left( \rho_{B_j \cup M_j \cup F_j}(C_{B_j \cup M_j \cup F_j}) - \rho_{B_j \cup M_j \cup F_j}(P^K_{F_j}C_{B_j \cup M_j \cup F_j})\right) \right. \\
		&\left. \otimes  \rho_{\all \setminus \cup_{j \in [\Delta]}B_j \cup M_j \cup F_j}(I)\circ (C_{L, 1} \circ \otimes_{j \in [\Delta-1]} C_{\sigma_j, \sigma_{j+1}} \circ C_{\sigma_{\Delta}, R})^\dagger \right\|.
	\end{align}
	
	\mnote{We should think about how to explain the previous line.  Perhaps say that expanding the middle tensor product (the one with minus sign in it) results in every possible combination of slices or omitting slices, with corresponding negative signs.}Lastly, by standard properties of the trace norm and noting that $C_{L, 1} \circ \otimes_{j \in [\Delta-1]} C_{\sigma_j, \sigma_{j+1}} \circ C_{\sigma_{\Delta}, R}$ is a unitary operator, we have
	\begin{equation}
		\left \| \sum_{\sigma\in\mathcal{P}([\Delta])} (-1)^{\abs{\sigma}} \rho_{\all}\left((\otimes_{j \in \sigma }\Pi^K_{F_j})C\right) \right \|= \prod_{j\in[\Delta]}\left\| \rho_{B_j \cup M_j \cup F_j}(C_{B_j \cup M_j \cup F_j}) - \rho_{B_j \cup M_j \cup F_j}(P^K_{F_j}C_{B_j \cup M_j \cup F_j})    \right\|. \label{eq:prod-form-no-post}
	\end{equation}	
	
	Now, since the terms $\otimes_{i\in [\Delta]} \ket{0_{M_i}}$, $\otimes_{i\in [\Delta]} \bra{0_{M_i}}$ commute with the terms $\otimes_{j \in \sigma }\Pi^K_{F_j}$, $P^K_{F_j}$ for all $i, j$, and $\sigma$, the form of Equation \ref{eq:prod-form-no-post} holds even under post-selection:
	\begin{align}
		&\left \| \sum_{\sigma\in\mathcal{P}([\Delta])} (-1)^{\abs{\sigma}} \rho_{\all}\left((\otimes_{j \in \sigma }\Pi^K_{F_j}\bra{0_{M_j}})C\right) \right \|= \nonumber \\
		& \prod_{j\in[\Delta]}\left\| \rho_{B_j \cup M_j \cup F_j}\left(\bra{0_{M_j}}C_{B_j \cup M_j \cup F_j}\right) - \rho_{B_j \cup M_j \cup F_j}\left(P^K_{F_j}\bra{0_{M_j}}C_{B_j \cup M_j \cup F_j}\right)    \right\|. \label{eq:prod-form-with-post}
	\end{align}
	
	We now bound each term of this product. By adding and subtracting $\rho_{B_j \cup M_j \cup F_j}\left(\ket{w_1}\bra{w_1}_{F_j}\bra{0_{M_j}}C_{B_j \cup M_j \cup F_j}\right)$ (where $\ket{w_1}\bra{w_1}_{F_j}$ is the projector onto the top Schmidt vector of $\ket{\Psi}_{B_j\cup F_j}$) and using the triangle inequality we have
	\begin{align}
		&\left\| \rho_{B_j \cup M_j \cup F_j}\left(\bra{0_{M_j}}C_{B_j \cup M_j \cup F_j}\right) - \rho_{B_j \cup M_j \cup F_j}\left(P^K_{F_j}\bra{0_{M_j}}C_{B_j \cup M_j \cup F_j}\right)    \right\| \leq \\
		&\left\| \rho_{B_j \cup M_j \cup F_j}\left(\bra{0_{M_j}}C_{B_j \cup M_j \cup F_j}\right)  - \rho_{B_j \cup M_j \cup F_j}\left(\ket{w_1}\bra{w_1}_{F_j} \bra{0_{M_j}}C_{B_j \cup M_j \cup F_j}\right) \right\| + \nonumber \\
		&\left\| \rho_{B_j \cup M_j \cup F_j}\left(\ket{w_1}\bra{w_1}_{F_j}\bra{0_{M_j}}C_{B_j \cup M_j \cup F_j}\right)  - \rho_{B_j \cup M_j \cup F_j}\left(P^K_{F_j}\bra{0_{M_j}}C_{B_j \cup M_j \cup F_j}\right)    \right\|.
	\end{align}
	By using Lemma \ref{clm:schmidtproj} (twice) we can bound the right summand by $2\projerror \equiv 2 \left (\frac{1 - \lambda_1^i}{\lambda_1^i} \right )^K$. By assumption in the lemma statement we have that the top Schmidt coefficient of $\ket{\psi}_{B_j \cup F_j}$ satisfies $\lambda_1^j \geq 1 - e(n)$ (for every $j$), and so (applying this bound twice) the left summand is bounded by $2e(n)$. Thus, we have 
	\begin{equation}
		\left\| \rho_{B_j \cup M_j \cup F_j}\left(\bra{0_{M_j}}C_{B_j \cup M_j \cup F_j}\right) - \rho_{B_j \cup M_j \cup F_j}\left(P^K_{F_j}\bra{0_{M_j}}C_{B_j \cup M_j \cup F_j}\right)    \right\| \leq 2(e)+2f(n). \label{eq:prod-term-bound}
	\end{equation}
	Combining Equations \ref{eq:prod-form-with-post} with Equation \ref{eq:prod-term-bound} we have the desired
	\begin{equation}
		\left \| \sum_{\sigma\in\mathcal{P}([\Delta])} (-1)^{\abs{\sigma}} \rho_{\all}\left((\otimes_{j \in \sigma }\Pi^K_{F_j}\bra{0_{M_j}})C\right) \right \| \leq \prod_{j\in[\Delta]} (2(e)+2\projerror) = (2(e)+2\projerror)^\Delta.
	\end{equation}

\end{proof}

\begin{lemma*}[Restatement of Lemma \ref{clm:breaktoproduct}]
	For any $K_i \in K_{heavy}$ (recall this means that $\ket{\psi}_{B_i \cup F_i}$ satisfies $\lambda_1^i \geq 1 - e(n)$), the state $\ket{\Omega_i}\bra{\Omega_i}$ is within $6\projerror$ of an unnormalized product state about $M_i$, described as follows:
	\begin{align}
	&\left \| \ket{\Omega_i}\bra{\Omega_i}  -  1/\lambda_1^i \tr_{F_i}\left (  \ket{\Xi_{L_i}}\bra{\Xi_{L_i}} \right )  \otimes  \tr_{B_i}\left ( \ket{\Xi_{R_i}}\bra{\Xi_{R_i}} \right )\right \|\leq 6\projerror
	\end{align}

	Here $\projerror \equiv \left (\frac{1 - \lambda_1^i}{\lambda_1^i} \right )^K \leq \left (\frac{e(n)}{1-e(n)} \right )^K$ just as in Lemma \ref{clm:schmidtproj}.  Recall the definitions: 
	
	\begin{align*}
	\ket{\Omega_i} &\equiv \Pi^K_{F_i}\bra{0_{M_i}}C\ket{0_{\all}} \\
	\ket{\Xi_{L_i}}&\equiv \Proj^K_{F_i} \bra{0_{M_i}}C_{L_i}C_{B_i \cup M_i \cup F_i}  \ket{0_{L_i\cup B_i \cup M_i \cup F_i}}\\
	\ket{\Xi_{R_i}}&\equiv \Proj^K_{B_i} \bra{0_{M_i}}C_{R_i}C_{B_i \cup M_i \cup F_i}  \ket{0_{R_i\cup B_i \cup M_i \cup F_i}}
	\end{align*}
	
\end{lemma*}

	\begin{proof}
	The proof proceeds in two parts. First, we demonstrate that the quantity
	\begin{equation*}
		\ket{\Omega_i}\bra{\Omega_i}  -  1/\lambda_1^i \tr_{F_i}\left (  \ket{\Xi_{L_i}}\bra{\Xi_{L_i}} \right )  \otimes  \tr_{B_i}\left ( \ket{\Xi_{R_i}}\bra{\Xi_{R_i}} \right )
	\end{equation*}
	is equal to the sum of three error terms, $E$, $H_1$, and $H_2$. We then bound the trace norm of each of these quantities using applications of Lemma \ref{clm:schmidtproj}. By definition of $\ket{\Omega_i}$ we have
	\begin{equation}
		\ket{\Omega_i}\bra{\Omega_i} = \Pi^K_{F_i} \bra{0_{M_i}}C\ket{0_{\all}}\bra{0_{\all}}C^{\dagger}\ket{0_{M_i}}\Pi^K_{F_i}.
	\end{equation}
	Now, recall by Definition \ref{def:wrappedpi}, $\Pi^K_{F_i} \equiv C_{Wrap_i} \Proj^K_{F_i} C^{\dagger}_{Wrap_i}$.  So,
		
		  \begin{align*}
		  &\Pi^K_{F_i} \bra{0_{M_i}}C\ket{0_{\all}}\bra{0_{\all}}C^{\dagger}\ket{0_{M_i}}\Pi^K_{F_i} =  \bra{0_{M_i}}\Pi^K_{F_i}C\ket{0_{\all}}\bra{0_{\all}}C^{\dagger}\Pi^K_{F_i} \ket{0_{M_i}} \nonumber\\
		  &  = \bra{0_{M_i}} C_{Wrap_i} \Proj^K_{F_i} C^{\dagger}_{Wrap_i} C\ket{0_{\all}}\bra{0_{\all}}C^{\dagger}C_{Wrap_i} \Proj^K_{F_i} C^{\dagger}_{Wrap_i}\ket{0_{M_i}} \nonumber \\
		  & = \bra{0_{M_i}} C_{Wrap_i} \circ \Proj^K_{F_i} \circ C'_{L_i} \circ  C_{B_i \cup M_i \cup F_i} \circ C'_{R_i}  \ket{0_{\all}}\bra{0_{\all}}(C')^{\dagger}_{L_i} \circ  C^{\dagger}_{B_i \cup M_i \cup F_i} \circ (C')^{\dagger}_{R_i}  \circ \Proj^K_{F_i} \circ C^{\dagger}_{Wrap_i}\ket{0_{M_i}} \nonumber \\
		  & = C_{Wrap_i}\Proj^K_{F_i} \left ( \bra{0_{M_i}}   \circ C'_{L_i} \circ  C_{B_i \cup M_i \cup F_i} \circ C'_{R_i}  \ket{0_{\all}}\bra{0_{\all}}(C')^{\dagger}_{L_i} \circ  C^{\dagger}_{B_i \cup M_i \cup F_i} \circ (C')^{\dagger}_{R_i}   \ket{0_{M_i}} \right )\Proj^K_{F_i} C^{\dagger}_{Wrap_i} \nonumber \\
		  & =  C_{Wrap_i}\Proj^K_{F_i} \left ( C'_{L_i} \ket{0_{L_i}}\bra{0_{L_i}}(C')^{\dagger}_{L_i} \otimes  \bra{0_{M_i}}    C_{B_i \cup M_i \cup F_i} \ket{0_{B_i \cup M_i \cup F_i}}\bra{0_{B_i \cup M_i \cup F_i}} C^{\dagger}_{B_i \cup M_i \cup F_i}\ket{0_{M_i}} \right . \\
		  & \left . \otimes   C'_{R_i}  \ket{0_{R_i}}\bra{0_{R_i}} (C')^{\dagger}_{R_i}    \right )\Proj^K_{F_i} C^{\dagger}_{Wrap_i} \\
		  & = C_{Wrap_i}\Proj^K_{F_i} \left ( C'_{L_i} \ket{0_{L_i}}\bra{0_{L_i}}(C')^{\dagger}_{L_i} \otimes  \ket{\psi} \bra{\psi}_{B_i \cup F_i}  \otimes   C'_{R_i}  \ket{0_{R_i}}\bra{0_{R_i}} (C')^{\dagger}_{R_i}    \right )\Proj^K_{F_i} C^{\dagger}_{Wrap_i} \\
		  &= C_{Wrap_i} \left ( C'_{L_i} \ket{0_{L_i}}\bra{0_{L_i}}(C')^{\dagger}_{L_i} \otimes  \Proj^K_{F_i}\ket{\psi} \bra{\psi}_{B_i  \cup F_i} \Proj^K_{F_i} \otimes   C'_{R_i}  \ket{0_{R_i}}\bra{0_{R_i}} (C')^{\dagger}_{R_i}    \right ) C^{\dagger}_{Wrap_i} \numberthis  \label{eq:lem20firstbatch}
		  \end{align*}
		  
		   Here the first equality holds because the $\Pi^K_{F_i}$ operators only act on register $F_i$, which is disjoint from register $M_i$.  The second equality holds by the the definition of $\Pi^K_{F_i}$, see Definition \ref{def:wrappedpi}.  The third equality holds by Equation \ref{eq:cwrapexamp}, repeated below for the convenience of the reader.  
		   
		   \begin{align} 
		   C_{Wrap_i}^{\dagger} \circ C   = C'_{L_i} \circ  C_{B_i \cup M_i \cup F_i} \circ C'_{R_i}  \nonumber
		   \end{align}
		   
		   The fourth equality holds because neither the operator $C_{Wrap_i}$, nor the operator $\Proj^K_{F_i}$ act (non-trivially) on the register $M_i$.  The fifth equality holds because the operators $C'_{L_i}$, $C'_{R_i}$, and $C_{B_i\cup M_i \cup F_i}$ all act on disjoint registers and are therefore in tensor product by definition.  The sixth equality holds by the definition of   $\ket{\psi}_{B_i \cup F_i}$ \mnote{see Equation ...}.  The seventh equality follows because $\Proj^K_{F_i}$ only acts (non-trivially) on the register $F_i$, by definition.

		  Now, define:
		      
		\begin{align*}
		&E \equiv  C_{Wrap_i} \left ( C'_{L_i} \ket{0_{L_i}}\bra{0_{L_i}}(C')^{\dagger}_{L_i} \otimes  \ket{w_1}\bra{w_1}_{F_i}\ket{\psi} \bra{\psi}_{B_i  \cup F_i} \ket{w_1}\bra{w_1}_{F_i} \otimes   C'_{R_i}  \ket{0_{R_i}}\bra{0_{R_i}} (C')^{\dagger}_{R_i}    \right ) C^{\dagger}_{Wrap_i} \\
		& -  C_{Wrap_i} \left ( C'_{L_i} \ket{0_{L_i}}\bra{0_{L_i}}(C')^{\dagger}_{L_i} \otimes  \Proj^K_{F_i}\ket{\psi} \bra{\psi}_{B_i  \cup F_i} \Proj^K_{F_i} \otimes   C'_{R_i}  \ket{0_{R_i}}\bra{0_{R_i}} (C')^{\dagger}_{R_i}    \right ) C^{\dagger}_{Wrap_i} \\
		& = C_{Wrap_i} \left ( C'_{L_i} \ket{0_{L_i}}\bra{0_{L_i}}(C')^{\dagger}_{L_i} \otimes  \left ( \ket{w_1}\bra{w_1}_{F_i}\ket{\psi} \bra{\psi}_{B_i  \cup F_i} \ket{w_1}\bra{w_1}_{F_i} - \Proj^K_{F_i}\ket{\psi} \bra{\psi}_{B_i  \cup F_i} \Proj^K_{F_i} \right ) \right. \\
		& \left . \otimes   C'_{R_i}  \ket{0_{R_i}}\bra{0_{R_i}} (C')^{\dagger}_{R_i}    \right )C^{\dagger}_{Wrap_i} \numberthis \label{eq:defofE}
		\end{align*}

		  This error term quantifies the difference between using our block-encoding construction $\Proj^K_{F_i}$ versus using the true projector onto the top Schmidt vector, $\ket{w_1}\bra{w_1}_{F_i}$, which is $\Proj^K_{F_i}$ is meant to approximate.  We will show below, using Lemma \ref{clm:schmidtproj}, that the trace norm of $E$ is small.  We have, from Equation \ref{eq:lem20firstbatch}, that:
		  
		  \mnote{Schmidt coefficients should be v's for the Bi register and w's for the Fi register.  Should change everywhere!}
		  
		  \begin{align*}
		  &\Pi^K_{F_i} \bra{0_{M_i}}C\ket{0_{\all}}\bra{0_{\all}}C^{\dagger}\ket{0_{M_i}}\Pi^K_{F_i}  \\
		  &= C_{Wrap_i} \left ( C'_{L_i} \ket{0_{L_i}}\bra{0_{L_i}}(C')^{\dagger}_{L_i} \otimes  \Proj^K_{F_i}\ket{\psi} \bra{\psi}_{B_i \cup F_i} \Proj^K_{F_i} \otimes   C'_{R_i}  \ket{0_{R_i}}\bra{0_{R_i}} (C')^{\dagger}_{R_i}    \right ) C^{\dagger}_{Wrap_i}\\
		  & = C_{Wrap_i} \left ( C'_{L_i} \ket{0_{L_i}}\bra{0_{L_i}}(C')^{\dagger}_{L_i} \otimes  \ket{w_1}\bra{w_1}_{F_i}\ket{\psi} \bra{\psi}_{B_i  \cup F_i} \ket{w_1}\bra{w_1}_{F_i} \otimes   C'_{R_i}  \ket{0_{R_i}}\bra{0_{R_i}} (C')^{\dagger}_{R_i}    \right ) C^{\dagger}_{Wrap_i} + E \\
		  & = 		  C_{Wrap_i} \left ( C'_{L_i} \ket{0_{L_i}}\bra{0_{L_i}}(C')^{\dagger}_{L_i} \otimes  \left ( \lambda_1^i\ket{v_1}_{B_i}\ket{w_1}_{F_i}\bra{v_1}_{B_i}\bra{w_1}_{F_i} \right ) \otimes   C'_{R_i}  \ket{0_{R_i}}\bra{0_{R_i}} (C')^{\dagger}_{R_i}    \right ) C^{\dagger}_{Wrap_i} + E\\
		  &= \lambda_1^i  C_{Wrap_i} \left ( C'_{L_i} \ket{0_{L_i}}\bra{0_{L_i}}(C')^{\dagger}_{L_i} \otimes  \left ( \ket{v_1}\bra{v_1}_{B_i} \otimes \ket{w_1}\bra{w_1}_{F_i} \right ) \otimes   C'_{R_i}  \ket{0_{R_i}}\bra{0_{R_i}} (C')^{\dagger}_{R_i}    \right ) C^{\dagger}_{Wrap_i} + E \nonumber \\
		  &= \lambda_1^i  C_{L-Wrap_i} \left ( C'_{L_i} \ket{0_{L_i}}\bra{0_{L_i}}(C')^{\dagger}_{L_i} \otimes   \ket{v_1}\bra{v_1}_{B_i} \right ) C^{\dagger}_{L-Wrap_i} \\
		  &\otimes C_{R-Wrap_i}\left ( \ket{w_1}\bra{w_1}_{F_i}   \otimes   C'_{R_i}  \ket{0_{R_i}}\bra{0_{R_i}} (C')^{\dagger}_{R_i}    \right ) C^{\dagger}_{R-Wrap_i} + E \nonumber \\
		  & = \lambda_1^i  C_{L-Wrap_i} \tr_{F_i}\left ( C'_{L_i} \ket{0_{L_i}}\bra{0_{L_i}}(C')^{\dagger}_{L_i} \otimes   \ket{v_1}\bra{v_1}_{B_i} \otimes \ket{w_1}\bra{w_1}_{F_i} \right ) C^{\dagger}_{L-Wrap_i} \\
		  &\otimes C_{R-Wrap_i} \tr_{B_i}\left (\ket{v_1}\bra{v_1}_{B_i} \otimes \ket{w_1}\bra{w_1}_{F_i}   \otimes   C'_{R_i}  \ket{0_{R_i}}\bra{0_{R_i}} (C')^{\dagger}_{R_i}    \right ) C^{\dagger}_{R-Wrap_i} + E \numberthis \label{eq:lem20batch2}
		  \end{align*}

		  Now that we have successfully approximated the starting state by a product state (plus an error term $E$) we want to switch back from the true Schmidt vector projectors $\ket{v_1}\bra{v_1}_{B_i}, \ket{w_1}\bra{w_1}_{F_i}$, to the original block encoding approximations of those projectors $\Proj^K_{B_i}, \Proj^K_{F_i}$, in order to complete the proof.  To do this we define two new error terms $H_1$, and $H_2$ as follows:
		  
		  \begin{align*}
		  &H_1 \equiv \lambda_1^i  C_{L-Wrap_i} \tr_{F_i}\left ( C'_{L_i} \ket{0_{L_i}}\bra{0_{L_i}}(C')^{\dagger}_{L_i} \otimes   \ket{v_1}\bra{v_1}_{B_i} \otimes \ket{w_1}\bra{w_1}_{F_i} \right ) C^{\dagger}_{L-Wrap_i} \\
		  &\otimes C_{R-Wrap_i} \tr_{B_i}\left (\ket{v_1}\bra{v_1}_{B_i} \otimes \ket{w_1}\bra{w_1}_{F_i}   \otimes   C'_{R_i}  \ket{0_{R_i}}\bra{0_{R_i}} (C')^{\dagger}_{R_i}    \right ) C^{\dagger}_{R-Wrap_i} \\
		  & -   C_{L-Wrap_i} \tr_{F_i}\left ( C'_{L_i} \ket{0_{L_i}}\bra{0_{L_i}}(C')^{\dagger}_{L_i} \otimes   \Proj^K_{F_i}\ket{\psi} \bra{\psi}_{B_i  \cup F_i} \Proj^K_{F_i} \right ) C^{\dagger}_{L-Wrap_i} \\
		  &\otimes C_{R-Wrap_i} \tr_{B_i}\left (\ket{v_1}\bra{v_1}_{B_i} \otimes \ket{w_1}\bra{w_1}_{F_i}   \otimes   C'_{R_i}  \ket{0_{R_i}}\bra{0_{R_i}} (C')^{\dagger}_{R_i}    \right ) C^{\dagger}_{R-Wrap_i} \numberthis \label{eq:defofG1} \\
		  \end{align*}
		  and
		  \begin{align*}
		  & H_2 =   C_{L-Wrap_i} \tr_{F_i}\left ( C'_{L_i} \ket{0_{L_i}}\bra{0_{L_i}}(C')^{\dagger}_{L_i} \otimes   \Proj^K_{F_i}\ket{\psi} \bra{\psi}_{B_i  \cup F_i} \Proj^K_{F_i} \right ) C^{\dagger}_{L-Wrap_i} \\
		  &\otimes C_{R-Wrap_i} \tr_{B_i}\left (\ket{v_1}\bra{v_1}_{B_i} \otimes \ket{w_1}\bra{w_1}_{F_i}   \otimes   C'_{R_i}  \ket{0_{R_i}}\bra{0_{R_i}} (C')^{\dagger}_{R_i}    \right ) C^{\dagger}_{R-Wrap_i} \\
		  &- 1/\lambda_1^iC_{L-Wrap_i} \tr_{F_i}\left ( C'_{L_i} \ket{0_{L_i}}\bra{0_{L_i}}(C')^{\dagger}_{L_i} \otimes   \Proj^K_{F_i}\ket{\psi} \bra{\psi}_{B_i  \cup F_i} \Proj^K_{F_i} \right ) C^{\dagger}_{L-Wrap_i} \\
		  &\otimes C_{R-Wrap_i} \tr_{B_i}\left (\Proj^K_{B_i}\ket{\psi} \bra{\psi}_{B_i  \cup F_i} \Proj^K_{B_i}   \otimes   C'_{R_i}  \ket{0_{R_i}}\bra{0_{R_i}} (C')^{\dagger}_{R_i}    \right ) C^{\dagger}_{R-Wrap_i} \numberthis \label{eq:defofG2}
		  \end{align*}

		  (we will later show that both error terms $H_1$ and $H_2$ are small in the trace norm) it follows from Equation \ref{eq:lem20batch2} that:

		  \begin{align*}
		  &\Pi^K_{F_i} \bra{0_{M_i}}C\ket{0_{\all}}\bra{0_{\all}}C^{\dagger}\ket{0_{M_i}}\Pi^K_{F_i} \\
		  & = \lambda_1^i  C_{L-Wrap_i} \tr_{F_i}\left ( C'_{L_i} \ket{0_{L_i}}\bra{0_{L_i}}(C')^{\dagger}_{L_i} \otimes   \ket{v_1}\bra{v_1}_{B_i} \otimes \ket{w_1}\bra{w_1}_{F_i} \right ) C^{\dagger}_{L-Wrap_i} \\
		  &\otimes C_{R-Wrap_i} \tr_{B_i}\left (\ket{v_1}\bra{v_1}_{B_i} \otimes \ket{w_1}\bra{w_1}_{F_i}   \otimes   C'_{R_i}  \ket{0_{R_i}}\bra{0_{R_i}} (C')^{\dagger}_{R_i}    \right ) C^{\dagger}_{R-Wrap_i} + E\\
		   &=  C_{L-Wrap_i} \tr_{F_i}\left ( C'_{L_i} \ket{0_{L_i}}\bra{0_{L_i}}(C')^{\dagger}_{L_i} \otimes   \Proj^K_{F_i}\ket{\psi} \bra{\psi}_{B_i  \cup F_i} \Proj^K_{F_i} \right ) C^{\dagger}_{L-Wrap_i} \\
		   &\otimes C_{R-Wrap_i} \tr_{B_i}\left (\ket{v_1}\bra{v_1}_{B_i} \otimes \ket{w_1}\bra{w_1}_{F_i}   \otimes   C'_{R_i}  \ket{0_{R_i}}\bra{0_{R_i}} (C')^{\dagger}_{R_i}    \right ) C^{\dagger}_{R-Wrap_i} + E + H_1 \nonumber \\
		   &=   1/\lambda_1^i C_{L-Wrap_i} \tr_{F_i}\left ( C'_{L_i} \ket{0_{L_i}}\bra{0_{L_i}}(C')^{\dagger}_{L_i} \otimes   \Proj^K_{F_i}\ket{\psi} \bra{\psi}_{B_i  \cup F_i} \Proj^K_{F_i} \right ) C^{\dagger}_{L-Wrap_i} \\
		   &\otimes C_{R-Wrap_i} \tr_{B_i}\left (\Proj^K_{B_i}\ket{\psi} \bra{\psi}_{B_i  \cup F_i} \Proj^K_{B_i}   \otimes   C'_{R_i}  \ket{0_{R_i}}\bra{0_{R_i}} (C')^{\dagger}_{R_i}    \right ) C^{\dagger}_{R-Wrap_i}  + E + H_1+ H_2 \nonumber \\
		    &= 1/\lambda_1^i  \tr_{F_i}\left ( \Proj^K_{F_i}C_{L-Wrap_i} \left ( C'_{L_i} \ket{0_{L_i}}\bra{0_{L_i}}(C')^{\dagger}_{L_i} \otimes   \ket{\psi} \bra{\psi}_{B_i \cup F_i} \right ) C^{\dagger}_{L-Wrap_i} \Proj^K_{F_i} \right )  \\
		    &\otimes  \tr_{B_i}\left (\Proj^K_{B_i} C_{R-Wrap_i} \left ( \ket{\psi} \bra{\psi}_{B_i \cup F_i}    \otimes   C'_{R_i}  \ket{0_{R_i}}\bra{0_{R_i}} (C')^{\dagger}_{R_i} \right ) C^{\dagger}_{R-Wrap_i} \Proj^K_{B_i}  \right )  + E + H_1+ H_2 \nonumber \\
		    & =  1/\lambda_1^i \tr_{F_i}\left (  \Proj^K_{F_i} \bra{0_{M_i}}C_{L_i}C_{B_i \cup M_i \cup F_i}  \ket{0_{L_i\cup B_i \cup M_i \cup F_i}}\bra{0_{L_i\cup B_i \cup M_i \cup F_i}} C^{\dagger}_{B_i \cup M_i \cup F_i}C^{\dagger}_{L_i}\ket{0_{M_i}}\Proj^K_{F_i}  \right )  \nonumber \\
		    &\otimes  \tr_{B_i}\left (  \Proj^K_{B_i} \bra{0_{M_i}}C_{R_i}C_{B_i \cup M_i \cup F_i}  \ket{0_{R_i\cup B_i \cup M_i \cup F_i}}\bra{0_{R_i\cup B_i \cup M_i \cup F_i}}C^{\dagger}_{B_i \cup M_i \cup F_i}C^{\dagger}_{R_i}\ket{0_{M_i}}\Proj^K_{B_i}  \right ) + E + H_1+ H_2 \\
		    &= 1/\lambda_1^i \tr_{F_i}\left (  \ket{\Xi_{L_i}}\bra{\Xi_{L_i}} \right )  \otimes  \tr_{B_i}\left ( \ket{\Xi_{R_i}}\bra{\Xi_{R_i}} \right ) + E + H_1 + H_2
		  \end{align*}
		  
		  It follows, by triangle inequality, that:
		  
		  \begin{align}
		  &\left \| \Pi^K_{F_i} \bra{0_{M_i}}C\ket{0_{\all}}\bra{0_{\all}}C^{\dagger}\ket{0_{M_i}}\Pi^K_{F_i} \right. \nonumber \\
		  &\left .	    -  1/\lambda_1^i \tr_{F_i}\left (  \Proj^K_{F_i} \bra{0_{M_i}}C_{L_i}C_{B_i \cup M_i \cup F_i}  \ket{0_{L_i\cup B_i \cup M_i \cup F_i}}\bra{0_{L_i\cup B_i \cup M_i \cup F_i}} C^{\dagger}_{B_i \cup M_i \cup F_i}C^{\dagger}_{L_i}\ket{0_{M_i}}\Proj^K_{F_i}  \right )  \nonumber \right .\\
		  &\left . \otimes  \tr_{B_i}\left (  \Proj^K_{B_i} \bra{0_{M_i}}C_{R_i}C_{B_i \cup M_i \cup F_i}  \ket{0_{R_i\cup B_i \cup M_i \cup F_i}}\bra{0_{R_i\cup B_i \cup M_i \cup F_i}}C^{\dagger}_{B_i \cup M_i \cup F_i}C^{\dagger}_{R_i}\ket{0_{M_i}}\Proj^K_{B_i}  \right )\right \| \\
		  &  =\left \| \ket{\Omega_i}\bra{\Omega_i}  -  1/\lambda_1^i \tr_{F_i}\left (  \ket{\Xi_{L_i}}\bra{\Xi_{L_i}} \right )  \otimes  \tr_{B_i}\left ( \ket{\Xi_{R_i}}\bra{\Xi_{R_i}} \right )\right \|  \\
		  &= \|E+H_1+H_2 \|\leq \|E\| + \|H_1\| + \|H_2\|
		  \end{align}
		  
		  It remains to bound the norms (in this case the trace norm) of $E, H_1, H_2$.  We will start with $E$:
		  
		 	From the definition of $E$ (Equation \ref{eq:defofE}) we see that:
		 	     
		 	     	\begin{align*}
		 	     \|E\| &= \left \| C_{Wrap_i} \left ( C'_{L_i} \ket{0_{L_i}}\bra{0_{L_i}}(C')^{\dagger}_{L_i} \otimes  \left ( \ket{w_1}\bra{w_1}_{F_i}\ket{\psi} \bra{\psi}_{B_i  \cup F_i} \ket{w_1}\bra{w_1}_{F_i} - \Proj^K_{F_i}\ket{\psi} \bra{\psi}_{B_i  \cup F_i} \Proj^K_{F_i} \right ) \right. \right . \\
		 	     	 &\left . \left . \otimes   C'_{R_i}  \ket{0_{R_i}}\bra{0_{R_i}} (C')^{\dagger}_{R_i}    \right ) C^{\dagger}_{Wrap_i}\right \|\\
		 	     	 & = \| C'_{L_i} \ket{0_{L_i}}\bra{0_{L_i}}(C')^{\dagger}_{L_i} \| \cdot  \|\ket{w_1}\bra{w_1}_{F_i}\ket{\psi} \bra{\psi}_{B_i  \cup F_i} \ket{w_1}\bra{w_1}_{F_i} - \Proj^K_{F_i}\ket{\psi} \bra{\psi}_{B_i \cup F_i} \Proj^K_{F_i}\| \\
		 	     	& \cdot   \|C'_{R_i}  \ket{0_{R_i}}\bra{0_{R_i}} (C')^{\dagger}_{R_i}\|\\
		 	     & = \|  \ket{0_{L_i}}\bra{0_{L_i}} \| \cdot  \|\ket{w_1}\bra{w_1}_{F_i}\ket{\psi} \bra{\psi}_{B_i  \cup F_i} \ket{w_1}\bra{w_1}_{F_i} - \Proj^K_{F_i}\ket{\psi} \bra{\psi}_{B_i  \cup F_i} \Proj^K_{F_i}\|  \cdot   \|  \ket{0_{R_i}}\bra{0_{R_i}} \|\\
		 	     & = \|\ket{w_1}\bra{w_1}_{F_i}\ket{\psi} \bra{\psi}_{B_i  \cup F_i} \ket{w_1}\bra{w_1}_{F_i} - \Proj^K_{F_i}\ket{\psi} \bra{\psi}_{B_i  \cup F_i} \Proj^K_{F_i}\| \leq 2\projerror = 2 \left (\frac{1 - \lambda_1^i}{\lambda_1^i} \right )^K 
		 	     	\end{align*}

		 	     Here the first equality follows by definition of $E$ (Equation \ref{eq:defofE}), the second equality follows because $C_{Wrap_i}$ is unitary and by using the tensor product structure after $C_{Wrap_i}$ is removed, the third equality follows because $C'_{L_i}$ and $C'_{R_i}$ are unitary, the fourth equality follows because $\|  \ket{0_{L_i}}\bra{0_{L_i}} \| = \|  \ket{0_{R_i}}\bra{0_{R_i}} \| = 1$, and the inequality follows by two sequential applications of Lemma \ref{clm:schmidtproj}.  
		 	     
		 	     Next we will bound $\|H_1\|$.  From the definition of $H_1$ in Equation \ref{eq:defofG1} we have that:
		 	     
		 	     \begin{align*}
		 	     &\|H_1\| \equiv \left \| \lambda_1^i  C_{L-Wrap_i} \tr_{F_i}\left ( C'_{L_i} \ket{0_{L_i}}\bra{0_{L_i}}(C')^{\dagger}_{L_i} \otimes   \ket{v_1}\bra{v_1}_{B_i} \otimes \ket{w_1}\bra{w_1}_{F_i} \right ) C^{\dagger}_{L-Wrap_i} \right .\\
		 	     & \otimes C_{R-Wrap_i} \tr_{B_i}\left (\ket{v_1}\bra{v_1}_{B_i} \otimes \ket{w_1}\bra{w_1}_{F_i}   \otimes   C'_{R_i}  \ket{0_{R_i}}\bra{0_{R_i}} (C')^{\dagger}_{R_i}    \right ) C^{\dagger}_{R-Wrap_i}  \\
		 	     & -   C_{L-Wrap_i} \tr_{F_i}\left ( C'_{L_i} \ket{0_{L_i}}\bra{0_{L_i}}(C')^{\dagger}_{L_i} \otimes   \Proj^K_{F_i}\ket{\psi} \bra{\psi}_{B_i \cup F_i} \Proj^K_{F_i} \right ) C^{\dagger}_{L-Wrap_i} \\
		 	     & \left . \otimes C_{R-Wrap_i} \tr_{B_i}\left (\ket{v_1}\bra{v_1}_{B_i} \otimes \ket{w_1}\bra{w_1}_{F_i}   \otimes   C'_{R_i}  \ket{0_{R_i}}\bra{0_{R_i}} (C')^{\dagger}_{R_i}    \right ) C^{\dagger}_{R-Wrap_i}  \right \| \\
		 	     & = \left \| \left ( \lambda_1^i  C_{L-Wrap_i} \tr_{F_i}\left ( C'_{L_i} \ket{0_{L_i}}\bra{0_{L_i}}(C')^{\dagger}_{L_i} \otimes   \ket{v_1}\bra{v_1}_{B_i} \otimes \ket{w_1}\bra{w_1}_{F_i} \right ) C^{\dagger}_{L-Wrap_i} \right . \right .\\
		 	     & \left .  -   C_{L-Wrap_i} \tr_{F_i}\left ( C'_{L_i} \ket{0_{L_i}}\bra{0_{L_i}}(C')^{\dagger}_{L_i} \otimes   \Proj^K_{F_i}\ket{\psi} \bra{\psi}_{B_i \cup F_i} \Proj^K_{F_i} \right ) C^{\dagger}_{L-Wrap_i}\right )    \\
		 	     & \left . \otimes  C_{R-Wrap_i} \tr_{B_i}\left (\ket{v_1}\bra{v_1}_{B_i} \otimes \ket{w_1}\bra{w_1}_{F_i}   \otimes   C'_{R_i}  \ket{0_{R_i}}\bra{0_{R_i}} (C')^{\dagger}_{R_i}    \right ) C^{\dagger}_{R-Wrap_i}  \right \|\\
		 	      &= \left \| \left ( \lambda_1^i  C_{L-Wrap_i} \tr_{F_i}\left ( C'_{L_i} \ket{0_{L_i}}\bra{0_{L_i}}(C')^{\dagger}_{L_i} \otimes   \ket{v_1}\bra{v_1}_{B_i} \otimes \ket{w_1}\bra{w_1}_{F_i} \right ) C^{\dagger}_{L-Wrap_i} \right . \right .\\
		 	      & \left . \left .  -   C_{L-Wrap_i} \tr_{F_i}\left ( C'_{L_i} \ket{0_{L_i}}\bra{0_{L_i}}(C')^{\dagger}_{L_i} \otimes   \Proj^K_{F_i}\ket{\psi} \bra{\psi}_{B_i  \cup F_i} \Proj^K_{F_i} \right ) C^{\dagger}_{L-Wrap_i}\right )   \right \| \\
		 	      & \cdot \left \|  C_{R-Wrap_i} \tr_{B_i}\left (\ket{v_1}\bra{v_1}_{B_i} \otimes \ket{w_1}\bra{w_1}_{F_i}   \otimes   C'_{R_i}  \ket{0_{R_i}}\bra{0_{R_i}} (C')^{\dagger}_{R_i}    \right ) C^{\dagger}_{R-Wrap_i}  \right \| \\
		 	       &= \left \| \left ( \lambda_1^i   \tr_{F_i}\left ( C'_{L_i} \ket{0_{L_i}}\bra{0_{L_i}}(C')^{\dagger}_{L_i} \otimes   \ket{v_1}\bra{v_1}_{B_i} \otimes \ket{w_1}\bra{w_1}_{F_i} \right )  \right . \right .\\
		 	       & \left . \left .  -    \tr_{F_i}\left ( C'_{L_i} \ket{0_{L_i}}\bra{0_{L_i}}(C')^{\dagger}_{L_i} \otimes   \Proj^K_{F_i}\ket{\psi} \bra{\psi}_{B_i  \cup F_i} \Proj^K_{F_i} \right ) \right )   \right \| \\
		 	       & \cdot \left \|     \ket{w_1}\bra{w_1}_{F_i}   \otimes   C'_{R_i}  \ket{0_{R_i}}\bra{0_{R_i}} (C')^{\dagger}_{R_i}     \right \| \\
		 	        &= \left \|    \tr_{F_i}\left ( C'_{L_i} \ket{0_{L_i}}\bra{0_{L_i}}(C')^{\dagger}_{L_i} \otimes  \left ( \lambda_1^i  \ket{v_1}\bra{v_1}_{B_i} \otimes \ket{w_1}\bra{w_1}_{F_i} - \Proj^K_{F_i}\ket{\psi} \bra{\psi}_{B_i  \cup F_i} \Proj^K_{F_i}\right ) \right )     \right \|\\
		 	        & \cdot \left \|     \ket{w_1}\bra{w_1}_{F_i}   \otimes  \ket{0_{R_i}}\bra{0_{R_i}}    \right \| \\
		 	       & \leq  \left \|     \lambda_1^i  \ket{v_1}\bra{v_1}_{B_i} \otimes \ket{w_1}\bra{w_1}_{F_i} - \Proj^K_{F_i}\ket{\psi} \bra{\psi}_{B_i  \cup F_i} \Proj^K_{F_i}    \right \|\\
		 	       &\leq 2\lambda_1^i \projerror \leq 2\projerror =2 \left (\frac{1 - \lambda_1^i}{\lambda_1^i} \right )^K
		 	     \end{align*}

		 	    Here the first equality follows by definition (Equation \ref{eq:defofG1}), the second equality follows by regrouping terms, and the third equality follows by the tensor product structure.  \mnote{ the fourth equality follows by} \mnote{unitariness of some operators}
		 	     
		 	     The proof for the bound on $H_2$ is extremely similar to the bound on $H_1$, and so we will not repeat the argument.
		 	     
		 	     \mnote{finish the argument for $H_1$, give argument for $H_2$ or say that it is very similar, complete the Lemma proof}

	\end{proof}

\subsection{Proofs for Statements in Section \ref{section:quasi-poly-time}}

\begin{lemma*} [Restatement of Lemma \ref{lem:multicuttrick}]
	\begin{align*}
	&\Bigg \|\sum_{\sigma\in\mathcal{P}(\{i+1,\dots,j-1\})\setminus\emptyset} (-1)^{|\sigma|+1} \Bigg (  \frac{1}{(\kappa^{i}_{T, \epsilon_2}\kappa^{j}_{T, \epsilon_2})^{4K+1}}\mathcal{A}(S_{L,i},\eta-1) \cdot \mathcal{A}(S_{j,R},\eta-1) \\
	& \cdot  \mathcal{B}\left ( \Big (\otimes_{k \in \sigma} \Pi^K_{F_k} \bra{0_{M_k}} \Big )\phi_{i,j}\Big (\otimes_{k \in \sigma} \ket{0_{M_k}}\Pi^K_{F_k}\Big ),\frac{\epsilon}{2^\Delta} \right) - \bra{0_{ALL}}\ket{\Psi_{\{i,j\}\cup \sigma}}\bra{\Psi_{\{i,j\}\cup \sigma}}\ket{0_{ALL}} \Bigg ) \Bigg\| \numberthis\\
	&\leq E_3(n, K, T, \epsilon_2, \epsilon, \Delta) + 16f(S,\eta-1,\Delta,\epsilon),
	\end{align*}
	
	where 
	
	\begin{align*}
	&E_3(n, K, T, \epsilon_2, \epsilon,\Delta) \equiv   O \left (2^{\Delta} (6\projerror) + 2^\Delta K \left (e(n)^{2T}+\epsilon_2 \right ) + \epsilon \right) 
	\end{align*}

\end{lemma*}

\begin{proof}
	The proof proceeds in two parts. First, we show by direct calculation that the desired error quantity can be upper bounded by the sum of four error quantities $G_1$, $G_2$, $G_3$, and $G_4$:
	\begin{align*}
	&\Bigg \|\sum_{\sigma\in\mathcal{P}(\{i+1,\dots,j-1\})\setminus\emptyset} (-1)^{|\sigma|+1} \Bigg (  \frac{1}{(\kappa^{i}_{T, \epsilon_2}\kappa^{j}_{T, \epsilon_2})^{4K+1}}\mathcal{A}(S_{L,i},\eta-1) \cdot \mathcal{A}(S_{j,R},\eta-1) \\
	& \cdot  \mathcal{B}\left ( \Big (\otimes_{k \in \sigma} \Pi^K_{F_k} \bra{0_{M_k}} \Big )\phi_{i,j}\Big (\otimes_{k \in \sigma} \ket{0_{M_k}}\Pi^K_{F_k}\Big ),\frac{\epsilon}{2^\Delta} \right) - \bra{0_{ALL}}\ket{\Psi_{\{i,j\}\cup \sigma}}\bra{\Psi_{\{i,j\}\cup \sigma}}\ket{0_{ALL}} \Bigg ) \Bigg\| \\
	&\leq G_1+G_2+G_3+G_4. \numberthis\label{eq:four-error-terms}
	\end{align*}
	We then bound these four terms individually. We begin by demonstrating Equation \ref{eq:four-error-terms}, and defining $G_1$, $G_2$, $G_3$, and $G_4$ in the process.
	
	\begin{align*}
	&\Bigg \|\sum_{\sigma\in\mathcal{P}(\{i+1,\dots,j-1\})\setminus\emptyset} (-1)^{|\sigma|+1}\left (  \frac{1}{(\kappa^{i}_{T, \epsilon_2}\kappa^{j}_{T, \epsilon_2})^{2K+1}}\mathcal{A}(S_{L,i},\eta-1) \cdot \mathcal{A}(S_{j,R},\eta-1) \right .\\
	& \cdot  \mathcal{B}\left ( \left (\otimes_{k \in \sigma} \Pi^K_{F_k} \bra{0_{M_k}} \right )\phi_{i,j}\left (\otimes_{k \in \sigma} \ket{0_{M_k}}\Pi^K_{F_k}\right ),\frac{\epsilon}{2^\Delta} \right)  - \bra{0_{ALL}}\ket{\Psi_{\{i,j\}\cup \sigma}}\bra{\Psi_{\{i,j\}\cup \sigma}}\ket{0_{ALL}}  \Bigg )
	\Bigg \| \\
	&\leq \Bigg \|\frac{1}{(\lambda_1^i \lambda_1^j)^{4K+1}} \sum_{\sigma\in\mathcal{P}(\{i+1,\dots,j-1\})\setminus\emptyset} (-1)^{|\sigma|+1}\Bigg ( \mathcal{A}(S_{L,i},\eta-1) \cdot \mathcal{A}(S_{j,R},\eta-1) \\
	& \left . \cdot   \mathcal{B}\left ( \left (\otimes_{k \in \sigma} \Pi^K_{F_k} \bra{0_{M_k}} \right )\phi_{i,j}\left (\otimes_{k \in \sigma} \ket{0_{M_k}}\Pi^K_{F_k}\right ),\frac{\epsilon}{2^\Delta} \right) \right .\\
	&-\bra{0_{ALL}} \phi_{L, i}\ket{0_{ALL}} \cdot \bra{0_{ALL}} \left (\otimes_{k \in \sigma} \Pi^K_{F_k} \bra{0_{M_k}} \right )\phi_{i,j}\left (\otimes_{k \in \sigma} \ket{0_{M_k}}\Pi^K_{F_k}\right )\ket{0_{ALL}}\cdot \bra{0_{ALL}} \phi_{j,R}\ket{0_{ALL}} \Bigg )\Bigg\|\\
	& + \sum_{\sigma\in\mathcal{P}(\{i+1,\dots,j-1\})\setminus\emptyset} \Bigg \|\bra{0_{ALL}}\ket{\Psi_{\{i,j\}\cup \sigma}}\bra{\Psi_{\{i,j\}\cup \sigma}}\ket{0_{ALL}}\\
	&- \frac{1}{(\lambda^{i}_1\lambda^{j}_{1})^{4K+1}}\bra{0_{ALL}} \phi_{L, i}\ket{0_{ALL}} \cdot \bra{0_{ALL}} \left (\otimes_{k \in \sigma} \Pi^K_{F_k} \bra{0_{M_k}} \right )\phi_{i,j}\left (\otimes_{k \in \sigma} \ket{0_{M_k}}\Pi^K_{F_k}\right )\ket{0_{ALL}}\cdot \bra{0_{ALL}} \phi_{j,R}\ket{0_{ALL}} \Bigg \|\\
	& +\sum_{\sigma\in\mathcal{P}(\{i+1,\dots,j-1\})\setminus\emptyset} \Bigg \|\left(\frac{1}{(\lambda^{i}_1\lambda^{j}_{1})^{4K+1}} - \frac{1}{(\kappa^{i}_{T, \epsilon_2}\kappa^{j}_{T, \epsilon_2})^{4K+1}}\right )\\
	&\cdot \mathcal{A}(S_{L,i},\eta-1) \cdot \mathcal{A}(S_{j,R},\eta-1) \cdot   \mathcal{B}\left ( \left (\otimes_{k \in \sigma} \Pi^K_{F_k} \bra{0_{M_k}} \right )\phi_{i,j}\left (\otimes_{k \in \sigma} \ket{0_{M_k}}\Pi^K_{F_k}\right ),\frac{\epsilon}{2^\Delta} \right)\Bigg \|
	\end{align*}
	which is equal to
	\begin{align*}
	&= \Bigg \| \frac{1}{(\lambda_1^i \lambda_1^j)^{4K+1}}\sum_{\sigma\in\mathcal{P}(\{i+1,\dots,j-1\})\setminus\emptyset} (-1)^{|\sigma|+1}\left (  \mathcal{A}(S_{L,i},\eta-1) \cdot \mathcal{A}(S_{j,R},\eta-1) \right .\\
	& \left . \cdot   \mathcal{B}\left ( \left (\otimes_{k \in \sigma} \Pi^K_{F_k} \bra{0_{M_k}} \right )\phi_{i,j}\left (\otimes_{k \in \sigma} \ket{0_{M_k}}\Pi^K_{F_k}\right ),\frac{\epsilon}{2^\Delta} \right) \right .\\
	&-\bra{0_{ALL}} \phi_{L, i}\ket{0_{ALL}} \cdot \bra{0_{ALL}} \left (\otimes_{k \in \sigma} \Pi^K_{F_k} \bra{0_{M_k}} \right )\phi_{i,j}\left (\otimes_{k \in \sigma} \ket{0_{M_k}}\Pi^K_{F_k}\right )\ket{0_{ALL}}\cdot \bra{0_{ALL}} \phi_{j,R}\ket{0_{ALL}} \Bigg )\Bigg\|\\
	&+G_1 + G_2 \numberthis \label{eq:part1bigerror}
	\end{align*}
	where $G_1$ and $G_2$ are defined as the error quantities
	
	\begin{align*}
	&G_1 \equiv \sum_{\sigma\in\mathcal{P}(\{i+1,\dots,j-1\})\setminus\emptyset} \Bigg \|\bra{0_{ALL}}\ket{\Psi_{\{i,j\}\cup \sigma}}\bra{\Psi_{\{i,j\}\cup \sigma}}\ket{0_{ALL}}\\
	&- \frac{1}{(\lambda^{i}_1\lambda^{j}_{1})^{4K+1}}\bra{0_{ALL}} \phi_{L, i}\ket{0_{ALL}} \cdot \bra{0_{ALL}} \left (\otimes_{k \in \sigma} \Pi^K_{F_k} \bra{0_{M_k}} \right )\phi_{i,j}\left (\otimes_{k \in \sigma} \ket{0_{M_k}}\Pi^K_{F_k}\right )\ket{0_{ALL}}\cdot \bra{0_{ALL}} \phi_{j,R}\ket{0_{ALL}} \Bigg \|
	\end{align*}
	and
	\begin{align*}
		&G_2 \equiv \sum_{\sigma\in\mathcal{P}(\{i+1,\dots,j-1\})\setminus\emptyset} \Bigg \|\left(\frac{1}{(\lambda^{i}_1\lambda^{j}_{1})^{4K+1}} - \frac{1}{(\kappa^{i}_{T, \epsilon_2}\kappa^{j}_{T, \epsilon_2})^{4K+1}}\right )\\
		&\cdot \mathcal{A}(S_{L,i},\eta-1) \cdot \mathcal{A}(S_{j,R},\eta-1) \cdot   \mathcal{B}\left ( \left (\otimes_{k \in \sigma} \Pi^K_{F_k} \bra{0_{M_k}} \right )\phi_{i,j}\left (\otimes_{k \in \sigma} \ket{0_{M_k}}\Pi^K_{F_k}\right ),\frac{\epsilon}{2^\Delta} \right)\Bigg \|.
	\end{align*}

		Later we will bound the size of $G_1$ and $G_2$ using Lemmas \ref{clm:breaktoproduct}, and \ref{lem:lambdaapprox} respectively.  For now we carry them along in our calculation.  So, continuing where we left off in Equation \ref{eq:part1bigerror}:

	\begin{align*}
	&\Bigg \|\sum_{\sigma\in\mathcal{P}(\{i+1,\dots,j-1\})\setminus\emptyset} (-1)^{|\sigma|+1}\left (  \frac{1}{(\kappa^{i}_{T, \epsilon_2}\kappa^{j}_{T, \epsilon_2})^{2K+1}}\mathcal{A}(S_{L,i},\eta-1) \cdot \mathcal{A}(S_{j,R},\eta-1) \right .\\
	& \cdot  \mathcal{B}\left ( \left (\otimes_{k \in \sigma} \Pi^K_{F_k} \bra{0_{M_k}} \right )\phi_{i,j}\left (\otimes_{k \in \sigma} \ket{0_{M_k}}\Pi^K_{F_k}\right ),\frac{\epsilon}{2^\Delta} \right)  - \bra{0_{ALL}}\ket{\Psi_{\{i,j\}\cup \sigma}}\bra{\Psi_{\{i,j\}\cup \sigma}}\ket{0_{ALL}}  \Bigg )
	\Bigg \| \\
	&\leq \Bigg \| \frac{1}{(\lambda_1^i \lambda_1^j)^{4K+1}}\sum_{\sigma\in\mathcal{P}(\{i+1,\dots,j-1\})\setminus\emptyset} (-1)^{|\sigma|+1}\left (  \mathcal{A}(S_{L,i},\eta-1) \cdot \mathcal{A}(S_{j,R},\eta-1) \right .\\
	& \left . \cdot   \mathcal{B}\left ( \left (\otimes_{k \in \sigma} \Pi^K_{F_k} \bra{0_{M_k}} \right )\phi_{i,j}\left (\otimes_{k \in \sigma} \ket{0_{M_k}}\Pi^K_{F_k}\right ),\frac{\epsilon}{2^\Delta} \right) \right .\\
	&-\bra{0_{ALL}} \phi_{L, i}\ket{0_{ALL}} \cdot \bra{0_{ALL}} \left (\otimes_{k \in \sigma} \Pi^K_{F_k} \bra{0_{M_k}} \right )\phi_{i,j}\left (\otimes_{k \in \sigma} \ket{0_{M_k}}\Pi^K_{F_k}\right )\ket{0_{ALL}}\cdot \bra{0_{ALL}} \phi_{j,R}\ket{0_{ALL}} \Bigg )\Bigg\|\\
	&+G_1 + G_2\\
	& \leq\Bigg \|\frac{1}{(\lambda_1^i \lambda_1^j)^{4K+1}}\mathcal{A}(S_{L,i},\eta-1) \cdot \mathcal{A}(S_{j,R},\eta-1) \sum_{\sigma\in\mathcal{P}(\{i+1,\dots,j-1\})\setminus\emptyset} (-1)^{|\sigma|+1}\left (   \right .\\
	& \left . \cdot   \mathcal{B}\left ( \left (\otimes_{k \in \sigma} \Pi^K_{F_k} \bra{0_{M_k}} \right )\phi_{i,j}\left (\otimes_{k \in \sigma} \ket{0_{M_k}}\Pi^K_{F_k}\right ),\frac{\epsilon}{2^\Delta} \right) \right .\\
	&- \bra{0_{ALL}} \left (\otimes_{k \in \sigma} \Pi^K_{F_k} \bra{0_{M_k}} \right )\phi_{i,j}\left (\otimes_{k \in \sigma} \ket{0_{M_k}}\Pi^K_{F_k}\right )\ket{0_{ALL}} \Bigg )\Bigg\|\\
	& +\Bigg \|\frac{1}{(\lambda_1^i \lambda_1^j)^{4K+1}} \Bigg (   \mathcal{A}(S_{L,i},\eta-1) \cdot \mathcal{A}(S_{j,R},\eta-1)-\bra{0_{ALL}} \phi_{L, i}\ket{0_{ALL}} \cdot \bra{0_{ALL}} \phi_{j,R}\ket{0_{ALL}} \Bigg )\\
	&\cdot \sum_{\sigma\in\mathcal{P}(\{i+1,\dots,j-1\})\setminus\emptyset} (-1)^{|\sigma|+1} \bra{0_{ALL}} \left (\otimes_{k \in \sigma} \Pi^K_{F_k} \bra{0_{M_k}} \right )\phi_{i,j}\left (\otimes_{k \in \sigma} \ket{0_{M_k}}\Pi^K_{F_k}\right )\ket{0_{ALL}}\Bigg\|\\
	&+G_1 + G_2\\
	& \leq  2^{\Delta} \cdot \frac{\epsilon}{2^{\Delta}} \cdot \Bigg \|\frac{1}{(\lambda_1^i \lambda_1^j)^{4K+1}}\mathcal{A}(S_{L,i},\eta-1) \cdot \mathcal{A}(S_{j,R},\eta-1)\Bigg\|\\
	& +\Bigg \|\frac{1}{(\lambda_1^i \lambda_1^j)^{4K+1}} \Bigg (  \mathcal{A}(S_{L,i},\eta-1) \cdot \mathcal{A}(S_{j,R},\eta-1)-\bra{0_{ALL}} \phi_{L, i}\ket{0_{ALL}} \cdot \bra{0_{ALL}} \phi_{j,R}\ket{0_{ALL}} \Bigg )\\
	&\cdot \sum_{\sigma\in\mathcal{P}(\{i+1,\dots,j-1\})\setminus\emptyset} (-1)^{|\sigma|+1} \bra{0_{ALL}} \left (\otimes_{k \in \sigma} \Pi^K_{F_k} \bra{0_{M_k}} \right )\phi_{i,j}\left (\otimes_{k \in \sigma} \ket{0_{M_k}}\Pi^K_{F_k}\right )\ket{0_{ALL}}\Bigg\|\\
	&+G_1 + G_2\\
	\end{align*}
	\begin{align*}
	& \leq \epsilon \cdot \Bigg \|\frac{1}{(\lambda_1^i \lambda_1^j)^{4K+1}}\mathcal{A}(S_{L,i},\eta-1) \cdot \mathcal{A}(S_{j,R},\eta-1)\Bigg\|\\
	& +\Bigg \|\frac{1}{(\lambda_1^i \lambda_1^j)^{4K+1}} \Bigg (  \mathcal{A}(S_{L,i},\eta-1) \cdot \mathcal{A}(S_{j,R},\eta-1)-\bra{0_{ALL}} \phi_{L, i}\ket{0_{ALL}} \cdot \bra{0_{ALL}} \phi_{j,R}\ket{0_{ALL}} \Bigg )\cdot \bra{0_{ALL}} \phi_{i,j}\ket{0_{ALL}}\Bigg\|\\
	&+\Bigg \|\frac{1}{(\lambda_1^i \lambda_1^j)^{4K+1}} \Bigg (   \mathcal{A}(S_{L,i},\eta-1) \cdot \mathcal{A}(S_{j,R},\eta-1)-\bra{0_{ALL}} \phi_{L, i}\ket{0_{ALL}} \cdot \bra{0_{ALL}} \phi_{j,R}\ket{0_{ALL}} \Bigg )\\
	&\cdot \Bigg( \bra{0_{ALL}} \phi_{i,j}\ket{0_{ALL}} -  \sum_{\sigma\in\mathcal{P}(\{i+1,\dots,j-1\})\setminus\emptyset} (-1)^{|\sigma|+1} \bra{0_{ALL}} \left (\otimes_{k \in \sigma} \Pi^K_{F_k} \bra{0_{M_k}} \right )\phi_{i,j}\left (\otimes_{k \in \sigma} \ket{0_{M_k}}\Pi^K_{F_k}\right )\ket{0_{ALL}}\Bigg)\Bigg\|\\
	&+G_1 + G_2\\
	& \leq G_1 + G_2 +G_3 + \Bigg \|\frac{1}{(\lambda_1^i \lambda_1^j)^{4K+1}} \Bigg (  \mathcal{A}(S_{L,i},\eta-1) \cdot \mathcal{A}(S_{j,R},\eta-1)-\bra{0_{ALL}} \phi_{L, i}\ket{0_{ALL}} \cdot \bra{0_{ALL}} \phi_{j,R}\ket{0_{ALL}} \Bigg )\Bigg\|\cdot\\
	& \Bigg( \Bigg \|  \bra{0_{ALL}} \phi_{i,j}\ket{0_{ALL}} -  \sum_{\sigma\in\mathcal{P}(\{i+1,\dots,j-1\})\setminus\emptyset} (-1)^{|\sigma|+1} \bra{0_{ALL}} \left (\otimes_{k \in \sigma} \Pi^K_{F_k} \bra{0_{M_k}} \right )\phi_{i,j}\left (\otimes_{k \in \sigma} \ket{0_{M_k}}\Pi^K_{F_k}\right )\ket{0_{ALL}}\Bigg\| + 1\Bigg ) \\
	& \leq G_1 + G_2 +G_3 + G_4
	\end{align*}

	Where

		\begin{align*}
		&G_3 \equiv \epsilon \cdot \Bigg \|\frac{1}{(\lambda_1^i \lambda_1^j)^{4K+1}}\mathcal{A}(S_{L,i},\eta-1) \cdot \mathcal{A}(S_{j,R},\eta-1)\Bigg\| 
		\end{align*}

	and
		
		\begin{align*}
		&G_4 \equiv\Bigg \|\frac{1}{(\lambda_1^i \lambda_1^j)^{4K+1}} \Bigg (  \mathcal{A}(S_{L,i},\eta-1) \cdot \mathcal{A}(S_{j,R},\eta-1)-\bra{0_{ALL}} \phi_{L, i}\ket{0_{ALL}} \cdot \bra{0_{ALL}} \phi_{j,R}\ket{0_{ALL}} \Bigg )\Bigg\|\cdot\\
		& \Bigg( \Bigg \|  \bra{0_{ALL}} \phi_{i,j}\ket{0_{ALL}} -  \sum_{\sigma\in\mathcal{P}(\{i+1,\dots,j-1\})\setminus\emptyset} (-1)^{|\sigma|+1} \bra{0_{ALL}} \left (\otimes_{k \in \sigma} \Pi^K_{F_k} \bra{0_{M_k}} \right )\phi_{i,j}\left (\otimes_{k \in \sigma} \ket{0_{M_k}}\Pi^K_{F_k}\right )\ket{0_{ALL}}\Bigg\| + 1\Bigg ) 
		\end{align*}
		
		We will now prove the bounds:  $G_1 \leq   2^{\Delta} (12\projerror) $, $G_2 \leq O\left (2^\Delta K \left (e(n)^{2T}+\epsilon_2 \right ) \right ) $, $G_3 \leq O( \epsilon)$, $G_4 \leq 8 (1 + (2e(n) + 2\projerror)^{\Delta -2})f(S,\eta-1,\Delta,\epsilon) \leq 16 \cdot f(S,\eta-1,\Delta,\epsilon)$.  The desired result follows from these bounds, so all that remains is to prove them, which we do below.

	We begin by bounding $G_1$.  For any fixed subset $\sigma\in\mathcal{P}(\{j-1, ..., i+1\})$ we know, by using two applications of Lemma \ref{clm:breaktoproduct} that:

		\begin{align*}
		& \Bigg \|\bra{0_{ALL}}\ket{\Psi_{\{i,j\}\cup \sigma}}\bra{\Psi_{\{i,j\}\cup \sigma}}\ket{0_{ALL}}\\
		&- \frac{1}{(\lambda^{i}_1\lambda^{j}_{1})^{4K+1}}\bra{0_{ALL}} \phi_{L, i}\ket{0_{ALL}} \cdot \bra{0_{ALL}} \left (\otimes_{k \in \sigma} \Pi^K_{F_k} \bra{0_{M_k}} \right )\phi_{i,j}\left (\otimes_{k \in \sigma} \ket{0_{M_k}}\Pi^K_{F_k}\right )\ket{0_{ALL}}\cdot \bra{0_{ALL}} \phi_{j,R}\ket{0_{ALL}} \Bigg \|\\
		& \leq \Bigg \|\ket{\Psi_{\{i,j\}\cup \sigma}}\bra{\Psi_{\{i,j\}\cup \sigma}}- \frac{1}{(\lambda^{i}_1\lambda^{j}_{1})^{4K+1}} \phi_{L, i}\otimes \left (\otimes_{k \in \sigma} \Pi^K_{F_k} \bra{0_{M_k}} \right )\phi_{i,j}\left (\otimes_{k \in \sigma} \ket{0_{M_k}}\Pi^K_{F_k}\right )\otimes \phi_{j,R} \Bigg \| \\
		& \leq 2 \cdot 6\projerror = 12\projerror
		\end{align*}
	
	This follows because we can use Lemma \ref{clm:breaktoproduct} to ``cut" the state $\ket{\Psi_{\{i,j\}\cup \sigma}}\bra{\Psi_{\{i,j\}\cup \sigma}}$ twice, once at cut $i$ and once at cut $j$, which produces the above product state, incurring error $2 \cdot 6f(n)$.  It follows that:
	
		\begin{align*}
		&G_1 \equiv \sum_{\sigma\in\mathcal{P}(\{i+1,\dots,j-1\})\setminus\emptyset} \Bigg \|\bra{0_{ALL}}\ket{\Psi_{\{i,j\}\cup \sigma}}\bra{\Psi_{\{i,j\}\cup \sigma}}\ket{0_{ALL}}\\
		&- \frac{1}{(\lambda^{i}_1\lambda^{j}_{1})^{4K+1}}\bra{0_{ALL}} \phi_{L, i}\ket{0_{ALL}} \cdot \bra{0_{ALL}} \left (\otimes_{k \in \sigma} \Pi^K_{F_k} \bra{0_{M_k}} \right )\phi_{i,j}\left (\otimes_{k \in \sigma} \ket{0_{M_k}}\Pi^K_{F_k}\right )\ket{0_{ALL}}\cdot \bra{0_{ALL}} \phi_{j,R}\ket{0_{ALL}} \Bigg \| \\
		&\leq 2^{\Delta} (12\projerror),
		\end{align*}
		
		as desired.

		For the next three bounds we will repeatedly use the fact that $(\lambda^{i}_{1})^{4K+1} = \Theta(1)= (\lambda^{j}_{1})^{4K+1}$, and thus, $\frac{1}{(\lambda^{i}_1\lambda^{j}_{1})^{4K+1}}  = \Theta(1)$.  The reason for this is that, we know, from the use of Lemma \ref{lem:highschmidtnew} in the error analysis of Algorithm \ref{alg:quasi-poly-subroutine}, that $(\lambda^{i}_{1})^{4K+1}, (\lambda^{j}_{1})^{4K+1} \geq 1 - O(e(n))$, where $e(n) \leq (1 - 2^{\frac{\log(\delta)}{\log^7(n)}}) = O(1/\log^4(n))$ (since $\delta > n^{-\log^2(n)}$ as verified in the check in the driver Algorithm \ref{alg:quasi-poly-driver}).  Since $K = O(\log^3(n))$, as specified in Algorithm \ref{alg:quasi-poly-subroutine}, it follows that  $(\lambda^{i}_{1})^{4K+1} = \Theta(1)= (\lambda^{j}_{1})^{4K+1}$.   
		
		    We now bound $G_2$.

		\begin{align*}
		&G_2 \equiv \sum_{\sigma\in\mathcal{P}(\{i+1,\dots,j-1\})\setminus\emptyset} \Bigg \|\left(\frac{1}{(\lambda^{i}_1\lambda^{j}_{1})^{4K+1}} - \frac{1}{(\kappa^{i}_{T, \epsilon_2}\kappa^{j}_{T, \epsilon_2})^{4K+1}}\right )\\
		&\cdot \mathcal{A}(S_{L,i},\eta-1) \cdot \mathcal{A}(S_{j,R},\eta-1) \cdot   \mathcal{B}\left ( \left (\otimes_{k \in \sigma} \Pi^K_{F_k} \bra{0_{M_k}} \right )\phi_{i,j}\left (\otimes_{k \in \sigma} \ket{0_{M_k}}\Pi^K_{F_k}\right ),\frac{\epsilon}{2^\Delta} \right)\Bigg \|\\
		& \leq \sum_{\sigma\in\mathcal{P}(\{i+1,\dots,j-1\})\setminus\emptyset} \Bigg \|\left(\frac{1}{(\lambda^{i}_1\lambda^{j}_{1})^{4K+1}} - \frac{1}{(\kappa^{i}_{T, \epsilon_2}\kappa^{j}_{T, \epsilon_2})^{4K+1}}\right ) \Bigg \| \\
		&= 2^{\Delta-2 } \Bigg \|\left(\frac{1}{(\lambda^{i}_1\lambda^{j}_{1})^{4K+1}} - \frac{1}{(\kappa^{i}_{T, \epsilon_2}\kappa^{j}_{T, \epsilon_2})^{4K+1}}\right ) \Bigg \| =  2^{\Delta-2 } \Bigg \|\left(\frac{(\lambda^{i}_1\lambda^{j}_{1})^{4K+1} -(\kappa^{i}_{T, \epsilon_2}\kappa^{j}_{T, \epsilon_2})^{4K+1} }{(\lambda^{i}_1\lambda^{j}_{1})^{4K+1}(\kappa^{i}_{T, \epsilon_2}\kappa^{j}_{T, \epsilon_2})^{4K+1}} \right ) \Bigg \|\\
		& = O(2^\Delta)  \Bigg \|(\lambda^{i}_1\lambda^{j}_{1})^{4K+1} -(\kappa^{i}_{T, \epsilon_2}\kappa^{j}_{T, \epsilon_2})^{4K+1} \Bigg \| = O\left (2^\Delta (4K+1) \left (|\lambda^{i}_1 - \kappa^{i}_{T, \epsilon_2}| + |\lambda^{j}_1 - \kappa^{j}_{T, \epsilon_2}|  \right )\right )\\
		& \leq O\left (2^\Delta K \left (\frac{e(n)^{2T}+\epsilon_2}{(\lambda_1^i)^{2T+1}} \right ) \right ) = O\left (2^\Delta K \left (e(n)^{2T}+\epsilon_2 \right ) \right )
		\end{align*}
		
		Where the first inequality follows because, by definition, $\mathcal{A}(S_{L,i},\eta-1) , \mathcal{A}(S_{j,R},\eta-1) ,  \mathcal{B}\left ( \left (\otimes_{k \in \sigma} \Pi^K_{F_k} \bra{0_{M_k}} \right )\phi_{i,j}\left (\otimes_{k \in \sigma} \ket{0_{M_k}}\Pi^K_{F_k}\right ),\frac{\epsilon}{2^\Delta} \right) = O(1)$ (since each is a close approximation of a quantum state amplitude squared, which is at most 1 by definition).  The remaining steps follow by using the fact that $(\lambda^{i}_{1})^{4K+1} = \Theta(1)= (\lambda^{j}_{1})^{4K+1}$ as discussed above (note that $(\lambda^{i}_{1})^{2T} = \Theta(1)$ for the same reason, since $T = O(\log^3(n))$), and by using Lemma \ref{lem:lambdaapprox} which gives the error bound for how well the $\kappa$ terms approximate the $\lambda$ terms.

		We now bound $G_3$:
		
		\begin{align*}
		&G_3 \equiv \epsilon \cdot \Bigg \|\frac{1}{(\lambda_1^i \lambda_1^j)^{4K+1}}\mathcal{A}(S_{L,i},\eta-1) \cdot \mathcal{A}(S_{j,R},\eta-1)\Bigg\| \leq O( \epsilon)
		\end{align*}
		
		Where we have used that $(\lambda^{i}_{1})^{4K+1} = \Theta(1)= (\lambda^{j}_{1})^{4K+1}$, and $\mathcal{A}(S_{L,i},\eta-1) , \mathcal{A}(S_{j,R},\eta-1)  = O(1)$, for the same reasons as in the bound of $G_2$.  
		
        We now bound $G_4$:

		\begin{align*}
		&G_4 \equiv\Bigg \|\frac{1}{(\lambda_1^i \lambda_1^j)^{4K+1}} \Bigg (  \mathcal{A}(S_{L,i},\eta-1) \cdot \mathcal{A}(S_{j,R},\eta-1)-\bra{0_{ALL}} \phi_{L, i}\ket{0_{ALL}} \cdot \bra{0_{ALL}} \phi_{j,R}\ket{0_{ALL}} \Bigg )\Bigg\|\cdot\\
		& \Bigg( \Bigg \|  \bra{0_{ALL}} \phi_{i,j}\ket{0_{ALL}} -  \sum_{\sigma\in\mathcal{P}(\{i+1,\dots,j-1\})\setminus\emptyset} (-1)^{|\sigma|+1} \bra{0_{ALL}} \left (\otimes_{k \in \sigma} \Pi^K_{F_k} \bra{0_{M_k}} \right )\phi_{i,j}\left (\otimes_{k \in \sigma} \ket{0_{M_k}}\Pi^K_{F_k}\right )\ket{0_{ALL}}\Bigg\| + 1\Bigg ) \\
		&  \leq 4 \Bigg \|   \mathcal{A}(S_{L,i},\eta-1) \cdot \mathcal{A}(S_{j,R},\eta-1)-\bra{0_{ALL}} \phi_{L, i}\ket{0_{ALL}} \cdot \bra{0_{ALL}} \phi_{j,R}\ket{0_{ALL}} \Bigg\|\cdot \\
		& \Bigg( (2e(n) + 2\projerror)^{\Delta-2} + 1\Bigg )\\
		&  \leq  8 \Bigg \|   \mathcal{A}(S_{L,i},\eta-1) \cdot \mathcal{A}(S_{j,R},\eta-1)-\bra{0_{ALL}} \phi_{L, i}\ket{0_{ALL}} \cdot \bra{0_{ALL}} \phi_{j,R}\ket{0_{ALL}} \Bigg\|\\
		& \leq 8 \cdot 2 \cdot  f(S,\eta-1,\Delta,\epsilon) = 16 f(S,\eta-1,\Delta,\epsilon)
		\end{align*}
	
	Here the first inequality follows by our previous argument that $\frac{1}{(\lambda_1^i \lambda_1^j)^{4K+1}} = \Theta(1)$, as well as Lemma \ref{clm:expansiontrick}.  (In fact, since we find it desirable to have an explicit constant for this particular error term, we are using $\frac{1}{(\lambda_1^i \lambda_1^j)^{4K+1}} \leq 4$, which the reader may verify, although we emphasize that the value of this constant does not matter for the asymptotic scaling and is only used for simplicity of presentation elsewhere in this paper.)  Note that our use of Lemma \ref{clm:expansiontrick}, while simple, was key here in order to avoid a factor of $2^\Delta$ appearing in the bound of $G_4$.  The second inequality follows because the bound $(2e(n) + 2\projerror)^{\Delta-2} =  o(1)$ is immediate (in fact, since $e(n), f(n) = o(1)$, and $\Delta = \Theta(\log(n))$, this quantity actually quite small, but here we only need that it is $o(1)$).  The final inequality follows by two uses of the definition of $f(S,\eta-1,\Delta,\epsilon)$, which, we recall, is defined, recursively, to be the error bound on $\mathcal{A}(\cdot,\eta-1)$, so that $f(S,\eta-1,\Delta,\epsilon) \geq |\mathcal{A}(S_{L,i},\eta-1) -\bra{0_{ALL}} \phi_{L, i}\ket{0_{ALL}} |$, and $f(S,\eta-1,\Delta,\epsilon) \geq |\mathcal{A}(S_{j,R},\eta-1)- \bra{0_{ALL}} \phi_{j,R}\ket{0_{ALL}} |$ by definition.  (This final step also uses the triangle inequality, and the facts that  $\mathcal{A}(S_{L,i},\eta-1) , \mathcal{A}(S_{j,R},\eta-1)  = O(1)$, etc).  
	
  Now that we have bounded $G_1, G_2, G_3$, and $G_4$, the proof is complete.

\end{proof}

\bibliography{cubebib}

\end{document}